\documentclass[12pt]{article}
\usepackage{amsmath,amsthm,amssymb}
\usepackage{graphicx}
\usepackage{titling}
\title{Synthesis of passive multidimensional scattering matrices: a survey}
\author{Sankar Basu\thanks{Any opinions, findings, conclusions, or recommendations expressed in this material are those of the author and do not necessarily reflect the views of the National Science Foundation.}\\Computer and Information Sciences and Engineering\\National Science Foundation\\Alexandria, Virginia 22314}

\date{ }
\numberwithin{equation}{section} 
\DeclareMathOperator{\rank}{rank}
\DeclareMathOperator{\diag}{diag}
\begin{document}
\begin{titlingpage}
    \maketitle
    \begin{abstract}
    The goal of this exposition is twofold. The first is to survey the status of synthesizability of rational multivariable ($n$-dimensional) passive linear shift invariant bounded (or positive) transfer functions as the scattering (or immittance) matrices of networks consisting of a finite number of $n$ different types of inductances/capacitances and memoryless reciprocal or nonreciprocal elements such as the transformers, gyrators etc. The theory thus attempts to extend the classical network synthesis from univariate (i.e., 1-D) case to the multidimensional ($n$-D) case. It is known that the main bottleneck in pursuing this goal of multivariable circuit synthesis is Hilbert's 17-th problem on expressibility of positive polynomials as a sum of squares, which is an active problem area in real algebraic geometry much studied for more than last hundred years.  The interplay between these two closely related topics, each important in their own right, is the main topic of this survey.  Synthesis of both lossless and lossy (i.e., dissipative) transfer functions are considered. While a 2-D lossless transfer function can be shown to be synthesizable, it turns out that the technique of embedding a dissipative transfer function within a lossless transfer function with a larger number of ports, which is equivalent to a certain matrix dilation problem, can be carried out by means of spectral factorization only partially in 2-D, but completely fails in higher dimensions (i.e., for $n>2$). The fundamental reason for this breakdown is the unavailability of an adequate sum of squares representation of multivariable positive polynomials. In higher dimensions (i.e., for $n>2$) even lossless transfer functions cannot in general be synthesized except for certain stable all-pass functions of low ``degree". Once again, this failure can be attributed to unavailability of sum of squares representation mentioned above. Failure of other 1-D synthesis procedures, e.g., via matrix factorizations of various types are also briefly discussed in the interest of a complete status report on the multidimensional synthesis problem. Depending on the convenience of exposition, our mode of presentation toggles between continuous time systems and discrete time systems. From a conceptual point of view this strategy is believed not to cause any loss of generality at least for the linear shift invariant systems under consideration.
    \end{abstract}
\end{titlingpage}
%
%
\newcommand{\R}{{\cal R}}
\newcommand{\C}{{\cal C}}
\newcommand{\scp}{{\cal S}}
\newcommand{\hp}{{\cal H}}
\newcommand{\z}{\mbox{\bf z}}
\newcommand{\p}{\mbox{\bf p}}
\newcommand{\0}{\mbox{\bf 0}}
\newcommand{\Ibold}{\mbox{\bf I}}
\newcommand{\ol}{\overline}
\newcommand{\ul}{\underline}
\newcommand{\n}{\underline n}
\newcommand{\m}{\underline m}
\newcommand{\be}{\begin{equation}}
\newcommand{\ee}{\end{equation}}
\newtheorem{definition}{Definition}[section]
\newtheorem{roole}{Rule}[section]
\newtheorem{example}{Example}[section]
\newtheorem{fact}{Fact}[section]
\newtheorem{lemma}{Lemma}[section]
\newtheorem{proposition}{Proposition}[section]
\newtheorem{property}{Property}[section]
\newtheorem{theorem}{Theorem}[section]
\newtheorem{remark}{Remark}[section]
\newtheorem{corollary}{Corollary}[section]
\newtheorem{claim}{Claim}[section]
 \newcommand{\WSH}{\mbox{$\cal WSH$}}
 \newcommand{\WSS}{\mbox{$\cal WSS$}}
 \newcommand{\SH}{\mbox{$\cal SH$}}
 \newcommand{\ScS}{\mbox{$\cal SS$}}
 \newcommand{\SPH}{\mbox{$\cal SPH$}}
 \newcommand{\SRS}{\mbox{$\cal SRS$}}
 \newcommand{\RH}{\mbox{$\cal RH$}}
 \newcommand{\RS}{\mbox{$\cal RS$}}
 \newcommand{\IH}{\mbox{$\cal IH$}}
 \newcommand{\IS}{\mbox{$\cal IS$}}
%
 \newcommand{\kpp}{{\bf \kappa}}
 \newcommand{\N}{{\bf N}}
 \newcommand{\x}{{\bf x}}
 \newcommand{\deter}{\mbox{\rm det}}
 \newcommand{\s}{\mbox{ ${\cal S}$ }}
 \newcommand{\bolddelta}{\mbox{\boldmath $\delta$}}
 \newcommand{\zt}{\mbox{$\zeta$}}

%
 \newcommand{\HH}{{\bf H}}
 \newcommand{\one}{{\bf 1}}
 \newcommand{\zz}{{\bf z}}
 \newcommand{\zero}{{\bf 0}}
 \newcommand{\pp}{\bf p}
 \newcommand{\PP}{{\bf P}}
 \newcommand{\DD}{{\bf D}}
 \newcommand{\MM}{{\bf M}}
 \newcommand{\y}{{\bf A}}
 \newcommand{\NN}{{\bf N}}
 \newcommand{\CC}{{\bf C}}
 \newcommand{\UU}{{\bf U}}
 \newcommand{\TT}{{\bf T}}
 \newcommand{\BB}{{\bf B}}
 \newcommand{\KK}{{\bf K}}
 \newcommand{\RR}{{\bf R}}
 \newcommand{\VV}{{\bf V}}
 \newcommand{\QQ}{{\bf Q}}
 \newcommand{\WW}{{\bf W}}
 \newcommand{\GG}{{\bf G}}
 \newcommand{\bchoose}{\left( \begin{array}{c}}
 \newcommand{\echoose}{\end{array} \right)}

 \newcommand{\Lambdabar}{\underline{\Lambda}}
 \newcommand{\pbar}{\underline{p}}
 \newcommand{\phibar}{\underline{\phi}}
 \newcommand{\Hbar}{\underline{H}}
 \newcommand{\omegabar}{\underline{\omega}}
 \newcommand{\Gbar}{\underline{G}}
 \newcommand{\Qbar}{\underline{Q}}
 \newcommand{\Abar}{\underline{A}}
 \newcommand{\onebar}{\underline{1}}
 \newcommand{\Vbar}{\underline{V}}
 \newcommand{\Ubar}{\underline{U}}
 \newcommand{\Pbar}{\underline{P}}
 \newcommand{\zerobar}{\underline{0}}
 \newcommand{\Rbar}{\underline{R}}
 \newcommand{\Dbar}{\underline{D}}
 \newcommand{\Wbar}{\underline{W}}
 \newcommand{\Lbar}{\underline{L}}
 \newcommand{\Sbar}{\underline{S}}
 \newcommand{\Fbar}{\underline{F}}
 \newcommand{\Bbar}{\underline{B}}
 \newcommand{\Ebar}{\underline{E}}
 \newcommand{\Kbar}{\underline{K}}
 \newcommand{\ubar}{\underline{u}}
 \newcommand{\ybar}{\underline{y}}
 \newcommand{\xbar}{\underline{x}}

\newcommand{\X}{\bf X}
\newcommand{\LL }{\bf L}
\newcommand{\A}{\bf A}
\newcommand{\B}{\bf B}
\newcommand{\D}{\bf D}
\newcommand{\E}{\bf E}
\newcommand{\F}{\bf F}
\newcommand{\V}{\bf V}
\newcommand{\K}{\bf K}
\newcommand{\G}{\bf G}
\newcommand{\je}{\bf J}
\newcommand{\I}{\bf I}
\newcommand{\Q}{\bf Q}
\newcommand{\uu}{\bf U}
\newcommand{\Sbold}{\bf S}
\newcommand{\boldO}{\bf O}

\newcommand{\bn}{\begin{enumerate}}
\newcommand{\en}{\end{enumerate}}

\newcommand{\ms}{\mbox{$1-\sqrt{3}$}}
\newcommand{\ps}{\mbox{$1+\sqrt{3}$}}
\newcommand{\sh}{ \mbox{$\sqrt{3}$}}
\newcommand{\st}{ \mbox{$\sqrt{2}$}}
\newcommand{\msn}{\mbox{$\sqrt{2-\sqrt{3}}$}}
\newcommand{\msp}{\mbox{$\sqrt{2+\sqrt{3}}$}}

\newcommand{\ch}{\cal H}
\newcommand{\cg}{\cal G}

\newcommand{\h}{\bf H}
\newcommand{\g}{\bf G}
\newcommand{\M}{\bf M}

\newcommand{\w}{\bf w}
\newcommand{\el}{\bf l}

\newcommand{\ccone}{\mbox{$(zI-A + BK)^{-1}$}}
\newcommand{\cctwo}{\mbox{$(zI-A + LC)^{-1}$}}

\newcommand{\bq}{\begin{equation}}
\newcommand{\eq}{\end{equation}}
\newcommand{\bea}{\begin{eqnarray}}
\newcommand{\eea}{\end{eqnarray}}
\newcommand{\ba}{\begin{array}}
\newcommand{\ea}{\end{array}}

\newcommand{\um}{\mbox{$\uparrow_{M}$}}
\newcommand{\dm}{\mbox{$\downarrow_{M}$}}
\newcommand{\gc}{\mbox{$\cal G$}}
\newcommand{\hc}{\mbox{$\cal H$}}
\newcommand{\de}{\mbox{$\delta$}}
\newcommand{\q}{\bf Q}
\newcommand{\e}{\bf E}
\newcommand{\U}{\bf U}
\newcommand{\T}{\bf T}

\newcommand{\ta}{\tilde a}
\newcommand{\tb}{\tilde b}
\newcommand{\tp}{\tilde{\bf P}}
\newcommand{\tq}{\tilde{\bf Q}}
\newcommand{\tr}{\tilde{\bf R}}
\newcommand{\tu}{\tilde{\bf U}}
\newcommand{\tT}{\tilde{\bf T}}

 \newcommand{\eord}{\stackrel{e}{<}}
 \newcommand{\oord}{\stackrel{o}{<}}
 \newcommand{\Scal}{{\cal S}}

\renewcommand{\Re}{\operatorname{Re}}
\renewcommand{\Im}{\operatorname{Im}}

\section{Introduction}
This introductory section is only intended to provide a brief listing of the various sections and subsections to follow, and to explain how they are tied together. More details on specific problems addressed, including information on background literature, are provided at the beginning of respective sections.

Section \ref{section-sos} starts with a leisurely introduction to the relevance of the sum of squares (SOS) problem to the 1-D spectral factorization problem. It then considers the SOS problem in 2-D and higher dimensions in subsection \ref{SOS-higher-dimension}. Here, some historical remarks beginning from the seminal work of Hilbert are made, and some references to the classical and modern literature relevant to our purpose are included. This is followed by a more detailed examination of the bivariate SOS problem, and its relevance to the scalar 2-D spectral factorization in subsection \ref{bivariate-SP-factor}. A matrix generalization of the spectral factorization problem is then undertaken and worked out in detail in subsection \ref{MatrixSPfactor}. While all this is considered in the continuous  domain,  Section \ref{section-sos} ends with  obvious analogs of the corresponding results in the discrete domain.

The synthesis of lossless 2-D bounded matrices forms the subject of entire Section \ref{scatt:S21}. A (global) synthesis procedure is essentially worked out in 3 steps elaborated in subsections \ref{first-step}-\ref{third-step}. In subsection \ref{minimal-synthesis} it is shown that such a synthesis is indeed minimal in some sense.

Corresponding problems in higher dimensions (i.e., $n$-D, $n>2$) are considered  in Section \ref{scatt:S24}. Here, the main result is demonstration of the fact that 3-D all-pass functions of ``low degree", as a scalar example of lossless bounded matrices, can indeed be synthesized. By drawing upon  counterexamples from the literature it is shown that synthesis of all-pass functions are infeasible, in general, particularly when the number $n$ of variables involved are larger than 3, or the prescribed transfer function is of higher degree. 

Section \ref{S22} demonstrates how by using the main spectral factorization result from Section \ref{section-sos}, which in fact fundamentally rests on Hilbert's results on the SOS problem,  a 2-D  bounded rational matrix can be embedded in a 2-D lossless bounded rational matrix. In circuit theoretic terms,  the problem of  synthesis of a prescribed lossy, or dissipative 2-D rational transfer function is then translated to the problem of synthesizing a 2-D lossless bounded rational matrix (i.e., the scattering matrix of a 2-D lossless circuit). Since, as demonstrated in Section \ref{scatt:S21}, this latter problem can be fully solved, synthesizability of a prescribed 2-D bounded rational matrix as the scattering matrix of a 2-D network  is thus established.

Passive state space realizations are considered in Section \ref{scatt:S23}. Discussion naturally leads to  2-D analog of the well-known bounded real lemma (or the Kalman-Popov-Yakubovitch lemma), and it is shown that  for 2-D, a weak form of the bounded real lemma can indeed be established. It is to be noted that all discussions leading up to this point fundamentally rests on issues surrounding the SOS problem originating in the work of Hilbert. For the sake of completeness, two other approaches to multidimensional synthesis are very briefly outlined in Section \ref{other-synthesis}. These, however, do not directly depend on the SOS problem. Finally, the main results are briefly summarized in Section \ref{conclude}.

In considering circuits and system theoretic problems, our presentation toggles between the continuous domain and the discrete domain depending on the convenience of exposition.  From a conceptual point of view, this strategy should not cause any loss of continuity at least for the linear shift invariant systems under consideration.

All circuit and systems theoretic notions are introduced and defined in the context, when it appears first. Obvious  notation are not explained, unless something nonstandard is used. However, at the very outset we introduce some notations commonly used in the abstract abstract algebra literature.

 The set of all complex numbers are to be denoted by $\mathbb{C}$. Correspondingly, the set of all rational functions in the variable
 \footnote{
Here $p$ denotes the Laplace transform variable; alternative notation $s$ is also used in system theoretic literature. In the bivariate context we have the variables $p_1$, $p_2$, etc. }
$p_1$ with complex  coefficients  will be denoted by $\mathbb{C}(p_1)$. Note that  $\mathbb{C}(p_1)$ forms an algebraic {\em field}.  The set of all polynomials in the variable $p_2$ with coefficients drawn from the field  $\mathbb{C}(p_1)$ will be denoted by
 $\mathbb{C}(p_1)[p_2]$. Such a set forms an algebraic {\em ring} of very special type, and it will be important for us to note that
 $\mathbb{C}(p_1)[p_2]$ is known as a principal ideal domain (PID), and in particular an {\em Euclidean domain} in which  Euclid's (pseudo) division algorithm holds.  We will also deal with matrices whose entries are
drawn from  $\mathbb{C}(p_1)[p_2]$. Strictly speaking, the set of all matrices of size $(m\times n)$ whose entries are from  $\mathbb{C}(p_1)[p_2]$ are to be denoted by  $\mathbb{C}^{(m\times n)}(p_1)[p_2]$. However, with some abuse of notation we shall drop the superscript $(m\times n)$ and use  the same notation to  denote the class of matrices whose entries belong to $\mathbb{C}(p_1)[p_2]$.   We will say a square matrix in $\mathbb{C}(p_1)[p_2]$ is $p_2\text{\em -unimodular}$ if its determinant is not identically zero, and is in $\mathbb{C}(p_1)$, i.e.,  independent of $p_2$. A matrix is simply said to be {\em unimodular} if its determinant is a nonzero constant.  The same notational convention would apply if the polynomials and/or rational functions involved have coefficients drawn from, instead of $\mathbb{C}$,  the set of real numbers $\mathbb{R}$, i.e., we will then use notations such as $\mathbb{R}(p_1)$ or  $\mathbb{R}(p_1)[p_2]$ etc. Likewise, if the transform variable is a Z-transform variable used in discrete domain considerations, then we may use notations, e.g., $\mathbb{C}(z_1)$, $\mathbb{C}(z_1)[z_2]$, or $\mathbb{R}(z_1)$, $\mathbb{R}(z_1)[z_2]$ as the case may be. Natural generalizations of these notations in the $n$-D ($n>2$) case, clear from the context, may also be used.

\section{Preliminaries on sum of squares} \label{section-sos}
We begin with a leisurely introduction of the necessary ingredients of our discussion in the $1$-D case.
Let $s(p)=b(p)/a(p)$ be a real rational bounded function of a single complex variable $p$, i.e., its numerator $b(p)$ and $a(p)$ are polynomials with \underline{real} coefficients and $|s(p)|<1$ in $\operatorname{Re}p>0$. The last condition can be equivalently replaced by 
\be\label{1.1}
1-{{s}_{*}}(p)s(p)>0\text{ for }p=j\omega ,\text{  }s(p)\text{ holomorphic in }\operatorname{Re}p>0.\ee 
The holomorphicity condition in (\ref{1.1}) requires $a(p)$ to be a stable polynomial, i.e., $a(p)\ne 0$ in $\operatorname{Re}p>0$. Clearly, we have 
\be\label{1.2}
1-s_{*}(p){s(p)=\frac{\phi (p)}{{{a}_{*}}(p)a(p)};\quad \phi (p)=a_{*}}(p)a(p)-b_{*}(p)b(p).\ee 	
Notice that since $\phi ({{p}_{0}})=0$ with some $\operatorname{Re}{{p}_{0}}>0$ implies $\phi (-{{p}_{0}})=0$ and vice versa, it is easy to see that the zeros of $\phi (p)$ away from the $j\omega $ axis form quadrantal symmetry and thus the following factorization of $\phi (p)$ holds
\be\label{1.3}
\phi (p)={{a}_{*}}(p)a(p)-{{b}_{*}}(p)b(p)={{c}_{*}}(p)c(p),\ee
where $c(p)$ is a real polynomial, thus establishing the so called spectral factorization  
\be\label{1.4}
1-{{s}_{*}}(p)s(p)={{r}_{*}}(p)r(p);\quad r(p)=\frac{c(p)}{a(p)}.\ee
Note further that the factorizations indicated in (\ref{1.3}) and (\ref{1.4}) are non-unique, and in particular,  $c(p)$ can be chosen to be a stable polynomial, i.e., $c(p)\ne 0$ in $\operatorname{Re}p>0$, thus endowing the rational spectral factor $r(p)$ the further property that it is not only stable but it also has a stable inverse. Such functions have been called the \underline{outer} functions in the literature.

It may be noted that the contractivity condition (\ref{1.1}) does, in fact, play a sublime but  critical role in the feasibility of the factorization (\ref{1.3}) or (\ref{1.4}), because it prevents the occurrence of factors of the type $({{p}^{2}}+{{\omega }^{2}})$ with odd multiplicity in $\phi (p)$, thus enforcing a quadrantal  symmetry on the zeros of $\phi (p)$ in the complex plane.

While the above spectral factorization result (\ref{1.4}) is easily seen via polynomial factorization, it can also be derived by using purely rational algebraic operations. This latter technique provides deeper insight into the feasibility of an analogous factorization for polynomials in two or more variables. 

In order to appreciate the developments to follow and set the mathematical environment against proper background, next we briefly digress to a well-known result on expressibility of univariate positive polynomials as a sum of squares and its adaptation to the present context. To this end, we have the following fact, examining the validity of which via a specific technique prove to be relevant to our purposes. 

\begin{fact}\label{SOS1D}
{\bf[sum of squares representation, 1D case]}\label{1DSOS} \mbox{ }\\
(a) If $f(x)$ is a real polynomial, positive for all real values of the variable $x$, then $f(x)$ can be expressed as
\be \label{1.5}
f(x)=q_{1}^{2}(x)+q_{2}^{2}(x),
\ee
where ${{q}_{1}}(x)$ and ${{q}_{2}}(x)$ are real polynomials in $x$. 
\\
(b) Furthermore, if $f(x)$ is an even polynomial in $x$ (i.e., $f(x)=f(-x)$, in other words, only powers of ${{x}^{2}}$ occur in the expression for $f(x)$), then it is possible to have ${{q}_{1}}(x)$ as an even polynomial, i.e., ${{q}_{1}}(x)={{q}_{1}}(-x)$, and ${{q}_{2}}(x)$ as an odd polynomial in $x$, i.e., ${{q}_{2}}(x)=-{{q}_{2}}(-x)$.
\end{fact}

Part (a) is well-known, and follows rather easily from the following arguments. First, $f(x)>0$ for all real $x$ implies, that the factors corresponding to real zeros of $f(x)$ must have even multiplicity, i.e., will occur as squares. Next observe that an irreducible (over reals) quadratic factor of $f(x)$ i.e., those corresponding to pairs of complex conjugate zeros, appear as sum of squares, e.g., ${{(x-\alpha )}^{2}}+{{\beta }^{2}}$, where $\alpha$  and $\beta$ are real numbers. Furthermore, the product of any two sum of squares can also be expressed as a sum of squares, e.g.,
\be\label{1.6}
({{a}^{2}}+{{b}^{2}})({{c}^{2}}+{{d}^{2}})={{(ac\pm bd)}^{2}}+{{(ad\mp bc)}^{2}}.
\ee
Part (a) thus follows from a combined application of the above mentioned facts in a straightforward manner.

Part (b) is usually not stated in the context of sum of squares representation, but is more relevant for system theoretic applications via its use in spectral factorization. To justify its validity, first observe that since $f(x)=f(-x)$, factors corresponding to real zeros of $f(x)$  not only come in even multiplicity but must occur in even powers of  even polynomials, i.e., of the form ${{({{x}^{2}}-{{\gamma }^{2}})}^{2}}$ , $\gamma$ real. Next, corresponding to every irreducible quadratic factor ${{(x-\alpha )}^{2}}+{{\beta }^{2}}$ of $f(x)$ mentioned above, ${{(x+\alpha )}^{2}}+{{\beta }^{2}}$ must also be an irreducible quadratic factor of $f(x)$. Now, setting $a=x-\alpha $, $b=\beta $, $c=x+\alpha $, and $d=\beta $  in (\ref{1.6}) with the choice of lower  signs one can write:
\be\label{1.7} 
\{{{(x-\alpha )}^{2}}+{{\beta }^{2}}\}\{{{(x+\alpha )}^{2}}+{{\beta }^{2}}\}={{({{x}^{2}}-{{\alpha }^{2}}-{{\beta }^{2}})}^{2}}+{{(2\beta x)}^{2}},
\ee
which is a sum of square of an even polynomial and an odd polynomial. Next, we make the observation that if on the left hand side of (\ref{1.6}) $a$, $c$ are both even polynomials, and $b$, $d$ are both odd polynomials then on the right hand side of (\ref{1.6}) $ac\pm bd$ is an even polynomial, whereas $ad\mp bc$ is an odd polynomial, i.e., in a sense the even/odd character of the sum of squares in the factors remains invariant in the product as well. Combining these facts, now it is easy to see that the products of all irreducible factors of $f(x)$ may be arranged in a form such that in (\ref{1.5}) we may have ${{q}_{1}}(x)={{q}_{1}}(-x)$, and ${{q}_{2}}(x)=-{{q}_{2}}(-x)$. \hfill{Q.E.D}

{\flushleft\bf Remark:} 
While the justification of the part (b) of the fact mentioned above is now complete, the above discussion also yields the non-uniqueness of the representation (\ref{1.5}). For example, the choice of lower signs in (\ref{1.6}) yields
\be\label{1..8}
 \{{{(x-\alpha )}^{2}}+{{\beta }^{2}}\}\{{{(x+\alpha )}^{2}}+{{\beta }^{2}}\}={{({{x}^{2}}-{{\alpha }^{2}}+{{\beta }^{2}})}^{2}}+{{(2\alpha \beta )}^{2}},
\ee 	
which is also a  sum of squares of two polynomials, both of which have the property of being even. 

We now return to application of the above general principles to spectral factorization. For this, we start by considering the function
\be\label{1.9}
\phi (j\omega )={{a}_{*}}(j\omega )a(j\omega )-{{b}_{*}}(j\omega )b(j\omega ),
\ee 
which, as a function of $\omega$,  is clearly a real polynomial, and is even, i.e., only monomials containing powers of  ${{\omega }^{2}}$ are present in it. Furthermore, it follows from (\ref{1.1}) and (\ref{1.2}) that $\phi (j\omega )=|a(j\omega ){{|}^{2}}-|b(j\omega ){{|}^{2}}>0$ for all real values of $\omega$, which after invoking both part (a) and part(b) of the above  discussion imply the existence of two further real even polynomials ${{q}_{1}}(\omega )$ and ${{q}_{2}}(\omega )$ such that 
\be\label{1.10}
\phi (j\omega )=q_{1}^{2}({{\omega }^{2}})+{{\omega }^{2}}q_{2}^{2}({{\omega }^{2}}).
\ee 
We next extend the polynomials in (\ref{1.10}) analytically to the complex plane by inserting $\omega=-jp $, i.e.,  $p=j\omega $, and we write
\be\label{1.11}
\phi (p)=q_{1}^{2}(-{{p}^{2}})-{{p}^{2}}q_{2}^{2}(-{{p}^{2}})=({{q}_{1}}(-{{p}^{2}})+p{{q}_{2}}(-{{p}^{2}}))({{q}_{1}}(-{{p}^{2}})-p{{q}_{2}}(-{{p}^{2}}).
\ee
If we define the real polynomial $c(p)$ and thus its para-conjugate ${{c}_{*}}(p)$ as 
\be\label{1.12} 
c(p)={{q}_{1}}(-{{p}^{2}})+p{{q}_{2}}(-{{p}^{2}});\quad {{c}_{*}}(p)={{q}_{1}}(-{{p}^{2}})-p{{q}_{2}}(-{{p}^{2}}),
\ee 
then  (\ref{1.11})  can also be written as
\be\label{1.13}
\phi (p)=c(p){{c}_{*}}(p),
\ee
which is indeed the desired spectral factorization  (\ref{1.3}) resulting in (\ref{1.4}).

{\flushleft\bf Remark:} While the critical step for the spectral factorability is the sum of squares representation (part (a)) as in (\ref{1.5}), the fact that the spectral factor $c(p)$ is a real polynomial is a consequence of part (b), i.e., the positive polynomial $\phi (j\omega )$ can be expressed as sum of squares of an  even polynomial and an odd polynomial as in (\ref{1.10}). In our multidimensional generalizations to be discussed later, this latter fact can prove to be a point of departure even under circumstances when a sum of squares representation as, e.g., in part (a) hold true.

\subsection{Generalizations to higher dimensions:}\label{SOS-higher-dimension}
The sum of squares representation of a positive polynomial in $n$-variables is the crux of Hilbert’s 17th problem. While the problem has a long history \cite{S_fact:lam,S_fact:marshall} leading to beginning of the field of real algebraic geometry, and its more recent use in nonconvex optimization, we highlight the main facts by restricting ourselves to the immediate needs of the present context.

The earliest example of a positive polynomial not expressible as a sum of squares is the well-known Motzkin polynomial \cite{S_fact:motz}
\be\label{1.14}
M(x,y)={{x}^{2}}{{y}^{4}}+{{x}^{4}}{{y}^{2}}+1-3{{x}^{2}}{{y}^{2}},
\ee
which is easily seen to be positive by using the arithmetic-geometric mean inequality $\frac{1}{3}({{x}_{1}}+{{x}_{2}}+{{x}_{3}})\ge \sqrt[3]{{{x}_{1}}{{x}_{2}}{{x}_{3}}}$ with ${{x}_{1}}=1$, ${{x}_{2}}={{x}^{2}}{{y}^{4}}$, and ${{x}_{2}}={{x}^{4}}{{y}^{2}}$. The proof of the fact that $M(x,y)$ cannot be expressed as a sum of squares of polynomials simply follows from a brute force attempt to  match the coefficients of the candidate polynomials, is widely known in open literature, and is thus not repeated here \cite{S_fact:marshall,contruct-algorithm,Reznick,S_fact:lam}. To wit, it may be noted that despite long history of the problem, concrete counterexamples of the above type are of relatively recent origin, and a well known counterexample due to Choi and Lam \cite{Reznick} that followed Motzkin's first counterexample was indeed inspired by circuit theorist Toshiro Koga's erroneous treatment \cite{rob:Koga1} of synthesizability of bivariate positive  real functions. 

Hilbert's 17th problem was solved by the so called Artin-Schrier theory of ordered fields \cite{S_fact:marshall}, the essential result of which is stated as follows. If $f(\mathbf{x})$ is a real positive polynomial in $n$-variables $\mathbf{x}=({{x}_{1}},{{x}_{2}},\cdots {{x}_{n}})$ then there exists a real polynomial $g(\mathbf{x})$ such that ${{g}^{2}}(\mathbf{x})f(\mathbf{x})$ can be expressed as a sum of squares of polynomials. The original result of Artin was an existence result, i.e., its proof was nonconstructive, and the minimum number of square terms needed in such an expression was not available from the Artin-Schrier theory either. Indeed, for the Motzkin polynomial $M(x,y)$ with $g(x,y)={{x}^{2}}+{{y}^{2}}$ one can have 
\be\label{1.15}
{{g}^{2}}(x,y)M(x,y)={{x}^{2}}{{y}^{2}}({{x}^{2}}+{{y}^{2}}+1){{({{x}^{2}}+{{y}^{2}}-2)}^{2}}+{{({{x}^{2}}-{{y}^{2}})}^{2}},
\ee
which is clearly a sum of squares of four polynomials. Thus, the Artin-Schrier result says that real positive polynomials can be expressed as sum of squares of real rational functions. The minimum number $\pi (n)$ of such square terms needed was later provided by the so-called theory of Pfister forms \cite{S_fact:lam,S_fact:marshall}. The essential result of relevance to us is that 
\be\label{1.16}
n+1\le \pi (n)\le {{2}^{n}}.
\ee 
Thus, in the univariate case, i.e., when $n=1$ we have the classical result $\pi (1)=2$. In the bivariate case i.e., when $n=2$ from (\ref{1.16}) we have $\pi (2)\le 4$. However, there exist bivariate polynomials, $M(x,y)$ being one such example, which cannot be expressed as sum of $3$ squares. Thus, for $n=2$ we have $\pi (2)=4$ in general. Since we will be mostly concerned with the bivariate case $n=2$, this latter fact will  play a significant role in further considerations. 

\subsubsection{Computation:}
While Artin's result, cast in the theory of ordered fields is nonconstructive, a constructive proof leading to an algorithm can indeed be given  by using a combination of various elementary techniques buried in the literature. To this end, one can cite \cite{contruct-algorithm} in the circuit theoretic context, but more modern and elaborate treatments have subsequently become available in the mathematical literature. Thus, the problem is solvable in a finite number of steps, and is thus ‘decidable’ in the sense of mathematical logic (\`a la Tarski, Seidenberg et. al.). However,  the algorithm in \cite{contruct-algorithm}, which can in principle be implemented via computations using symbolic algebra, is bound to be unmanageably large for problems of even reasonable size. It may be noted that although an elementary and constructive proof of Artin's result is available in this way, the algorithm does not produce, nor does it justify, the minimum number of factors $\pi (n)$ as given in (\ref{1.16}). An elementary proof of this latter fact, suitable for the current purpose, is unknown to us at present even for the case $n=2$.

Relevance to non-convex optimization has more recently triggered interest in algorithms having more numerical flavor in the computing community \cite{Parrilo}. Here one seeks ‘certificates’ of positivity of a multi-variable polynomial $f(\mathbf{x})$, by seeking a positive semidefinite matrix $\mathbf{A}$ such that $f(\mathbf{x})={{\mathbf{X}}^{T}}\mathbf{AX}$, where $\mathbf{X}$ is a vector of all relevant monomials. The square root operation $\mathbf{A}={{\mathbf{B}}^{T}}\mathbf{B}$ then yields the sum of squares representation$f(\mathbf{x})=|\mathbf{BX}{{|}^{2}}$ in terms of $r$ sum of squares with $r=\text{rank}(\mathbf{A})$, whereby the problem reduces to that of solving a semidefinite program. The size of the vector $\mathbf{X}$ clearly grows rapidly with the number of variables $n$ and the degree $d$ of the polynomial $f(\mathbf{x})$, and such algorithms, although have a numerical flavor, are known to be NP-hard \`{a} la computational complexity theory.

{\flushleft\bf The cones ${{\mathbf{\Sigma }}_{n,d}}$ and ${{\mathbf{\Pi }}_{n,d}}$:} 
It is easy to see that the set of all real positive polynomials  ${{\mathbf{\Pi }}_{n,d}}$ in $n$ variables with degree\footnote{%
Total degree of  a monomial is the sum of its partial degrees in each variable. Total degree of a polynomial is the largest of the total degrees of its monomials. By degree of a polynomial without further qualification, we always refer to its total degree.
}
 $d$,  as well as the set of all real polynomials  ${{\mathbf{\Sigma }}_{n,d}}$ in $n$ variables with degree $d$ admitting sum of squares representation form cones in the vector space of all polynomials in $n$ variables with degree $d$. In view of the preceding discussions it is clear that ${{\mathbf{\Sigma }}_{n,d}}\subseteq {{\mathbf{\Pi }}_{n,d}}$ in general, and ${{\mathbf{\Sigma }}_{n,d}}\subset {{\mathbf{\Pi }}_{n,d}}$ for most values of $n$ and $d$. A summary of such results for different values of $n$  and $d$ is displayed in the following Table, in which ``yes" indicates that a real polynomial can be expressed as a sum of squares of real polynomials, and ``no" indicates otherwise. Note the entry $n=2$, $d=4$ shows that all bivariate quartic polynomials can be expressed as sum  of (three) squares of polynomials, as was originally proved by Hilbert \cite{S_fact:hilb} as a part of his broader study of the problem. 
\begin{table}
\centering
  \begin{tabular}{| | c | c | c |  c | c | c || }
\hline\hline
\mbox{ } & $d=2$ & $d=4$ & $d=6$ & $d=8$ & $\cdots$ \\\hline
$n=1$    &   Yes    &   Yes    &   Yes    &   Yes    & $\cdots$ \\\hline
$n=2$    &   Yes    &   {\bf Yes}    &   No      &   No     & $\cdots$ \\\hline
$n=3$    &   Yes    &   No     &   No      &   No     & $\cdots$ \\\hline
$n=4$    &   Yes    &   No     &   No      &   No     & $\cdots$ \\\hline
\vdots & \vdots & \vdots & \vdots  & \vdots & $\ddots$ \\\hline
\hline 
\end{tabular}
\caption{Feasibility of SOS representation of polynomials in $n$ variables with total degree $d$. Note $d=4$, $n=2$  correspond to ternary quartic forms \cite{S_fact:hilb}.}
\label{table_SOS}
\end{table}
The Motzkin polynomial $M(x,y)$ clearly belongs to the  ``gap"  ${{\mathbf{\Sigma }}_{2,6}}\backslash {{\mathbf{\Pi }}_{2,6}}$ between  ${{\mathbf{\Sigma }}_{2,6}}$ and ${{\mathbf{\Pi }}_{2,6}}$, and it is nontrivial to generate examples of this type. Characterization of the nonempty gap between the cones ${{\mathbf{\Sigma }}_{n,d}}$ and ${{\mathbf{\Pi }}_{n,d}}$ for relevant values of $n$ and $d$ as indicated in Table \ref{table_SOS} (indeed, $M(x,y)\in {{\mathbf{\Sigma }}_{2,6}}\backslash {{\mathbf{\Pi }}_{2,6}}$ serves as an element of this gap), has been the object of much study, which we shall not delve into. 

It may however be interesting to note that the probability that an arbitrarily chosen polynomial from ${{\mathbf{\Sigma }}_{n,d}}$ also belongs to ${{\mathbf{\Pi }}_{n,d}}$ becomes vanishingly small (i.e., the relative size of the set ${{\mathbf{\Sigma }}_{2,6}}\backslash{{\mathbf{\Pi }}_{2,6}}$ tends to become larger) as $n$ tends to infinity \cite{Blekherman}. On the other hand, it has also been proven that  for a fixed $n$, the set ${{\mathbf{\Sigma}}_{n,d}}$ approaches  ${{\mathbf{\Pi}}_{n,d}}$ as $d$ tends to infinity \cite{Lasserre}.

\subsubsection{Bivariate case $n=2$:}
For reasons that will be obvious in course of the discussion, the bivariate case $n=2$ is of more interest to us for applications in circuits and systems theory. To this end, we first summarize aspects of the above discussion relevant to our purpose.

\begin{fact}
{\bf[SOS representation, bivariate case]}\label{SOS4S} \mbox{ }\\ 
Any real positive bivariate polynomial $f(x,y)$can be expressed as the sum of 4 squares of real rational functions, i.e., 
\be\label{1.17}
d_{0}^{2}(x,y)f(x,y)=\sum\limits_{i=1}^{4}{d_{i}^{2}(x,y)},
\ee 	
where ${{d}_{i}}(x,y)$, $i=0\text{ to }4$ are real polynomials in $x$ and $y$. 
\end{fact}
In order to adapt the latter result to system theoretic context we need to carry out another step known as the {\em Cassel reduction} in the mathematics literature \cite{S_fact:lam}. For this, we treat each ${{d}_{i}}(x,y)$ as a real polynomial in $y$ with coefficients as real rational functions of $x$, i.e., a member of $\mathbb{R}(p_{1})[p_{2}]$ and divide ${{d}_{i}}(x,y)$, $i=1\text{ to }4$ by ${{d}_{0}}(x,y)$, which produces
\be\label{1.18}
{{d}_{i}}(x,y)={{q}_{i}}(x,y){{d}_{0}}(x,y)+{{r}_{i}}(x,y),\quad {{\deg }_{y}}{{r}_{i}}(x,y)<{{\deg }_{y}}{{d}_{0}}(x,y)
\ee
in which ${{q}_{i}}(x,y)$, ${{r}_{i}}(x,y)$ for $i=1\text{ to }4$ are each in  $\mathbb{R}(p_{1})[p_{2}]$. In what follows for notational convenience we drop explicit reference to the variables $(x,y)$, whenever possible,  for compactness we define ${{\mathbf{q}}_{0}}=1$, ${{\mathbf{r}}_{0}}=0$, the vectors $\mathbf{d}={{({{d}_{0}},{{d}_{1}},\cdots {{d}_{4}})}^{T}}$, $\mathbf{q}={{({{q}_{0}},{{q}_{1}},\cdots {{q}_{4}})}^{T}}$, and the associated bilinear form $B(\mathbf{d},\mathbf{q})$ as 
\be\label{1.19}
B(\mathbf{d},\mathbf{q})={{\mathbf{d}}^{T}}\text{diag }\!\![\!\!\text{ }-f,1,\cdots ,1]\mathbf{q}.
\ee
This makes it possible to write (\ref{1.17}) in the compact form $B(\mathbf{d},\mathbf{d})=0$. Furthermore, if we define the vector $\mathbf{h}={{({{h}_{0}},{{h}_{1}},\cdots {{h}_{4}})}^{T}}$ as
\be\label{1.20}
\mathbf{h}=\alpha \mathbf{d}+\beta \mathbf{q},\text{ where }\alpha =B(\mathbf{q},\mathbf{q}),\quad \beta =-2B(\mathbf{d},\mathbf{q}),
\ee
then straightforward manipulations yields $B(\mathbf{h},\mathbf{h})=0$, which in more explicit form can be  written as
\be\label{1.21}
h_{0}^{2}(x,y)f(x,y)=\sum\limits_{i=1}^{4}{h_{i}^{2}(x,y)}.
\ee
Further algebraic manipulations with the stated equations produce 
\be\label{1.22}
{{h}_{0}}(x,y){{d}_{0}}(x,y)=\sum\limits_{i=1}^{4}{r_{i}^{2}(x,y)},
\ee
from which we conclude 
\be\label{1.23}
\begin{split}
  {{\deg }_{y}}{{h}_{0}}(x,y)+{{\deg }_{y}}{{d}_{0}}(x,y) & \leq 2{{\max }_{i}}({{\deg }_{y}}{{r}_{i}}(x,y)),\text{ } \\ 
 \text{i.e., }{{\deg }_{y}}{{h}_{0}}(x,y)                                & <     {{\deg }_{y}}{{d}_{0}}(x,y). 
\end{split}
\ee
The last inequality in (\ref{1.23}) follows from the inequality in (\ref{1.18}). Note that the representation in (\ref{1.21}) is akin to that in (\ref{1.17}) , but with $\deg_y h_0 < \deg _i d_0$. Thus, by repetitive application of the above procedure the degree of the denominator in (\ref{1.17}) or (\ref{1.21}), treated as a polynomial in the variable $y$, can be reduced to zero. Noting that ${{h}_{i}}(x,y)$ for $i=1\text{ to }4$ are polynomials in $y$,  the coefficients of which are real rational functions of $x$, we have then established the following fact.

\begin{fact}{\bf[bivariate SOS representation, univariate denominator]}\label{2DSOS1Ddeno} \mbox{ }\\
Any real positive bivariate polynomial $f(x,y)$ can be expressed 
\be\label{1.24}
s_{0}^{2}(x)f(x,y)=\sum\limits_{i=1}^{4}{s_{i}^{2}(x,y)},
\ee 
where ${{s}_{0}}(x)$ is a real polynomial in $x$ only, ${{s}_{i}}(x,y)$, $i=1\text{ to }4$ are real bivariate polynomials in $x$ and $y$. Equivalently, $f(x,y)$ can also be expressed as 
\be\label{1.25}
f(x,y)=\sum\limits_{i=1}^{4}{t_{i}^{2}(x,y)},
\ee 	
where each ${{t}_{i}}(x,y)$, $i=1\text{ to }4$ is a polynomial in $y$, the coefficients of which are real rational functions of $x$, i.e., $t_i \in \mathbb{R}(p_{1})[p_{2}]$.
\end{fact}
{\flushleft\bf Remark:} It should be clear from the Cassel reduction procedure outlined above that although (\ref{1.25}) is presented in the bivariate case, an $n$-variable counterpart holds true as well. This generalized version states that a real positive polynomial $f(\mathbf{x})$ in the variables $\mathbf{x}=({{x}_{1}},{{x}_{2}},\cdots {{x}_{n}})$ can be expressed as the sum of at most ${{2}^{n}}$ squares of, say, ${{t}_{i}}({{x}_{n}})$, each of which are polynomials in ${{x}_{n}}$, the coefficients of which are real rational functions of the remaining variables $({{x}_{1}},{{x}_{2}},\cdots {{x}_{n-1}})$. However, this more general version does not appear to be of much use for system theoretic considerations, and will not be discussed any further.

\subsection{Bivariate ($n=2$) scalar spectral factorization:}\label{bivariate-SP-factor}
Returning to application in 2-D spectral factorization, we now have a real irreducible bounded rational function $s({{p}_{1}},{{p}_{2}})$ as in (\ref{1.26a}), in which  $a({{p}_{1}},{{p}_{2}})$ and  $b({{p}_{1}},{{p}_{2}})$ are real bivariate polynomials in ${{p}_{1}}$ and ${{p}_{2}}$. Specifically, we have $1-s({{p}_{1}},{{p}_{2}}){{s}_{*}}({{p}_{1}},{{p}_{2}})>0$ for ${{p}_{i}}=j{{\omega }_{i}}$, $i=1,2$,  and $s({{p}_{1}},{{p}_{2}})$ holomorphic in $\operatorname{Re}{{p}_{1}}>0$, $\operatorname{Re}{{p}_{2}}>0$. The bounded property of $s({{p}_{1}},{{p}_{2}})$ requires $a({{p}_{1}},{{p}_{2}})$ to be a scattering Hurwitz polynomial \cite{stb:Fett2} (i.e., $a$ and ${{a}_{*}}$ be coprime, and $a({{p}_{1}},{{p}_{2}})\ne 0$ in $\operatorname{Re}{{p}_{1}}>0$, $\operatorname{Re}{{p}_{2}}>0$). The 2-D counterpart of equation (\ref{1.2}) can then be written as  (\ref{1.26b}) and (\ref{1.26c}) to follow.
\begin{subequations}
\begin{gather}
  s({{p}_{1}},{{p}_{2}})                                                       =\frac{b({{p}_{1}},{{p}_{2}})}{a({{p}_{1}},{{p}_{2}})}, \label{1.26a}\\ 
 1-s_* ({{p}_{1}},{{p}_{2}}){{s}}({{p}_{1}},{{p}_{2}})  = \frac{\phi ({{p}_{1}},{{p}_{2}})}{a_* ({{p}_{1}},{{p}_{2}}){{a}}({{p}_{1}},{{p}_{2}})}, \label{1.26b}\\ 
 \phi ({{p}_{1}},{{p}_{2}})                                                   = a_* ({{p}_{1}},{{p}_{2}}){{a}}({{p}_{1}},{{p}_{2}})-b_* ({{p}_{1}},{{p}_{2}}){{b}}({{p}_{1}},{{p}_{2}}) \label{1.26c}. 
\end{gather}
\end{subequations}
To proceed further we note that $\phi ({{p}_{1}},{{p}_{2}})$ is a real polynomial with $\phi ({{p}_{1}},{{p}_{2}})={{\phi }_{*}}({{p}_{1}},{{p}_{2}})$, i.e., each monomial in $\phi ({{p}_{1}},{{p}_{2}})$ is of even (total) degree (odd partial degree in one of the variables ${{p}_{2}}$ or ${{p}_{2}}$ is not excluded). Thus, $\phi (j{{\omega }_{1}},j{{\omega }_{2}})$ is a real polynomial in $\omega_1$ and $\omega_2$ with its monomials having even (total) degree. Furthermore, since $1-s({{p}_{1}},{{p}_{2}}){{s}_{*}}({{p}_{1}},{{p}_{2}})>0$ for ${{p}_{i}}=j{{\omega }_{i}}$, $i=1,2$ we have
\be\label{1.27}
\phi (j{{\omega }_{1}},j{{\omega }_{2}})=|a(j{{\omega }_{1}},j{{\omega }_{2}}){{|}^{2}}-|b(j{{\omega }_{1}},j{{\omega }_{2}}){{|}^{2}}>0.
\ee 	
Since $\phi (j{{\omega }_{1}},j{{\omega }_{2}})$is a real positive polynomial, a decomposition of the type (\ref{1.24}) exists, and thus for some univariate real polynomial ${{s}_{0}}({{\omega }_{1}})$, and bivariate real polynomials ${{s}_{i}}({{\omega }_{1}},{{\omega }_{2}})$, $i=1\text{ to }4$ one can write\footnote{%
Although because of the specific manner in which we have constructed it, $\phi$ has real coefficients,  the following arguments holds under the more general condition $\phi=\phi_*$, and $\phi(j\omega) >0$ for real all $\omega$.
} 
\be\label{1.28}
s_{0}^{2}({{\omega }_{1}})\phi (j{{\omega }_{1}},j{{\omega }_{2}})=\sum\limits_{i=1}^{4}{s_{i}^{2}({{\omega }_{1}},{{\omega }_{2}})}.
\ee
We note that ${{s}_{1}}({{\omega }_{10}})\ne 0$ for any real ${{\omega }_{10}}$, because otherwise from (\ref{1.28}) we would have ${{s}_{i}}({{\omega }_{10}},{{\omega }_{2}})=0$ for all real $\omega_2$, implying in view of Bezout’s theorem that $({{s}_{1}}-{{\omega }_{10}})$ is a factor of ${{s}_{i}}({{\omega }_{1}},{{\omega }_{2}})$. Since this consideration holds for all $i=1$ to $4$, it would then be possible to cancel the factor ${{({{s}_{1}}-{{\omega }_{10}})}^{2}}$ from both sides of (1.28). 
As in the 1-D case, we now consider the substitutions ${{\omega }_{1}}=-j{{p}_{1}}$, ${{\omega }_{2}}=-j{{p}_{2}}$  and define 
\be\label{1.29}
\begin{split}
   {{g}_{1}}({{p}_{1}},{{p}_{2}})   &={{s}_{1}}(-j{{p}_{1}},-j{{p}_{2}})+j{{s}_{2}}(-j{{p}_{1}},-j{{p}_{2}}), \\ 
 {{g}_{2}}({{p}_{1}},{{p}_{2}})     &={{s}_{3}}(-j{{p}_{1}},-j{{p}_{2}})+j{{s}_{4}}(-j{{p}_{1}},-j{{p}_{2}}). 
\end{split}
\ee
Noting that ${{s}_{1}}$, ${{s}_{2}}$ are real polynomials it trivially follows that
\be\label{1.30}
\begin{split}
  {{g}_{1*}}({{p}_{1}},{{p}_{2}})    &={{s}_{1}}(-j{{p}_{1}},-j{{p}_{2}})-j{{s}_{2}}(-j{{p}_{1}},-j{{p}_{2}}), \\ 
  {{g}_{2*}}({{p}_{1}},{{p}_{2}})    &={{s}_{3}}(-j{{p}_{1}},-j{{p}_{2}})-j{{s}_{4}}(-j{{p}_{1}},-j{{p}_{2}}), 
\end{split}
\ee
whence we obtain from (\ref{1.28})
\begin{subequations}\label{1.31}
\be\label{1.31a}
u({{p}_{1}})\phi ({{p}_{1}},{{p}_{2}})=\boldsymbol{G}_{*}(p_{1},p_{2}) \boldsymbol{G}(p_{1},p_{2}), 
\ee
where
\be\label{1.31b}
u({{p}_{1}})=s_{0}^{2}(-j{{p}_{1}}), \quad 
\boldsymbol{G}(p_{1},p_{2})={{\left[ \begin{matrix} {{g}_{1}}({{p}_{1}},{{p}_{2}}) & {{g}_{2}} \end{matrix}({{p}_{1}},{{p}_{2}}) \right]}^{T}}. 
\ee
\end{subequations}
By considering the para-conjugate of  (\ref{1.31a}) and recalling that  $\phi=\phi_*$, it follows  that  $u=u_*$. Since we have seen that $u(j\omega_{10})=s_{0}^{2}(j\omega_{10})\neq 0$ for any real $\omega_{10}$, the zeros of  $u(p_1 )$  form quadrantal symmetry in the complex ${{p}_{1}}$-plane. Therefore, the following factorization holds
\be\label{1.32}
u({{p}_{1}})={{d}_{*}}({{p}_{1}})d({{p}_{1}}).
\ee
It is easy to ensure that in such a factorization the polynomial $d({{p}_{1}})$ is a Hurwitz polynomial, i,e.  $d({{p}_{1}})\ne 0$ in $\operatorname{Re}{{p}_{1}}>0$. Consequently, it follows from (\ref{1.31}) and (\ref{1.32}) that
\be\label{1.33}
\phi ({{p}_{1}},{{p}_{2}})={\boldsymbol{H}_{*}}({{p}_{1}},{{p}_{2}})
\boldsymbol{H}({{p}_{1}},{{p}_{2}});\quad \boldsymbol{H}({{p}_{1}},{{p}_{2}})=\frac{\boldsymbol{G}({{p}_{1}},{{p}_{2}})}{d({p}_{1} )},
\ee 
in which $\boldsymbol{H}({{p}_{1}},{{p}_{2}})$  is holomorphic in $\operatorname{Re}{{p}_{1}}>0$, $\operatorname{Re}{{p}_{2}}>0$.

{\flushleft\bf Remark:}
Although ${{s}_{1}}$ and ${{s}_{2}}$ are real polynomials, there is no guarantee that ${{g}_{1}}$ and ${{g}_{2}}$ are real polynomials, and thus $\boldsymbol{H}({{p}_{1}},{{p}_{2}})$ has real coefficients. One way to ensure this would be to force ${{s}_{1}}$ even and ${{s}_{2}}$ odd polynomials respectively (i.e., in the multivariable case, this would imply that the  monomials in them have total degrees that are respectively only even or only odd). Indeed, in the 1-D case this property was achieved via the use of part (b) of Fact \ref{SOS1D}.
 \renewcommand{\be}{\begin{equation}}
 \renewcommand{\ee}{\end{equation}}
 \renewcommand{\bchoose}{\left( \begin{array}{c}}
 \renewcommand{\echoose}{\end{array} \right)}
%
%

Thus, along with the results of \cite{contruct-algorithm} we have  provided a {\em constructive} proof of the following theorem.

\begin{theorem}\label{S_fact:T*}
 Let $\phi=\phi({\bf p})$ be a polynomial in two variables ${\bf p}=(p_{1}, p_{2})$ such that $\phi=\phi_*$, $\phi({\bf p})
 \geq 0$ for all ${\bf p}=j\mbox{\boldmath $\omega$}$, $\boldsymbol\omega$ real.  Then there exists a
vector ${\bf H} \in \mathbb{C}(p_{1})[p_{2}]$ of size  $(2 \times 1)$, holomorphic in $\Re {\bf p}>0$, 
 such that $\phi = {\bf  H_{*}}{\bf H}$.   
\end{theorem}

{Proof:}
The discussions preceding the statement of the theorem essentially provides the proof  when $\phi$ has real coefficients. However, as already remarked, all augments carry over under the slightly more general condition $\phi=\phi_*$. \hfill{Q.E.D}  

 Our discussions of the last sections also yield the following result which will be useful for treatment of the matrix version of the spectral factorization problem. In particular, it will be needed in the developments leading up to proof of the matrix version of Theorem \ref{S_fact:T*}.
 \begin{corollary}\label{C1}
 For every scalar $\lambda \in \mathbb{C}(p_{1})[p_{2}]$ with $\lambda \not\equiv 0, \lambda_{*} = \lambda$,
 and $\lambda(j\mbox{\boldmath $\omega$})\geq 0$  wherever $\lambda (j\mbox{\boldmath $\omega$})$ is
 holomorphic, there exists a matrix $\boldsymbol\Lambda ({\bf p}) \in\mathbb{C}(p_{1})[p_{2}]$ 
 of size $(2\times 2)$, with $\det \boldsymbol\Lambda \not\equiv 0$,  $\boldsymbol\Lambda$ holomorphic in $\Re {\bf p}> 0$
such that 
\be\label{inflated-corollary}
\boldsymbol\Lambda_* \boldsymbol\Lambda = \boldsymbol\Lambda \boldsymbol\Lambda_* = \lambda\one_{2},
\ee 
where $\one_{2}$ is the identity matrix of order two.
 \end{corollary}

 {\em Proof:}  Obviously, $\lambda$ can be written as
 \begin{equation}\label{S_fact:77}
 \lambda = \frac{\phi}{w_{*}w},
 \end{equation}
 where $w = w(p_{1})$ and $\phi = \phi({\bf p})$ are polynomials.
 Furthermore, from the given properties of $\lambda$ and the relation
 $\phi = \lambda w_{*}w$, from which it follows that $\phi(j\mbox{\boldmath
 $\omega$}) = \lambda (j\mbox{\boldmath $\omega$})|w(j\omega_{1})|^{2}$, we conclude
 $\phi_{*} = \phi$ and $\phi(j\mbox{\boldmath $\omega$}) \geq 0$ for all real
 $\mbox{\boldmath $\omega$}$.  Due to Theorem \ref{S_fact:T*}, $\phi$ can be
 decomposed as $\phi = {\bf H}_{*}{\bf H}$. In particular, let ${\bf H}=[h_1, h_2]^T$, i.e., $h_{i} \in \mathbb{C}(p_{1})[p_{2}]$,
 for $i = 1,2$ be the two elements of ${\bf H}$.  We write $w_* w = d_* d$, where $d$ is a Hurwitz polynomial. Then, the matrix
 \begin{equation}\label{S_fact:78}
 {\boldsymbol\Lambda} = \frac{1}{d}
 \left[ \begin{array}{cc}
 h_{1} & -h_{2*} \\
 h_{2} &  \  h_{1*}
 \end{array} \right]
 \end{equation}
 obviously satisfies (\ref{inflated-corollary}), with elements in $\mathbb{C}(p_{1})[p_{2}]$,
because $w, d \in \mathbb{C}(p_{1})$, $\boldsymbol{H}$ is holomorphic in $\Re {\bf p} >0$, and $d\neq 0$ in $\Re p_1 >0$.\hfill{Q.E.D}

Note that for a given $\lambda$ satisfying the conditions of Corollary \ref{C1}, the matrix $\boldsymbol\Lambda$ satisfying  
(\ref{inflated-corollary}) is highly non-unique. This non-uniqueness will be further exploited in course of our development to sharpen the last mentioned result in a nontrivial manner (cf. Fact \ref{lambda-sequence} to follow).

 \subsection{Matricial $2$-D spectral factorization}\label{MatrixSPfactor}
We now turn to the matrix case.  Let $\boldsymbol\Phi$ be a polynomial matrix of size $(m \times m)$
 in two variables ${\bf p}=(p_{1},p_{2})$, which  is para-Hermitian, i.e.,
 $\boldsymbol\Phi_{*} = \boldsymbol\Phi$, and nonnegative
 definite (i.e., $\boldsymbol\Phi (j\mbox{\boldmath $\omega$}) \geq 0$ for all ${\bf p} = j\mbox{\boldmath $\omega$}$ ,
$\boldsymbol\omega$ real. 
We wish to show that a factorization of the form
 \begin{equation} \label{S_fact:72}
 \boldsymbol\Phi ({\bf p}) = {\bf H}_{*}({\bf p}){\bf H}({\bf p})
 \end{equation}
 is feasible, where ${\bf H}$ is a $2r \times  m$-matrix which is polynomial
 in $p_{2}$, but rational in $p_{1}$,  where $r = \rank \boldsymbol\Phi ({\bf p})$.  Furthermore, 
 ${\bf H}$ is required to be holomorphic for
 ${\bf p} = (p_{1},p_{2})$, with Re $p_{1} > 0$.  
 \subsubsection{Preliminaries}
We will need to use the following  result regarding matrices with entries belonging to 
$\mathbb{C}(p_1)[p_2]$.
\begin{fact}\label{smith}
Any $(m\times m)$ matrix $\boldsymbol\Phi \in \mathbb{C}(p_1)[p_2]$ of normal rank $r$, can be factored as 
\be\label{smith-factor}
\boldsymbol\Phi={\bf U}{\bf B} \boldsymbol\Lambda{\bf V},
\ee
(i)  ${\bf U}, {\bf V} \in  \mathbb{C}(p_1)[p_2]$ are {\em unimodular} matrices,\\
(ii) ${\bf B}\in \mathbb{C}(p_1)$ and $\boldsymbol\Lambda \in \mathbb{C}(p_1)[p_2]$ are both diagonal matrices with
\begin{gather}
{\bf B} = \diag [b_{1}, b_{2},\cdots,b_{m}], \label{smith-B}\\
\boldsymbol\Lambda =\diag[\lambda_1, \lambda_2, \cdots, \lambda_r, 0\cdots , 0], \label{smith-Lambda}
 \end{gather}
 such that:\\
(iii) for each $i$, $\lambda_i \in \mathbb{C}(p_1)[p_2]$ is {\em monic}, i.e.,  coefficient of the term with highest degree in $p_2$ in $\lambda_i$  is equal to one, and \\ 
(iii) for each $i$, $\lambda_i$  divides $\lambda_{i+1}$, i.e., $\lambda_{i+1}=\mu_{i+1}\lambda_i$ for some $\mu_{i+1} \in \mathbb{C}(p_1)[p_2]$.
\end{fact}
Slightly other versions of  Fact \ref{smith} are available, but our presentation here caters to the current  needs. We will not undertake its proof, but simply recall that  it is known as the {\em Smith normal form} of a matrix, and is always available for matrices with entries belonging to a Principal Ideal Domain (PID) \cite[page 404]{fb:vidyasagar}, and in particular an Euclidean domain \cite[page 179]{cameron}. The terms $\lambda_i$'s are known as the {\em invariant factors}. In standard linear systems theory (cf. \cite{dar:Kail1}), it is more commonly known for the case $\mathbb{C}(p_1)=\mathbb{C}$, i.e., when all entries of the matrices are $p_1$-independent, and are polynomials in $p_2$ only. A proof for the more general version stated above can be given essentially by noting that main ingredient of the proof is rational operations with polynomial coefficients as in Euclidean (pseudo) division algorithm, which holds for elements in $\mathbb{C}(p_1)[p_2]$ equally well.

We will also need the following result as an essential tool for our treatment of the 2-D spectral factorability result that we are about to develop in this section.
\begin{fact}\label{yasuura}
Let ${\bf R} \in \mathbb{C}(p_1)[p_2]$ be a matrix of size  $(m\times  m)$ that satisfy: 
(i) $\det {\bf R} \in \mathbb{C}(p_1)$, i.e., ${\bf R}$ is $p_2$-unimodular 
(ii) ${\bf R}={\bf R}_{*}$, i.e.,  ${\bf R}$ is parahermitian, and
(iii) ${\bf R}(j\mbox{\boldmath $\omega$})\geq 0$ for all real $\mbox{\boldmath $\omega$}$, i.e.,  ${\bf R}$ is nonnegative definite for all real frequencies. Then there exist a $p_2$-unimodular matrix ${\bf W} \in  \mathbb{C}(p_1)[p_2]$, and  a matrix ${\bf D}\in  \mathbb{C}(p_1)$, both of size  $(m\times m)$,  such that 
\be\label{yasuura-factor}
{\bf R}={\bf W}_{*}{\bf D}{\bf W}.
\ee
\end{fact}
A proof of  Fact \ref{yasuura}  is rather lengthy in detail, and is skipped for the sake of brevity. Suffice it to mention, however, that it is a straightforward generalization \cite{fettweis-spectral} of an univariate result corresponding to the situation $\mathbb{C}(p_1)=\mathbb{C}$, that reportedly first appeared in the work of Oono and Yasuura \cite{oono}. Once again, the fact that the elements of the matrix belong to $\mathbb{C}(p_1)[p_2]$, an Euclidean domain, is the key to the generalized proof. 

Furthermore, we will require the following  sharpening of Corollary \ref{C1} critical to the developments in the sequel.

\begin{fact}\label{lambda-sequence}
Let each element of the finite sequence of scalars $\lambda_i, i=1,2,\cdots, m$ from $\mathbb{C}(p_1)[p_2]$, satisfy the conditions of Corollary \ref{C1}, i.e., $\lambda_{i*} = \lambda_i \not\equiv 0$,
 and $\lambda_i (j\mbox{\boldmath $\omega$})\geq 0$  wherever $\lambda_i (j\mbox{\boldmath $\omega$})$ is
 holomorphic. Furthermore, assume that  $\lambda_i$  divides $\lambda_{i+1}$ for each $i$
 (as, for example, in Fact \ref{smith}), i.e., in particular,
\be\label{define-mu}
\lambda_1=\mu_1, \text{ and } \lambda_{i+1}=\mu_{i+1 \lambda_i }, \quad \mu_{i+1} \in \mathbb{C}(p_1)[p_2]
\ee
for each $i=1,2\cdots, m-1$. Then for each $i=1,2,\cdots, m$ we may have 
\be\label{lambda-unitarities}
\boldsymbol\Lambda_{i*} \boldsymbol\Lambda_i = \boldsymbol\Lambda_i \boldsymbol\Lambda_{i*} = \lambda_i\one_{2},
\ee
in which the $\boldsymbol\Lambda_i \in \mathbb{C}(p_1)[p_2]$ are $(2\times2)$ matrices, holomorphic in $\Re {\bf p}>0$, such that %
\be\label{matrix-divisibility}
\boldsymbol\Lambda_{r}{\boldsymbol\Lambda}^{-1}_{s} \in \mathbb{C}(p_1)[p_2], \text{  for } r\geq s.
\ee
\end{fact}
{\em Proof:}
It follows from equation (\ref{define-mu}) that each $\mu_i$ also satisfies the conditions of Corollary \ref{C1}, i.e., $\mu_{i*} = \mu_i \not\equiv 0$,
 and $\mu_i (j\mbox{\boldmath $\omega$})\geq 0$  wherever $\mu_i (j\mbox{\boldmath $\omega$})$ is
 holomorphic. Thus one can write by using Corollary \ref{C1}
\be\label{mu-unitarities}
\boldsymbol\mu_{i*} \boldsymbol\mu_i = \boldsymbol\mu_i \boldsymbol\mu_{i*} = \mu_i\one_{2},
\ee
 in which $\boldsymbol\mu_i$s are $(2 \times 2)$ matrices with elements from $\mathbb{C}(p_1)[p_2]$, and are holomorphic in $\Re {\bf p} >0$. If we now define our $\boldsymbol\Lambda_i$s as  
\be\label{prod-mu}
\boldsymbol\Lambda_i =\boldsymbol\mu_i \cdots \boldsymbol\mu_2\boldsymbol\mu_1 
\ee
then (\ref{lambda-unitarities}) follows from (\ref{define-mu}) and (\ref{mu-unitarities}), whereas it is immediate from (\ref{prod-mu}) that  for $r > s$ we have
\be\label{prod-mu-s-to-i}
 \boldsymbol\Lambda_{r}{\boldsymbol\Lambda}^{-1}_{s}=\boldsymbol\mu_r \cdots \boldsymbol\mu_{s+2}\boldsymbol\mu_{s+1}. 
\ee
The fact that elements of $\boldsymbol\Lambda_i$ and $\boldsymbol\Lambda_{r}{\boldsymbol\Lambda}^{-1}_{s}$ belong to $\mathbb{C}(p_1)[p_2]$ follows trivially from (\ref{prod-mu}) and (\ref{prod-mu-s-to-i}) respectively in view of the fact that $\boldsymbol\mu_i$s have the same property. The proof is then completed by observing that the last claim, i.e., (\ref{matrix-divisibility}) is trivial for $r=s$. \hfill{Q.E.D}

{\flushleft\bf Remark:}
While the existence of $\boldsymbol\Lambda_i$s satisfying (\ref{lambda-unitarities}) trivially follows from Corollary \ref{C1}, the additional requirement (\ref{matrix-divisibility}) is not automatically satisfied. The divisibility condition $\lambda_i | \lambda _{i+1}$, in a sense, induces a `divisibility' of the matrix $\boldsymbol\Lambda_r$ by $\boldsymbol\Lambda_s$, $r>s$, by making 
$\boldsymbol\Lambda_r {\boldsymbol\Lambda}^{-1}_s$ a member of the same ring $\mathbb{C}(p_1)[p_2]$. However, this requires a special construction  exploiting the fact that the $\boldsymbol\mu_i$s have the same property as the $\boldsymbol\Lambda_i$s, and by invoking Corollary \ref{C1} on them. 
\subsubsection{Nonsingular case}
 We first assume that $\boldsymbol\Phi$ to be of full normal rank, i.e.,   $r = \rank\boldsymbol\Phi = m$.
 By invoking Fact \ref{smith}, $\boldsymbol\Phi$ can be transformed to its Smith normal form as in (\ref{smith-factor}). 
 
  Next, we need to recall the fact that (cf. proof of Lemma 4  in \cite{S_fact:youl}) that the invariant polynomials 
 associated with  a univariate parahermitian positive definite  polynomial matrix are (self) paraconjugate, and sign definite on the imaginary axis.

 Now if we freeze $p_{1}$ at $j\omega_{10}$ in $\boldsymbol\Phi$ to
 obtain $\tilde{\boldsymbol\Phi}=\tilde{\boldsymbol\Phi}(p_2)$, then $\tilde{\boldsymbol\Phi}$ 
is certainly a parahermitian positive definite polynomial matrix in
 $p_2$, of which the invariant polynomials $\tilde{\lambda} _{i} = \tilde{\lambda} _{i} (p_2)$ are given by the 
 univariate polynomials obtained by freezing $p_{1}$ at $j\omega_{10}$ in
 $\lambda_i$. Consequently, due to the univariate result cited above \cite{S_fact:youl},
 $\tilde\lambda_i$ is (self) paraconjugate. Since this is, in fact, true for almost
 all  imaginary values of $p_1$ (except possibly those for which
 $\boldsymbol\Phi$ is not well defined), we conclude that $\lambda_i$'s themselves are (self) paraconjugate.
 In a similar manner,  it also follows from the univariate result on sign definiteness cited above from \cite{S_fact:youl}  that $\lambda_i$'s are sign definite\footnote{%
We believe that the arguments used to reach the conclusion of this paragraph can be improved, and can be made independent of results in \cite{S_fact:youl}. However, this has not been fully worked out yet.
}

 Thus, the conditions for Fact  \ref{lambda-sequence} are satisfied by the $\lambda_i$s in the Smith form (\ref{smith}), and thus the decompositions indicated by (\ref{lambda-unitarities}) holds.

Our strategy to factor $\boldsymbol\Phi$ will be to consider the factorization of $\diag [\boldsymbol\Phi, \boldsymbol\Phi ]$ first.
For this, we consider the right Kronecker product\footnote{%
For matrices of appropriate sizes  we have  
${\bf A}{\bf B}\otimes {\bf C}{\bf D}=({\bf A}\otimes{\bf C})({\bf B}\otimes{\bf D})$.} 
of (\ref{smith-factor}) with the $(2 \times 2)$ identity matrix $\one_2$, which yields:
\be\label{Phi-otimes-I}
 \boldsymbol\Phi\otimes \one_2 =
({\bf V}\otimes \one_2)
({\bf B}{\boldsymbol\Lambda}\otimes \one_2)
({\bf U}\otimes \one_2).
\ee
Furthermore, by using (\ref{lambda-unitarities}) and the Kronecker product notation one can be compactly write
\be\label{BA-otimes-I}
({\bf B}{\boldsymbol\Lambda}\otimes \one_2)=
\boldsymbol\Gamma_*({\bf B}\otimes \one_2)\boldsymbol\Gamma,
\ee
where
\be\label{Gamma-defined}
\boldsymbol\Gamma=\diag [ \boldsymbol\Lambda_1, \cdots \boldsymbol\Lambda_m ].
\ee
Next, by inserting (\ref{BA-otimes-I}) in (\ref{Phi-otimes-I}), and  using the identities
${\bf U}\otimes \one_2 = ({\bf U}{\bf V}_{*}^{-1}\otimes \one_2)({\bf V}_{*}\otimes \one_2)$
and
$\boldsymbol\Gamma_* (\boldsymbol\Lambda^{-1}\otimes\one_2)\boldsymbol\Gamma=\one_{2m}$
(cf. (\ref{smith-Lambda}), (\ref{lambda-unitarities}), (\ref{Gamma-defined}))
one can write 
\be\label{R-center}
 \boldsymbol\Phi\otimes \one_2 = 
({\bf V}\otimes \one_2)\boldsymbol\Gamma_*{\bf R} \boldsymbol\Gamma ({\bf V}_* \otimes \one_2),
\ee
where
\be\label{R-defined}
{\bf R}= 
 ({\bf B}\otimes \one_2)\boldsymbol\Gamma
 ({\bf U}{\bf V}_{*}^{-1}\otimes \one_2)\boldsymbol\Gamma_*
 ({\boldsymbol\Lambda}^{-1}\otimes \one_2).
\ee
 We now claim the following property of ${\bf R}$ so obtained.

 \begin{proposition}\label{S_fact:P4}
 The rational matrix ${\bf R}$ satisfies ${\bf R}_{*} ={\bf R}$, is an element of
 $\mathbb{C}(p_{1})[p_{2}]$, and is $p_2$-unimodular, i.e., $\det {\bf R}\not\equiv 0$ is a rational function of $p_{1}$ only.
 \end{proposition}

 {\em Proof}:
 The  property ${\bf R} = {\bf R}_{*}$ easily follows from the first equation in ({\ref{R-center}), 
 the fact that $\boldsymbol\Phi_{*}=\boldsymbol\Phi$, and that  ${\bf V}$ and ${\boldsymbol\Gamma}$ are invertible matrices.
 
We next embark on demonstrating that ${\bf R} \in \mathbb{C}(p_{1})[p_{2}]$. For this, we consider a partition of ${\bf R}$ into blocks of size $(2\times 2)$, and define the $(m\times m)$ matrix ${\bf Q}={\bf U}{\bf V}_{*}^{-1}$, the $(r,s)$-the entry of which is depicted as $q_{rs}$. From (\ref{R-defined}) straightforward matrix multiplication then yields that the $(r,s)$-th block of size $(2\times 2)$ in the partitioned matrix ${\bf R}$ is given by
\be\label{block-element-R}
[{\bf R}]_{rs}= 
b_s q_{rs}\lambda^{-1}_{s} \boldsymbol\Lambda_r \boldsymbol\Lambda_{s*}=
b_s q_{rs} \boldsymbol\Lambda_r \boldsymbol\Lambda_{s}^{-1}.
\ee
 The last equality follows from $\boldsymbol\Lambda_s \boldsymbol\Lambda_{s*}=\lambda_s \one_2$ (cf. (\ref{lambda-unitarities})).
Since ${\bf V}$ is unimodular, it also follows that the elements of  ${\bf Q}={\bf U}{\bf V}_{*}^{-1}$, thus, e.g., $q_{rs}$ belong to 
$\mathbb{C}(p_{1})[p_{2}]$. Since $b_s$ is a member of $\mathbb{C}(p_{1})$, by invoking the property (\ref{matrix-divisibility}), it then follows from (\ref{block-element-R}) that $[{\bf R}]_{rs}$ belongs to  $\mathbb{C}(p_{1})[p_{2}]$ for $r\geq s$. However, since we have ${\bf R}={\bf R}_*$, by symmetry we also have the same property for $r\leq s$, thus completing the proof that every element of ${\bf R}$ is in $\mathbb{C}(p_{1})[p_{2}]$.

Also, by considering the determinant  of (\ref{R-defined}) we obtain
 \be\label{R-unimodular}
\det {\bf R} \det ({\bf V}\otimes \one_2) = \det ({\bf B}\otimes\one_2) \det ({\bf U}\otimes\one_2).
\ee
Since ${\bf U}$ and ${\bf V}$ are unimodular, $\det ({\bf U}\otimes \one_2)$ and $\det ({\bf V}\otimes \one_2)$ are constants. Moreover\footnote{%
If ${\bf A}$ and ${\bf B}$  are squares matrices of size $m$ and $n$ respectively then $\det ({\bf A}\otimes{\bf B})= [\det ({\bf A}]^n [\det{\bf B})]^m$.}, 
 $\det ({\bf B}\otimes \one_2)=(\det {\bf B})^{2}=(b_{1} b_{2} \cdots b_{m})^{2}$ 
belongs to $\mathbb{C}(p_{1})$, since each $b_i$ in (\ref{smith-B}) belongs to $\mathbb{C}(p_{1})$. It thus follows from (\ref{R-unimodular})  that $\det {\bf R}$ belongs to $\mathbb{C}(p_{1})$, which completes the proof of the proposition.\hfill{Q.E.D}

We now claim that the matrix $\bf R$ constructed above  satisfy all conditions required by Fact \ref{yasuura}. We have already demonstrated that $\bf R$ is (i) $p_1$-unimodular, and (ii) para-Hermitian. The fact that it is nonnegative  definite for all real frequencies, i.e., ${\bf R} (j\boldsymbol\omega) \geq 0$ for all real $\boldsymbol\omega$, follows by appealing to (\ref{R-center}), in which $\boldsymbol\Phi(\boldsymbol\omega)\geq 0$ for all real $\boldsymbol\omega$, and both $\bf V$ and $\boldsymbol\Gamma$ are invertible  matrices. Thus, by invoking Fact \ref{yasuura} it follows that 
$\bf R$ can be factored as in (\ref{yasuura-factor}), where ${\bf W} \in \mathbb{C}(p_{1})[p_{2}]$ is an unimodular $(2 \times 2)$ matrix and ${\bf D}(p_{1}) \in \mathbb{C}(p_{1})$.  

It is clear from an argument that has already been used several times by now, 
that $\bf D$ inherits the para-Hermitian property and non-negativity on the real axis, i.e., ${\bf D}={\bf D}_*$ and ${\bf D}(j\omega_1)\geq 0$ for all real $\omega_1$ for which $\bf D$ is well defined. Thus, the standard spectral factorization result of 1-D system theory (e.g., \cite{S_fact:youl}) can now be applied to $\bf D$, which yields
 \be\label{S_fact:86}
 {\bf D} = {\bf L}_{*}{\bf L},
 \ee
 where ${\bf L} \in\mathbb{C}(p_{1})$, is of size $(2m \times 2m)$, and is holomorphic in $\Re p_1 > 0$.  

Combining equations  (\ref{yasuura-factor})  and (\ref{S_fact:86}) with (\ref{R-center}) we then have 
\be\label{Phi-Phi-factor}
 \boldsymbol\Phi\otimes \one_2 = 
({\bf V}\otimes \one_2)\boldsymbol\Gamma_*
{\bf W}_{*}
 {\bf L}_{*}{\bf L}
{\bf W}
\boldsymbol\Gamma ({\bf V}_* \otimes \one_2).
\ee
Furthermore, by using a column permutation matrix $\bf P$ such that
\be\label{permutation}
\diag[\boldsymbol\Phi, \boldsymbol\Phi] =
{\bf P}_* (\boldsymbol\Phi\otimes\one_2){\bf P},
\ee
 by pre-multiplying (\ref{permutation}) by $[\one_m {\bf 0}]$, post-multiplying by its transpose, and finally by substituting (\ref{Phi-Phi-factor}) into the resulting equation, we obtain the desired 
spectral factor ${\bf H}$ as 
\[
\boldsymbol\Phi={\bf H}_* {\bf H},
\quad
{\bf H}={\bf L}{\bf W}
\boldsymbol\Gamma ({\bf V}_* \otimes \one_2){\bf P}
\left[ \begin{array}{c} \one_{m} \\ {\bf 0} \end{array} \right].
\]
Indeed, it is then easily seen that the $2m\times m$ matrix ${\bf H}$ so obtained as a spectral factor of $\boldsymbol\Phi$ is holomorphic in $\Re {\bf p} >0$.
 \subsubsection{Singular Case}\label{singular_case}
 Thus far, we have assumed that $\boldsymbol\Phi$ to be full normal rank.  We
 now allow $\boldsymbol\Phi$ to be a {\em singular} $(m \times m)$
 polynomial matrix having normal rank $r=\rank\boldsymbol\Phi <m$.  By post-multiplying $\boldsymbol\Phi$ 
 by an appropriate $p_2$-{\em unimodular} $(m \times m)$ matrix, ${\bf F}
 \in \mathbb{C}(p_{1})[p_{2}]$, $\boldsymbol\Phi$ can be transformed to column Hermite form \cite{dar:Kail1}, i.e.,
 \be\label{S_fact:88}
 \boldsymbol\Phi {\bf F} = \left[ \begin{array}{cc}
 {\bf S} & {\bf 0} \end{array} \right],
 \end{equation}
 where ${\bf S}$ is of size $(m \times r)$ and is an element of
 $\mathbb{C}(p_{1})[p_{2}]$.  The last $m - r$ columns of $\boldsymbol\Phi {\bf
 F}$ are identically zero.  Clearly, the last $m - r$ columns of the $(m
 \times m)$ matrix ${\bf B} = {\bf F}_{*}\boldsymbol\Phi {\bf F} \in
 \mathbb{C}(p_{1})[p_{2}]$ are also identically zero, and since, ${\bf B}_{*} = {\bf
 B}$, the last $m - r$ rows of ${\bf B}$ are identically zero as well.  In other
 words, ${\bf B}$ can be written as
 \be\label{S_fact:90}
 {\bf B} = {\bf F}_{*}\boldsymbol\Phi {\bf F} = 
 \frac{1}{f_{*}f} 
\left[
 \begin{array}{cc}
 {\bf E} & {\bf 0} \\
 {\bf 0} & {\bf 0}
 \end{array} 
\right],
\ee
 where $f \in \mathbb{C}[p_{1}]$ and ${\bf E}$ is an $(r \times r)$ polynomial matrix
 having the properties
 ${\bf E} \in  \mathbb{C}[p_{1},p_{2}]$,
 ${\bf E}_{*} = {\bf E}$, ${\bf E}(j\mbox{\boldmath $\omega$}) \geq 0$ 
 for all real
 $\mbox{\boldmath $\omega$}$, $r = \rank {\bf E}$.
 Hence, due to our earlier results, the nonsingular matrix ${\bf E}$ can be
 factorized as ${\bf E} = {\bf K}_{*}{\bf K}$,
 where ${\bf K} \in \mathbb{C}(p_{1})[p_{2}]$ and is of size $(2r \times r)$.  
 Thus, in
 view of ({\ref{S_fact:90}}) we finally obtain
 $\boldsymbol\Phi = {\bf H}_{*}{\bf H}$,
 where ${\bf H}$ is the $(2r \times r)$ matrix
 \be\label{S_fact:92}
 {\bf H} = \frac{1}{f}[{\bf K} \;\;\; {\bf 0}]{\bf F}^{-1}.
 \ee
 Here ${\bf H}$ is an element of $\mathbb{C}(p_{1})[p_{2}]$, because $\bf F$ is $p_2$-unimodular
 (i.e. ${\bf F}^{-1} \in \mathbb{C}(p_{1})[p_{2}]$, and in fact  $p_2$-unimodular).  The
 holomorphicity of ${\bf H}$ in 
 $\Re p_{1} > 0$ can be achieved by appropriately collecting the zeros of $ff_*$ in $f$ and $f_*$ .
 Hence, the desired factorization is achieved.

 The following theorem summarizes the highlights of  discussions of the last two subsections.
 \begin{theorem}\label{S_fact:T5}
 Let $\boldsymbol\Phi=\boldsymbol\Phi({\bf p})$ be an $m \times m$ polynomial matrix in two
 variables ${\bf p} = (p_{1},p_{2})$ with normal rank $r$ such that (i)
 $\boldsymbol\Phi$ is para-hermitian (i.e. $\boldsymbol\Phi_{*}
 =\boldsymbol\Phi$), (ii) $\boldsymbol\Phi$ is nonnegative
 definite for all ${\bf p} = j\mbox{\boldmath $\omega$}$.
 Then there exists a $( 2r \times m)$ rational matrix ${\bf H} \in \mathbb{C}(p_{1})[p_{2}]$, which is
 holomorphic in $\Re {\bf p} > 0$ and is a spectral factor of $\boldsymbol\Phi$ in the sense
 \be\label{S_fact:93}
\boldsymbol\Phi = {\bf H}_{*}{\bf H}.
 \ee
Moreover, it is possible to have the least common denominator of $\bf H$ as a polynomial in $p_1$ only.
 \end{theorem}
 In view of the central importance of this result in many areas
 of $2$-D system theory, we shall also make a definitive statement of the discrete
 version of the above spectral factorability result in the form of a theorem.
 \begin{theorem}\label{discrete_spectral_factor}
  Let $\boldsymbol\Psi = \boldsymbol\Psi (\bf z)$
 be an $(m \times m)$ polynomial matrix in two variables ${\bf  z}=(z_{1},z_{2})$ with normal rank $r$ such that 
(i) $\boldsymbol\Psi$ is parahermitian in the discrete sense\footnote{%
The notation $\tilde{\mbox{}}$ denotes the discrete paraconjugate of a rational matrix obtained by considering replacement of $z_i$ by $z_{i} ^{-1}$, complex conjugate of the coefficients, and finally by matrix transposition. A rational matrix is called discrete parahermitian, if it is equal to its own discrete paraconjugate}
 i.e., $\tilde{\boldsymbol\Psi} = \boldsymbol\Psi$ and (ii) $\boldsymbol\Psi$ 
 is nonnegative definite on the distinguished boundary of the unit bi-disc $|{\bf z}|=1$.
 Then there exists a $(2r \times m)$ rational matrix ${\bf G} \in \mathbb{C}(z_{1})[z_{2}]$,
 holomorphic in the open unit bi-disc $|\bf z| < 1$, such that the spectral factorization
 \be\label{S_fact:**}
 \boldsymbol\Psi = \tilde{\bf G}{\bf G}
 \ee
 holds.  Moreover, it is possible to have the least common denominator of ${\bf G}$ as a separable polynomial, the $z_2$ factors of which are only of the type $(1+ z_2)$.
 \end{theorem}
 We shall sketch a proof via the use of the double bilinear transformation:
 \be\label{S_fact:a*}
 z_{i} = \frac{1-p_{i}}{1+p_{i}}, \quad
 p_{i}=\frac{1-z_{i}}{1+z_{i}}; \quad i = 1,2
 \ee
 Define $\boldsymbol\Phi = \boldsymbol\Phi(\bf p)$ obtained
 from $\boldsymbol\Psi = \boldsymbol\Psi (\bf z)$ by the
 transformation (\ref{S_fact:a*}).  
It can then be verified that
 $\tilde{\boldsymbol\Psi} = \boldsymbol\Psi$ translates into
 $\boldsymbol\Phi_{*}\boldsymbol\Phi=\one$, and the property
 $\boldsymbol\Psi \geq 0$ on the distinguished boundary of the
 polydisc $|{\bf z}|=1$ translates into $\boldsymbol\Phi \geq 0$ for ${\bf p} =
 j\mbox{\boldmath $\omega$}$.  
 Thus, $\boldsymbol\Phi$ admits the
 spectral factorization (\ref{S_fact:93}) as in Theorem \ref{S_fact:T5}. 
 Now, if ${\bf G} = {\bf G}(\bf z)$ is obtained from ${\bf H} = {\bf H}(\bf p)$ via
 the inverse bilinear transformation  in (\ref{S_fact:a*}) then
 (\ref{S_fact:93}) translates to (\ref{S_fact:**}).
 The property of holomorphy of ${\bf H}$ in $\Re {\bf p}>0$
 is translated to the property of holomorphy of ${\bf G}$ in $|{\bf z}| < 1$.
 This completes the proof of the discrete version of Theorem \ref{discrete_spectral_factor}.

 {\flushleft\bf Remark:} Note the fact that ${\bf G}$ so obtained via ${\bf H}$ need not have a univariate
 denominator even if ${\bf H}$ has the same property.  This is
 because of the fact that the denominator of ${\bf G}$ may contain factors of
 the type $(1+z_{2})$ even if the denominator of ${\bf H}$ is independent of
 $p_{2}$.  It is likely that as in the continuous domain case, a spectral factor ${\bf G}$ with univariate denominator can
 indeed be obtained. However, its consideration would require a more direct and detailed analysis of the
 discrete case possibly exploiting extra  steps of Cassel Reduction, but such details need to be worked out.

{\flushleft\bf Remark:} The spectral factors do not show all the properties known from corresponding
 1-D techniques.  In particular, they have no right inverses in general, which
 prevents applicability to certain problems in their \mbox{2-D} systems
 theory.  On the other hand, it is best that can be achieved in view of the
 fact that there exists two-variable para-even polynomials that {\em cannot} have
 {\em scalar} spectral factors, i.e., every possible spectral factor is of
 order $\ell \times 1$, where $\ell \geq 2$.

{\flushleft\bf Remark:} Clearly, a decomposition such as the one given in Theorem \ref{S_fact:T5} is
 highly {\em non-unique} even in the scalar case.  The non-uniqueness arises from several
 sources.  First, the sum of squares representation
 for a nonegative polynomial is not unique.  Thus, for a given $\phi$ there
 exists more than one set of $s_{i}$'s, satisfying (\ref{1.28}).  Furthermore, for a fixed set of $s_{i}$'s the choice of
 $g_{1}$ and $g_{2}$ in (\ref{1.29}) is not unique either.  For example,
 the choice $g_{1} = s_{1} + js_{3},\; g_{2}=s_{2}+js_{4}$ is also valid for
 the rest of the development.  Clearly, there exists yet other possibilities.
 Further non-uniqueness is introduced in the martix case by the non-uniqueness
 of the unimodular matrices ${\bf U}$ and ${\bf V}$ in the Smith canonical
 form representation (\ref{smith-factor}) and by the non-uniqueness of the
 decomposition (\ref{S_fact:86}).

 {\flushleft\bf Remark:} If ${\bf H}$ is a valid spectral factor having the above properties, then
 $\boldsymbol\Psi{\bf H}$ is also a valid spectral factor where
 $\boldsymbol\Psi$ is a stable {\em all-pass} rational $(2r \times 2r)$
 matrix i.e., $\boldsymbol\Psi = \boldsymbol\Psi_{*}$ and
 $\boldsymbol\Psi$ is holomorphic in Re ${\bf p} > 0$.  In $1$-D, it is
 known that any two (stable) spectral factors are related via multiplication
 by such an all-pass spectral factor \cite{Dewilde_and_Vandewalle}.  The characterization of the entire
 family of spectral factors are thus obtained.  This raises the issue of
 minimality of the decomposition (\ref{S_fact:**}) or (\ref{S_fact:93}) which has an important
 role in 1-D system theory.  A subclass of spectral factors is said to be
 minimal if the matrix $\boldsymbol\Psi$ relating any two members of
 the subclass is trivial in the sense that $\boldsymbol\Psi$ has
 McMillan degree equal to zero i.e., no dynamical elements are needed in its
 realization.  While there is no proof that the ${\bf H}$ obtained by using
 the algorithm prescribed in the proof of Theorem \ref{S_fact:T5} is minimal,
 this consideration remains an open question.

  {\flushleft\bf Remark:} Note that even if $\phi$ is a polynomial with {\em real
 coefficients}, our spectral factors may not have real coefficients.  This is because $g_{1}$ and
 $g_{2}$ in (\ref{1.29}) are not a priori real polynomials.  In fact, for
 $g_{1}$ and $g_{2}$ to be real it is necessary that $s_{1}$ and $s_{2}$ be
 even and $s_{3}$ and $s_{4}$ be odd functions of $p_{1}$ and $p_{2}$ i.e.,
 $s_{i}(\omega_{1},\omega_{2}) = s_{i}(-\omega_{1},-\omega_{2})$ 
 for $i = 1,3$  and $s_{i}(\omega_{1},\omega_{2}) = -s_{i}(-\omega_{1},-\omega_{2})$ for
 $i=2,4$.  Since this is not guaranteed, nor do we know exactly how to ensure
 this, if at all possible, by using the non-uniqueness of the sum the of
 squares representation, the realness of ${\bf H}$  or ${\bf G}$ cannot be ensured.

Due to the above, in contrast to 1-D, at the present state of our knowledge internally passive synthesis of prescribed 2-D real rational positive or bounded transfer matrices, in the form of  impedance  or the scattering matrix of a network is not necessarily feasible with real valued elements, and may  require complex elements. 

 {\flushleft\bf Remark:}
The proof technique also brings into focus the non-uniqueness as
 well as the complex nature of the 2-D spectral factors. The characterization of
 the entire family of these factors, especially those with  minimal  
 values of  $r$ as well as of minimal degree still remain open problems of
 research. While the proof of spectral factorability presented here is constructive, and can be implemented in principle 
 in symbolic processors  with infinite degree of precision, other numerical algorithms for
 obtaining the spectral factors, possibly including approximations to them,  such as those exploiting  semidefinite programs \cite{trentelman}
and cepstral methods \cite{S_fact:wood} are potentially feasible. These issues will not be considered here and remain to be
 investigated.

 \section{Synthesis of 2-D lossless bounded matrices}\label{scatt:S21}
  In this section the issue of internally passive synthesis of
 bounded matrices\footnote{%
An  $(n\times n)$ rational matrix ${\bf S}({\bf p})$ is said to be bounded if it is holomorphic in $\Re {\bf p}>0$, 
and $\one_n - {\bf S}^* (j\boldsymbol\omega){\bf S} (j\boldsymbol\omega)$ is nonnegative definite for 
all real values of $\boldsymbol\omega$ for which ${\bf S}$ is well defined. If, in addition, ${\bf S}_*{\bf S}=\one_n$ we say that ${\bf S}$ is lossless bounded. 
If the coefficients of $\bf S$ are real
then we may call it a bounded real, or a lossless bounded  real matrix as the case may be. 
}
 in arbitrary topological structure is considered.  Our
 discussion begins with $2$-D lossless bounded rational marices and it is
 shown that such matrices can indeed be synthesized.
 Furthermore, if the lossless bounded matrix is real then a minimal network
 with real elements exists.  The question of synthesizability of bounded, but
 not necessarily lossless rational matrices is taken up next.  Given a
 lossless synthesis, our approach is feasible if and only if a certain unitary
 dilation or embedding problem for $2$-D rational matrices is solvable in
 section.  It is shown that the problem can indeed be solved in
 $2$-D via a spectral factorability interpretation of the sum of squares representation, 
which essentially states the positive
 polynomials in two variables can be expressed as sum of squares of
 polynomials in one variable the coefficients of which are rational functions
 of the other variable.  The synthesizability of $2$-D bounded matrices are
 thus established.  This in turn establishes the $2$-D bounded real lemma
 which provides a state space characterization of the property of
 dissipativeness.  We discuss this development in Section \ref{scatt:S23}.
 Synthesizability of three or higher dimensional lossless bounded matrices are
 considered in Section \ref{H-dim}.  It is proved that such a synthesis is not
 feasible in general.  The essential reason for this is rooted in the answer
 to Hilbert's seventeenth problem that positive polynomials cannot in general
 be represented as a sum of squares of polynomials.  However, since such
 representations are feasible for polynomials of total degree at most four, it
 turns out that $3$-D scalar lossless bounded functions are not synthesizable.
 This raises the problem of characterization of synthesizable multidimensional
 rational functions or matrices.  Clearly, such functions must be a subclass
 of a class of $k$-D (lossless) matrices.  While this problem remains largely
 open, in Section \ref{scatt:S24}, we have included an interesting category of
 multidimensional functions that can be synthesized.

 It will be shown that every $(m \times m)$ two-variable
 discrete lossless bounded matrix ${\bf S}$, can be visualized with a minimal
 number of frequency-dependent elements.  This synthesis procedure can be
 divided into ta number of steps.  First we try to find an $(m + mn_{2}) \times (m +
 mn_{2})$ one-variable discrete scattering matrix, ${\bf
 S}^{'}={\bf S}^{'}(z_{1})$, which describes a $z_{1}$-dependent one-dimensional
 $(m + mn_{2})$-port $\NN^{'}$, where $n_{2}$ will be specified later.  The
 matrix ${\bf S}^{'}$ may not be a discrete lossless bounded matrix, we only
 demand that we obtain the two-dimensional $m$-port $\NN$, which is described by
 the given discrete lossless bounded matrix ${\bf S}$, if the last $mn_2$ ports of
 $\NN^{'}$ are terminated with $z_{2}$-type elements (cf. Figure \ref{fig:lossless-synthesis}).  
 Thus, it is possible to apply a $z_{1}$ dependent similarity
 transformation to ${\bf S}^{'}$ in such a way that the internal structure of
 $\NN$ is changed without influencing the external behavior, which is
 described by ${\bf S}$.

 In the second step we use such a similarity transformation in
 order to obtain a new one-variable matrix ${\bf S}^{''}={\bf S}^{''}(z_{1})$,
 describing a coupling network, $\NN^{''}$, which can be realized with methods
 known in classical network theory. In a third step we will use an analytical trick to reduce the number of $z_2$ dependent elements in the realization in such a way that the resulting realization would be minimal. Our presentation is a streamlined version of \cite{S_fact:kumm} based on original inspirations from  the work of Youla \cite{stb:Youl1}.
\begin{figure}[h]
    \centering
    \includegraphics[width=14cm]{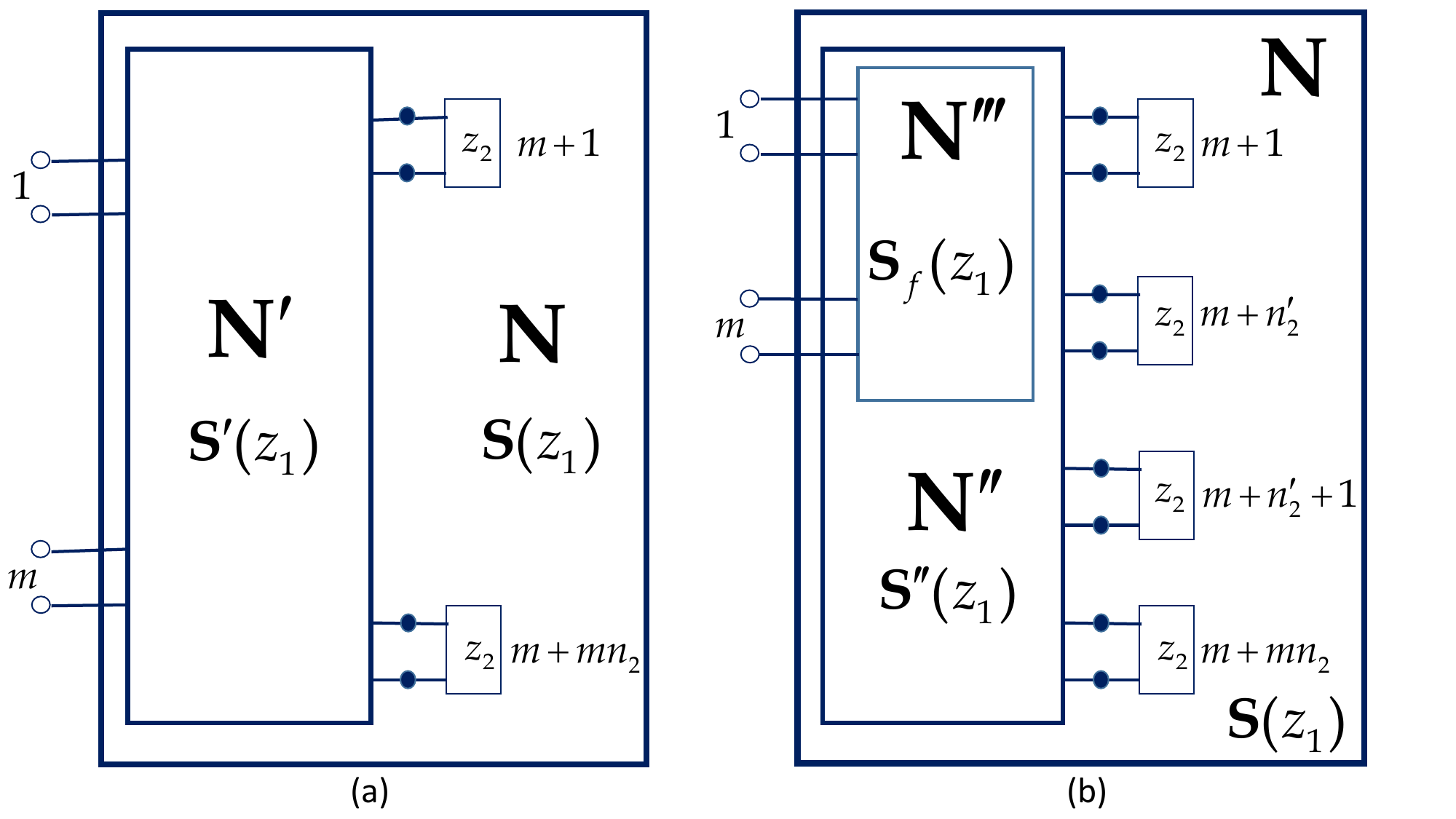}
    \caption{Steps in synthesis of a bounded lossless scattering matrix.
  $\NN^{"}$ is obtained from $\NN^{'}$ by applying the similarity transformation $\diag [\one\;\; {\bf T}(z_1) ]$. }
    \label{fig:lossless-synthesis}
\end{figure}

Before delving into the details, we establish some notation and recall some properties of the prescribed 
 discrete lossless scattering matrix $\bf S$ that will be useful for further discussions.
 First, note that\footnote{%
The `inverse conjugate polynomial' $\hat{a}$ is obtained by replacing in $a$ the variables $z_i$ by $z_{i}^{-1}$, multiplying by $n_i$, where $n_i=\deg _{i} a$, for each $i$, and complex conjugating the coefficients.
} 
 \be\label{scatt:3}
 {\det} \; {\bf S} = \sigma {\bf z}^{\bf m}\frac{\hat{a}}{a},
 \ee
 where $ {\bf m}=(m_1, m_2)$ is a nonnegative integer pair,
 $a$ is a scattering Schur polynomial \cite{rob:Basu1}, $\sigma$ is a unimodular constant, 
 i.e., $|\sigma|=1$, and ${\bf S}$ can be written as
 \be\label{scatt:4}
 {\bf S} = \frac{\PP}{a},
 \ee
 in which $\PP$ is a polynomial matrix in $z_{1}, z_{2}$.

 Furthermore, considering the determinant of (\ref{scatt:4}) we have $\det {\bf
 P}=a^{m} \det {\bf S} = a^{m-1}\hat{a}{\bf z}^{\bf m}$ where (\ref{scatt:3})
 has been used in the last step.  Thus, if ${\deg}_{i}\PP$ denotes the
 largest degree in the polynomial entries in $\PP$ then we have
 \[
m \; {\deg}_{i}\PP \geq {\deg}_{i} \; {\rm det} \; \PP = {\deg}_{i}
 \; \hat{a}+(m-1) \; {\deg}_{i} \; a + m_{i} \geq m \; {\deg}_{i} \;
 a,
\]
 The last step follows from the fact that ${\deg}_{i} \; a = {\rm
 deg}_{i} \; \hat {a}$ for any scattering Schur polynomial $a$, because 
$a$ cannot have a factor $z_1$ or $z_2$.  Thus, we have $n_{i}={\rm
 deg}_{i}\PP \geq {\deg}_{i} \; a$, $i = 1,2$.  Consequently, $\PP$ and $a$ can be written as
 \be\label{scatt:5}
 \PP=\PP_{0}+\PP_{1}z_{2} + \cdots + \PP_{n_{2}}z_{2}^{n_{2}},
 \ee
 \be\label{scatt:6}
 a = a_{0} + a_{1}z_{2} + \cdots + a_{n_{2}}z_{2}^{n_{2}},
 \ee
 where the coefficients $\PP_{i}=\PP_{i}(z_{1})$ and $a_{i}=a_{i}(z_{1}), \;
 i=0$ to $n_{2}$, are  polynomials in $z_{1}$ only with $\PP_{n_{2}}
 \not \equiv 0$.  Besides, since $a$ is scattering Schur, the polynomial 
$a_{0}=a_{0}(z_{1}) = a(z_{1},0)$ is also scattering
 Schur.

 \subsection{First step}\label{first-step}
 Given a two-variable $(m \times m)$ discrete lossless bounded matrix,
 ${\bf S}={\bf S}(z_{1},z_{2})$, our immediate task is to find a rational matrix ${\bf S}^{'}(z_{1})$ of
 order $(m \times mn_{2})$ satisfying 
 \be\label{scatt:1}
 {\bf S} = \HH_{11}+\HH_{12}(z_{2}^{-1}\one_{mn_{2}}-\HH_{22})^{-1}\HH_{21},
 \ee
 where
 \be\label{scatt:2}
 {\bf H}=
 \left[\begin{array}{cc}
 \HH_{11} & \HH_{12} \\
 \HH_{21} & \HH_{22}
 \end{array} \right].
 \ee
 In order to determine the elements of ${\bf H}$ such  that equation ({\ref{scatt:1}}) holds,  we first define the quantities
 \be\label{scatt:7}
 u_{i}= \frac{a_{i}}{a_{0}}, \;\; 
 \NN_{i} = \frac{\PP_{i}}{a_{0}}, \;\;\;\;\; i=0 \; {\rm to} \; n_{2},
 \ee
and
 \be\label{scatt:9}
 \y_{i}=\NN_{i}-\NN_{0}u_{i}, \;\;\;\;\;\;\;\; i = 1 \;{\rm to} \; n_{2}.
 \ee
Then standard algebraic manipulations yield the following result.
 \begin{fact}\label{scatt:T1}
 The rational matrices
 $\HH_{11}$, $\HH_{12}$, $\HH_{21}$, and $\HH_{22}$ defined in (\ref{scatt:10})-(\ref{scatt:13})  satisfy ({\ref{scatt:1}}).
 \be\label{scatt:10}
 \HH_{11} = \NN_{0},
 \ee
 \be\label{scatt:11}
 \HH_{21}=[\y_{n_{2}}^{T} \;\; \y_{n_{2}-1}^{T} \;\; \cdots \;\; \y_{2}^{T}
 \;\; \y_{1}^{T}]^{T},
 \ee
 \be\label{scatt:12}
 \HH_{12}=[{\bf 0} \;\; {\bf 0} \;\; \cdots \;\; {\bf 0} \;\; {\bf 1}_{m}],
 \ee
 \be\label{scatt:13}
 \HH_{22}= \left[ \begin{array}{ccccc}
 \zero & \zero & \cdots & \zero & -u_{n_{2}}\one_{m} \\
 \one_{m} & \zero & \cdots & \zero & -u_{n_{2-1}}\one_{m} \\
 \zero & \one_{m} & \cdots & \zero & -u_{n_{2-2}}\one_{m} \\
 \vdots & \vdots &      & \vdots & \vdots \\
 \zero & \zero & \cdots & \one_{m} & -u_{1}\one_{m} \end{array} \right].
 \ee
 \end{fact}
Since $a$ is scattering Schur  $a_{0} = a_{0}(z_{1}) = a(z_{1},0)$ is also a scattering Schur polynomial \cite{rob:Basu1}. 
 Thus, in view of  equations (\ref{scatt:7})-(\ref{scatt:13}), and the fact that the factor as
 well as products of scattering Schur polynomials are necessarily scattering Schur, 
 we note that the matrix ${\bf H}$ defined via (\ref{scatt:2}) and
 ({\ref{scatt:10}})-({\ref{scatt:13}}) is holomorphic in $|z_{1}|<1$,
 because the denominator polynomial of ${\bf H}$ is composed only of factors
 of a scattering Schur polynomial. 

 However, ${\bf H}$ so obtained  is not a discrete
 lossless bounded matrix.  We show next that a  $z_1$ dependent similarity transformation can
 be found, which transforms ${\bf H}$ to a discrete lossless bounded matrix.

 \subsection{Second step}
 Next, for any  square nonsingular rational matrix $\TT=\TT(z_{1})$ of size $mn_2$
 if we define the quantities $\BB_{11}$, $\BB_{12}$, $\BB_{21}$ and $\BB_{22}$ 
 by appropriately partitioning $\BB$ via
\be\label{define_B}
\BB=
\left[\begin{array}{cc}
 \BB_{11} & \BB_{12} \\
 \BB_{21} & \BB_{22}
 \end{array} \right]
=
\TT_{0} \HH\TT_{0}^{-1},\quad
\TT_{0}=
 \left[ \begin{array}{cc}
 \one_{m} & \zero \\
 \zero & \TT \end{array} \right]
\ee
then from (\ref{scatt:1}) we obviously have
 \be\label{scatt:26a}
 {\bf S}=\BB_{11}+\BB_{12}(z_{2}^{-1}\one_{mn_{2}}-\BB_{22})\BB_{21}.
 \ee
One can then interpret the matrix $\BB$ defined in  (\ref{define_B})
as the discrete scattering matrix of a one-dimensional
 system $\NN^{''}$, which is obtained by transformation of the internal variables 
 of the network $\NN^{'}$ described by $\BB$, and results in a two-dimensional network
 $\NN$ described by the given matrix ${\bf S}$, if the last $mn_{2}$ ports of
 $\NN^{''}$ are terminated with $z_{2}$-type elements. We refer to  Figure \ref{fig:lossless-synthesis} for  diagrammatic details. 

 So far $\TT$ has been left arbitrary except for that it is invertible.  In order to take advantage of this fact, 
we now wish to determine the transformation matrix $\TT$ in such a way that
 \be\label{scatt:27}
 \BB \left[\begin{array}{cc}
 \one_{m+r} & \zero \\
 \zero & \zero
 \end{array} \right] \tilde{\BB} =
 \left[\begin{array}{cc}
 \one_{m+r} & \zero \\
 \zero & \zero
 \end{array} \right],
 \ee
 where $r \leq mn_{2}$.  The advantage of such a transformation will be
 illuminated later. For this, we replace $\BB$ in ({\ref{scatt:27}}) by
 (\ref{define_B}) and compare corresponding block-elements on both sides 
of the resulting equation, and thus obtain
 \be\label{scatt:29}
 \HH_{11}\tilde{\HH}_{11} + \HH_{12}\KK\tilde{\HH}_{12} = \one_{m},
 \;\;
 \HH_{21}\tilde{\HH}_{11} + \HH_{22}\KK\tilde{\HH}_{12} = \zero,
 \ee
 \be\label{scatt:31}
 \HH_{21}\tilde{\HH}_{21} + \HH_{22}\KK\tilde{\HH}_{22} = \KK,
\;\;
 \KK = \TT^{-1} \left[ \begin{array}{cc}
 \one_{r} & \zero \\
 \zero & \zero \end{array} \right]
 \tilde{\TT}^{-1}.
 \ee
Our task is to find a rational matrix $\KK$ which is a solution
 to the system of equations ({\ref{scatt:29}})-({\ref{scatt:31}}), and in particular, 
 can be decomposed as indicated in the second equation ({\ref{scatt:31}}).  To this end, we claim that the
 matrix
\begin{multline}\label{scatt:32}
\KK = \left[ \begin{array}{cccc}
 u_{0}\one_{m} & u_{1}\one_{m} & \cdots & u_{n_{2}-1}\one_{m} \\
                          & u_{0}\one_{m} & \cdots & u_{n_{2}-2}\one_{m} \\
                          &                           &            & u_{n_{2}-3}\one_{m} \\
                         &           \bigcirc      & \ddots & \vdots \\
                         &                           &            & u_{0}\one_{m} \end{array} \right]
 \left[ \begin{array}{cccc}
 \tilde{u}_{0}\one_{m} &                       &       &    \\
 \tilde{u}_{1}\one_{m} & \tilde{u}_{0}\one_{m} & \bigcirc &    \\
 \tilde{u}_{2}\one_{m} & \tilde{u}_{1}\one_{m} &       &    \\
 \vdots                & \vdots                & \ddots &    \\
 \tilde{u}_{n_{2}-1}\one_{m} & \tilde{u}_{n_{2}-2}\one_{m} & \cdots &
 \tilde{u}_{0}\one_{m} \end{array} \right] 
\\
 - \left[ \begin{array}{cccc}
 \NN_{0} & \NN_{1} & \cdots & \NN_{n_{2}-1} \\
    & \NN_{0} & \cdots & \NN_{n_{2}-2} \\
    &    &    & \NN_{n_{2}-3} \\
    & \bigcirc & \ddots & \vdots \\
    &    &    & \NN_{0} \end{array} \right]
 \left[ \begin{array}{cccc}
 \tilde{\NN}_{0} &    &    &    \\
 \tilde{\NN}_{1} & \tilde{\NN}_{0} & \bigcirc   &  \\
 \tilde{\NN}_{2} & \tilde{\NN}_{1} &    &    \\
 \vdots & \vdots & \ddots &    \\
 \tilde{\NN}_{n_{2}-1} & \tilde{\NN}_{n_{2}-2} & \cdots & \tilde{\NN}_{0}
 \end{array} \right]
\end{multline} 
 is a solution of the system of equations ({\ref{scatt:29}}) and the first equation in ({\ref{scatt:31}}), 
where the quantities $u_{i}$ and $\NN_{i}$ are given by ({\ref{scatt:7}}).  The proof of this
 claim is routine algebraic verification and is left ot the reader.

 We next show that the matrix $\KK$, given by  ({\ref{scatt:32}}), can be decomposed as required by the second equation ({\ref{scatt:31}}).  For this, we first consider the eigenvalues of the matrix $\HH_{22}$, and note that routine algebraic manipulation by using the structure of the matrix $\HH_{22}$ yields the identity 
\be\label{scatt:33} 
{\det}(z_2 ^{-1} \one_{mn_{2}}-\HH_{22}) = 
\left(  z_2^{-n_{2}}  \frac{a}{a_{0}}  \right)^m, 
\ee
Since $a$ is a scattering Schur polynomial, $a(z_{1},z_{2}) \neq  0$ for $|z_{2}|\leq 1$ for all $z_{1}$ on $|z_{1}|=1$ with at most  a finite number of exceptions.  Thus, we arrive at the following result.
 \begin{fact}\label{scatt:P1A}
For almost all values  of $z_{1}$ on $|z_{1}|=1$ the
 eigenvalues of ${\bf H}_{22}={\bf H}_{22}(z_{1})$ are located in the open
 unit disc.
 \end{fact}
 By recognizing that (\ref{scatt:31}) is, in fact, a Lyapunov equation, it
 then follows that the ${\bf K}$ as given by (\ref{scatt:32}) is the unique
 nonnegative definite Hermitian solution to (\ref{scatt:31}) for almost all
 values of $z_{1}$ on $|z_{1}|=1$.  This also implies, in particular, that
 ${\bf K}=\tilde{{\bf K}}$, i.e., ${\bf K}$ is parahermitian in the discrete sense.

 Thus, due to the classically known 1-D spectral factorability\footnote{%
Note this is a crucial step. An attempt to extend the present synthesis method to higher dimensions (e.g., in 3-D), runs into difficulties because of lack of this factorability (in 2-D).
}
 result \cite{S_fact:youl,Dewilde_and_Vandewalle}, the univariate rational matrix $\KK$ can be factored as in the first equation (\ref{scatt:34}), where $\DD$ is a unimodular polynomial matrix, $\GG$ and $\GG^{-1}$ are
 invertible $(r \times r)$ rational matrices\footnote{%
Note also that $\bf G$ and $\bf D$ have real coefficients if the specified ${\bf S}$ has real coefficients. 
}
both holomorphic in $|z_{1}|<1$, where
\be\label{rank-K}
 r = {\rank} \; \KK.
\ee
Then  $\TT$ and $\TT^{-1}$ defined in (\ref{scatt:34}) are invertible rational matrices, which are both
 holomorphic in $|z_{1}|<1$, because $\GG$ and $\GG^{-1}$ have the same
 property and $\DD$ is unimodular.
 \be\label{scatt:34}
 \KK=\DD^{-1} \left[\begin{array}{cc}
 \GG\tilde{\GG} & \zero \\
 \zero & \zero \end{array} \right]
 \tilde{\DD}^{-1},
\quad
 \TT = \left[\begin{array}{cc}
 \GG^{-1} & \zero \\
 \zero & \one \end{array} \right] \DD,
 \ee
 It then immediately follows that $\KK$ admits the factorization required by ({\ref{scatt:31}}),
thus demonstrating that a $\TT$ can indeed be obtained that satisfies ({\ref{scatt:29}}) and ({\ref{scatt:31}}).
\subsection{Third step}\label{third-step}
Thus far, we  terminated the network ${\NN}^{''}$ with $mn_{2}$ of $z_{2}$-type elements.
We now demonstrate how  to reduce the number of required frequency dependent elements of the type $z_2$.  
For this, we next examine some properties of the connection
 network  ${\bf N}^{''}$ by considering a further partition of the block-elements of $\BB$ as in 
(\ref{partition_step_3}), and define the square  matrix ${\bf S}_f$ of size $ (m+r)$ as:
\be\label{partition_step_3}
\BB
=
\left[\begin{array}{c|c}
 \BB_{11}             & \BB_{12} \\ \hline
 \BB_{21}             & \BB_{22} 
\end{array}\right]
=
\left[\begin{array}{c|cc}
 \BB_{11}             & \BB_{12}^{(1)} & \BB_{12}^{(2)} \\ \hline
 \BB_{21}^{(1)}  & \BB_{22}^{(1)} & \BB_{22}^{(2)} \\
 \BB_{21}^{(2)}  & \BB_{22}^{(3)} & \BB_{22}^{(4)}
\end{array}\right];
\quad
{\bf S}_f= \left[\begin{array}{cc}
 \BB_{11} & \BB_{12}^{(1)} \\
 \BB_{21}^{(1)} & \BB_{22}^{(1)} \end{array} \right],
\ee
where the newly defined block square matrices $\BB_{22}^{(1)}$ and  $\BB_{22}^{(4)}$ are of sizes $r$ and $(mn_2-r)$ respectively.  By using the first equation (\ref{define_B}) and (\ref{scatt:27}) and (\ref{partition_step_3}),  it then follows that
 \be\label{scatt:41}
 {\bf S}_f\tilde{{\bf S}}_f= \one_{m+r},
\text{ and }
\left[\begin{array}{ccc}
 \BB_{21}^{(2)} &  & \BB_{22}^{(3)} \end{array} \right]
 \left[\begin{array}{c}
 \tilde{\BB}_{21}^{(2)} \\
 \tilde{\BB}_{22}^{(3)} \end{array} \right] = \zero,
 \ee
 the last equation, in fact, implying that
 \be\label{scatt:43}
 \left[\begin{array}{cc}
 \BB_{21}^{(2)} & \BB_{22}^{(3)} \end{array} \right] \equiv \zero.
 \ee
If we use the block partition of $\bf B$ from (\ref{partition_step_3}), (\ref{scatt:43}), we obtain from (\ref{scatt:26a})
 \be\label{scatt:44}
 {\bf S}=\BB_{11} + \BB_{12}^{(1)}(z_{2}^{-1}\one_{r} - \BB_{22}^{(1)})^{-
 1}\BB_{21}^{(1)},
 \ee
 which implies that the matrix ${\bf S}_f$ can
 be interpreted as the discrete scattering matrix of an $(m + r)$-port
 $\NN^{'''}$, which results in a two-dimensional $m$-port $\NN$, described by
 the discrete scattering matrix ${\bf S}$, if the last $r$ ports of $\NN^{'''}$
 are terminated with $z_{2}$-type elements. Thus, we only need $r$ of $z_{2}$-type elements instead of
 $mn_{2}$ if we use $\NN^{'''}$ as a coupling network for the synthesis of the
 desired network $\NN$. 

We now collect several pieces of information already established above in a form important for the present purpose, and state the following.
\begin{proposition}
The univariate matrix $S_f$ constructed as in (\ref{scatt:41}) is a discrete lossless bounded matrix, and thus  can be veiwed as a lossless scattering matrix in the discrete domain.
\end{proposition}
{\em Proof:}
Since we already have (\ref{scatt:41}), we only need to show that $S_f$ is has no poles in $z_1$. First, note that since $a_0=a(z_1 , 0)$ is scattering Schur \cite[Theorem 11]{rob:Basu1},  it follows from  (\ref{scatt:7})-(\ref{scatt:13})  that $\HH _{ij}$'s are analytic in $|z_1|<1$. Next, note that in (\ref{scatt:34}) both ${\bf G}$ and ${\bf G}^{-1}$ are analytic in $|z_1 |<1$ by the construction of rational spectral factors of ${\bf K}$, and ${\bf D}$ is unimodular. Thus, from (\ref{scatt:34}) it follows that $\TT$, and  $\TT^{-1}$ are analytic in $|z_1 |<1$, and consequently, due to its definition as in (\ref{define_B})  the matrix  $\BB$ and its block-elements, including ${\bf S}_f$ as in (\ref{partition_step_3}), are rational in $z_{1}$ and analytic for $|z_{1}|<1$.  \hfill{Q.E.D}

 Thus, the synthesis problem is solved, since ${\bf S}_f$ can be realized with known
 methods as the discrete scattering matrix of a 1-D lossless network
 $\NN^{'''}$. The required number of $z_{1}\text{-type}$ and $z_{2}\text{-type}$ elements, however, depends 
on $r=\rank \KK$ and the degree of $\det {\bf S}_f$.  This issue is discussed next.

 \subsection{Minimality of synthesis}\label{minimal-synthesis}
 We now  show that a synthesis of ${\bf S}$ can be
 achieved according to the above scheme which uses only a minimum number of
 dynamic elements.  Specifically, if $\nu_{i} = m_{i} + {\deg}_{i}a$ for
 $i=1,2,$ then at least $\nu_{1}$ of $z_{1}$-type elements and $\nu_{2}$ of $z_{2}$-
 type elements are needed in the realization (cf. (\ref{scatt:3}) for definition of $m_1$ and $m_2$).  
 To demonstrate this we shall prove
 two key results.  The first result  establishes that $r=\rank \KK =\nu_{2}$.
 The second result
 we will need to prove is that $\det {\bf S}_f$ is an all-pass function whose numerator, in
 its irreducible rational form, is a polynomial in $z_{1}$ of degree $\nu_{1}$,
 and thus, it will immediately follow by invoking a standard result \cite{Dewilde_and_Belevitch} on synthesis 
of $1$-D  lossless scattering matrices that only $\nu_{1}$ of $z_{1}$ type elements 
are needed in the lossless realization of ${\bf S}_f$.

 \subsubsection{Number of $z_{1}\text{-type}$ elements}
 The technical arguments of this section rely largely on standard facts from
 the realization theory of $1$-D linear systems \cite{dar:Kail1}.  For this, we define the
 $(mn_{2} \times mn_{2})$ matrices:
\be\label{scatt:*}
\begin{split}
{\bf C} &=
\left[ {\bf H}_{21}, {\bf H}_{22}{\bf H}_{21}, \cdots, {\bf  H}^{mn_{2}-1}_{22}{\bf H}_{21} \right],
\\ 
{\bf O}^T&= 
\left[ {\bf H}_{12}^T, {\bf H}_{22}^T {\bf H}_{12}^T  \cdots, ({\bf H}_{12}^{mn_{2}-1})^T {\bf H}_{22}^T \right]
 \end{split}
\ee
as the controllability and observability matrices associated with the pairs
 $({\bf H}_{22},{\bf H}_{21})$ and $({\bf H}_{22},{\bf H}_{12})$
 respectively.  Note that  in the present context ${\bf C}={\bf C} (z_2)$ and ${\bf O}={\bf O} (z_2)$
are rational functions of $z_{1}$, and are independent of $z_2$.  We will need the following result.
 \begin{proposition}\label{scatt:*T1}
The rank $r$ of the matrix  ${\bf K}$  satisfying (\ref{scatt:29}) - (\ref{scatt:31}), 
 further satisfies the inequality
\be\label{propT1*}
 r=\rank {\bf K} \leq  m_{2} + \deg_{2} a. 
\ee
 \end{proposition}
 {\em Proof:}
 Let $q$ be the normal rank of ${\bf C}$.  Since for almost all $z_{1}$ on $|z_{1}|=1$ 
 the eigenvalues of ${\bf K}={\bf K}(z_{1})$ are in the open unit
 disc, it follows from the Lyapunov equation ({\ref{scatt:31}}) that ${\bf K}$
 can be expressed as
 \be\label{scatt:**}
 {\bf K}= \sum_{i=0}^{\infty}{\bf H}_{22}^{i}{\bf H}_{21}{\bf H}_{21}^{*}{\bf
 H}_{22}^{*i}\;,
 \ee
 where we have used the fact that on $|z|=1$, $\tilde{z} = z^{*}$, and thus,
 for each $i$, $j$ we also have $\tilde{\bf H}_{ij}={\bf H}^{*}_{ij}$.

 It can be shown that via the Hermite-type elementary row reduction algorithm
 \cite{dar:Kail1} that one can find a unimodular rational matrix ${\bf V} = {\bf
 V}(z_{1})$ such that
 \be\label{scatt:*3}
 {\bf V}{\bf C}  = \begin{array}{cc} \left[ \begin{array}{c}
 {\bf X} \\ {\bf 0}
 \end{array} \right] \begin{array}{c} q \\ mn_{2}-q \end{array} \end{array}.
 \ee
 Since by the Cayley Hamilton theorem, ${\bf H}_{22}^{i}$ for any $i \geq
 mn_{2}$ is a linear combination of ${\bf H}_{22}^{i}$'s with $i =
 0,1,\cdots,(mn_{2}-1)$, it follows from (\ref{scatt:*3}) and the structure of
 ${\bf C}$ in (\ref{scatt:*}) that for some ${\bf X}_{i}$'s of appropriate size
\be\label{scatt:*4}
 {\bf V}{\bf H}_{22}^{i}{\bf H}_{21}= \begin{array}{cc}\left[ \begin{array}{c}
 {\bf X}_{i} \\ {\bf 0}
 \end{array} \right] \begin{array}{c} q \\ mn_{2}-q \end{array} \end{array}.
 \ee
 Thus, in view of (\ref{scatt:**}) and (\ref{scatt:*4})  for almost all $z_{1}$ on $|z_{1}|=1$ we have
 \[ 
{\bf V}{\bf K}{\bf V}^{*}= \sum_{i=0}^{\infty} \begin{array}{cc} \left[
 \begin{array}{cc}
 {\bf X}_{i}{\bf X}_{i}^{*} & {\bf 0} \\
 {\bf 0} & {\bf 0} \end{array} \right] \begin{array}{c} q \\ mn_{2}-q
 \end{array} \end{array} .
\]
 Since ${\bf V}$ is invertible, this shows that $\rank {\bf K} \leq q$ for almost all $z_{1}$ on
 $|z_{1}|=1$.

 We next set out to show  that $q \leq m_2 +{\deg}_{2}a$. For this purpose, we note that
 for a generic (i.e., almost any), but  fixed value of $z_{1}=z_{10}$ on $|z_{1}|=1$ the synthesis of ${\bf S}$
 developed can be viewed as a realization of the $(m \times m)$
 univariate discrete lossless bounded matrix ${\bf W} = {\bf W}(z_{2})={\bf S}(z_{10},z_{2})$.
Also, ${\bf C}$ and ${\bf O}$ with $z_{1}=z_{10}$ are  controllability and
 observability matrices of this realization, and  by a standard $1$-D system theoretic result 
\cite{dar:Kail1} we have 
 \be\label{scatt:*6}
 \rank {\bf O}{\bf C} = \deg {\bf W},
 \ee
 where ${\deg}{\bf W}$ is the {\em McMillan} degree  of $\bf W$, i.e., the minimum number of $z_2\text{-type}$ elements needed in its realization. From (\ref{scatt:*6})  it follows by the use of Sylvester's inequality for ranks of
 matrix products that
 \be\label{scatt:*7}
\rank {\bf O} + \rank {\bf C} - mn_{2} \leq \rank {\bf O}{\bf  C}.
 \ee
 However, from (\ref{scatt:10})-(\ref{scatt:13}) it is routine to verify that $\rank {\bf O} = mn_{2}$
which together with (\ref{scatt:*6}) and (\ref{scatt:*7}) imply that 
\be\label{rank_C_less_deg_W}
q=\rank {\bf C} \leq \deg {\bf W}.  
\ee
It remains to show that for a generic value of $z_{1}=z_{10}$ on $|z_{1}|=1$  we have 
$ \deg {\bf W}= m_2+\deg_2 a$. For such a generic $z_{10}$  in view of (\ref{scatt:3}) we can express $\det {\bf W}$ in irreducible rational form as
\begin{gather}
\det {\bf W}=\det {\bf S}(z_{10}, z_{2})=\sigma_{1} z_{2} ^{m_2} \frac{\hat{\alpha}_1}{\alpha_1},
\\
\alpha_1(z_2)=a(z_{10},z_2),\; \sigma_1 = \sigma z_{10}^{m_1}=1,\;|\sigma_1|=1, \nonumber
\end{gather}
in which due to \cite[Theorem 11]{rob:Basu1}, $\alpha_1$ is scattering Schur and $\deg \alpha_1 =\deg_2 a(z_{10} , z_2)$. Furthermore, by invoking standard 1-D result \cite{Dewilde_and_Belevitch} it follows that the number of $z_2\text{-type}$ elements required for a lossless realization of ${\bf W}$ is equal to $m_2 + \deg\alpha_1$. Since $\deg_2 a$ is  equal to the value of $\deg \alpha_1$ for almost all $z_{10}$, we have shown that $\deg {\bf W} = m_2 + \deg _2 a$, which together with (\ref{rank_C_less_deg_W}) yields (\ref{propT1*}). Proposition \ref{scatt:*T1} is thus proved.\hfill{Q.E.D}
 \subsubsection{Number of $z_{2}\text{-type}$ elements}
 We will first need  to establish the following result.
 \begin{proposition}\label{scatt:*T2}
 The determinant of the matrix ${\bf S}_f$ describing the connection network ${\bf N}^{'''}$ and defined in (\ref{partition_step_3}) can be expressed as 
\be\label{z_1-degree} 
\det{\bf S}_f = (-1)^{r} \sigma\frac{\hat{a}_0}{a_0} z_{1}^{m_1},
 \ee
in irreducible rational form, in which $a_0 =a_0 (z_{1})= a(z_{1}, 0)$ is as in (\ref{scatt:6}).
 \end{proposition}
 {\em Proof}:
 Algebraic manipulation with second equation (\ref{partition_step_3}), (\ref{scatt:44}), (\ref{scatt:3}) and the use of a well-known  identity%
\footnote{Assuming  ${\bf M}_{22}^{-1}$ to exist,
$\det
\left[ \begin{array}{cc} 
{\bf M}_{11} & {\bf M}_{12} \\
{\bf M}_{21} & {\bf M}_{22} 
\end{array} \right] 
=\det{\bf M}_{22}
\det({\bf M}_{11}-{\bf M}_{12}{\bf M}_{22}^{-1}{\bf M}_{21})
$.
} 
pertaining to determinant of block matrices yield
\be\label{tricky_part}
\sigma {\bf z}^{\bf m}\hat{a}
\det \left[\tilde{\bf S}_f - z_2 \begin{pmatrix} \zero & \zero \\ \zero & \one_{r}\end{pmatrix}\right]
=
a\det [ \tilde{\bf B}_{22}^{(1)} - z_2 \one_{r}].
\ee
We  view (\ref{tricky_part}) as a polynomial in $z_2$, and appeal to the fact that  due to the scattering Schur property of $a$, the bivariate  polynomials $a$ and $\hat{a}$ are relatively prime. We denote the leading coefficient of $\hat{a}$ viewed as a polynomial in $z_2$, by $\alpha$ as in (\ref{define_alpha}). We then obtain from (\ref{tricky_part}) by invoking the relative primeness property of $a$ and $\hat{a}$ just mentioned
\be\label{define_alpha}
\det [ \tilde{\bf B}_{22}^{(1)} - z_2 \one_{r}]
=
(-1)^r z_2 ^{m_2}\frac{\hat{a}}{\alpha},\quad \alpha =z_1 ^{n_1}\tilde{a_0}
\ee
and thus also from (\ref{tricky_part})
\be
\det \left[\tilde{\bf S}_f - z_2 \begin{pmatrix} \zero & \zero \\ \zero & \one_{r}\end{pmatrix}\right]
=
(-1)^{r} \sigma^* \frac{a}{\tilde{a}_0} z_{1}^{-(n_1 + m_1)}.
\ee
 Equation (\ref{z_1-degree}) is then obtained by substituting $z_{2}=0$ in this latter expression thus establishing the proposition.
The observation that due to \cite[Theorem 11]{rob:Basu1} the polynomial $a_0$ is scattering Schur, and thus (\ref{z_1-degree})  is in irreducible rational form, completes the desired proof of the Proposition. \hfill{Q.E.D}

 Due to Theorems \ref{scatt:*T1}, the 1-D coupling network $\NN^{'''}$,
 described by ${\bf S}_f$, can be realized with $v_{1}^{'}$ elements of
 $z_{1}$-type.  After terminating the last $v_{2}^{'}$ ports of $\NN^{'''}$
 with $z_{2}$-type elements we obtain the desired 2-D network $\NN$,
 described by ${\bf S}$.  Thus we have obtained a minimal realization.  To further illustrate
 this, suppose there exists a realization of ${\bf S}$ with less than $m_{2}^{'}$
 of $z_{2}$-dependent elements.  If we replace the $z_{1}$-type elements
 in this network by {\em constant} one-ports with element $z_{10}$, where
 $|z_{10}| = 1$ and ${\deg}_{2} \; g(z_{10},z_{2}) = n_{2}^{'}$, we would
 have a realization of $\WW(z_{2}) = {\bf S}(z_{10},z_{2})$.  But in the proof of
 Theorem 6 we have seen that every realization of $\WW(z_{2})$ requires at
 least $m_{2}^{'}$ of $z_{2}$-dependent elements, which contradicts the
 assumption.  With similar arguments it can be shown that $m_{1}^{'}$ is the
 minimal number of $z_{1}$-dependent elements in any realization of ${\bf S}$.

\section{Three and Higher Dimensional Synthesis}\label{scatt:S24}
 Having demonstrated that a rational discrete para-unitary
 matrix in two variables can be viewed
 as the scattering matrix of a network consisting exclusively of more
 elementary passive building blocks (e.g., two types of delays along with
 memoryless lossless building blocks) we now consider the issue for three or
 larger number of  variables. It will be shown that such a synthesis is
 provably infeasible for scalar as well as for matrix valued rational
 para-unitary functions, holomorphic in the right half poly-plane (or poly-disc
 in the discrete case) in the generic instance.  However, in the special case,
 when a three variable scalar all-pass rational function, of which the denominator polynomial is  {\em linear} in each of the
 three variables is considered, the desired synthesis can be carried out.
 This latter aspect is considered in the next section by following developments in \cite{kummert-linear-3D}.  
Our exposition  will be for continuous systems, whereas obvious discrete analogs of the main results
 indeed hold and can be derived in a similar manner.

 \subsection{Higher dimensions $(k \geq 3)$ }\label{H-dim}
  For the purpose of this section we shall consider the scalar rational function $s =  s({\bf p})$ in $k$-variables which is a lossless reflectance i.e., it is holomorphic in
 $\Re{\bf p} >0$, and furthermore $ss_{*} = 1$ (we will also refer to such functions as $k\text{-D}$ all-pass functions).  We will restrict ourselves to specific values of $k$, e.g., $k=3$ later as the context dictates. It is known \cite{stb:Fett2} that any such $s({\bf p})$ can be expressed in irreducible rational form as
 \be\label{scatt:120}
 s = \sigma \frac{g_{*}}{g},
 \ee
 where $g=g({\bf p})$ is a scattering Hurwitz polynomial, and $\sigma$ is a
 unimodular constant, i.e., $|\sigma|=1$. We presently assume that $g$ is linear in  $p_{1}$ i.e., we
 can write
 \be\label{scatt:121}
 g = g_{0} + g_{1}p_{1},
 \ee
 where $g_{0}=g_{0}({\bf p}')$ and $g_{1}=g_{1}({\bf p}')$ are
 polynomials in $(k-1)$ variables ${\bf p}' = (p_{2}, p_{3}, \cdots, p_{k})$.  Next, we assume 
 that there exists an $(m+1) \times (m+1)$ lossless scattering matrix ${\bf S}
 = {\bf S}({\bf p}')$ in $k-1$ variables ${\bf p}'$, synthesizable in
 terms of inductive and capacitive elements of $(k-1)$ different types such
 that when $m$ of its $m+1$ ports are terminated in inductive and/or
 capacitive elements of type $p_{i}; \; i =2,3, \cdots, k$ then we obtain the
 reflectance $s$.  It can be shown by essentially following the technique pursued in Section \ref{first-step} that the problem is equivalent to requiring the existence of rational matrices ${\bf S}_{ij}$ of appropriate sizes such that\footnote{%
Here the notation $z_1$ does not correspond to discrete domain, but is used as a convenient notation.
} 
 \be\label{scatt:122}
 s = S_{11} + {\bf S}_{12}(z_{1}^{-1}\one_{m} - {\bf S}_{22})^{-1}{\bf  S}_{21},
\text{ where }
z_{1} = \frac{1-p_{1}}{1+p_{1}},
 \ee
 and
 \be\label{scatt:123}
{\bf  S}=
  \left[\begin{array}{cc}
 S_{11} & {\bf S}_{12} \\
 {\bf S}_{21} & {\bf S}_{22}
 \end{array} \right] 
\ee
 is a paraunitary scattering matrix of an $(m+1)$ port involving the variables ${\bf p}^{'} = (p_{2},p_{3}, \cdots, p_{k})$.
Substituting $p_{1} = 1$, i.e., $z_{1} =0$ in (\ref{scatt:120}) and ({\ref{scatt:122}}) it
 would then follow via the use of (\ref{scatt:121}) that
 \be\label{scatt:125}
 S_{11} = \sigma\frac{a_{0*}}{a_{1}},
 \text{ where } 
 a_{0} = g_{0} - g_{1},\text{ and } a_{1} = g_{0} + g_{1}.
 \ee
 Straightforward algebraic manipulation with (\ref{scatt:120}) to  (\ref{scatt:125}) yields 
 \be\label{scatt:127}
 {\bf S}_{12}({\bf I}_{m}-z_{1}{\bf S}_{22})^{-1}{\bf S}_{21} =
 \frac{a_{1}a_{1*} -a_{0}a_{0*}}{a_{1}(a_{1} + z_{1}a_{0})},
 \ee
 in which the further substitution $p_{1}=1$, i.e., $z_{1} =0$  produces
 \be\label{scatt:128}
a_{1}^{2} {\bf S}_{12}{\bf S}_{21} = a_{1}a_{1*} - a_{0}a_{0*}.
 \ee
 Since ${\bf S}$ is assumed to be a lossless scattering matrix,
 we have ${\bf S}{\bf S}_{*}={\bf S}_{*}{\bf S} = {\bf I}_{m+1}$  it thus follows from (\ref{scatt:123}) that
 \be\label{scatt:129}
 S_{11}S_{11*} + {\bf S}_{12}{\bf S}_{12*} = 1; 
\quad
S_{11*}S_{11} + {\bf  S}_{21*}{\bf S}_{21} = 1.
 \ee
 Since $S_{11}$ is a scalar, thus $S_{11}S_{11*} = S_{11*}S_{11}$, we have from
 (\ref{scatt:129}) and (\ref{scatt:125})
 \be\label{scatt:130}
 a_{1}a_{1*}{\bf S}_{12}{\bf S}_{12*} = a_{1}a_{1*}{\bf S}_{21*}{\bf S}_{21} =
 a_{1}a_{1*} - a_{0}a_{0*}.
 \ee
 Next, we consider the $(m \times 1)$ rational vector in $k-1$ variables ${\bf
 p}'$ defined by ${\bf N} = a_{1}{\bf S}_{21} - a_{1*}{\bf S}_{12*}$.  It
 is then straightforward to verify by using (\ref{scatt:128}) and
 (\ref{scatt:130}) that ${\bf N}_{*}{\bf N} = (a_{1}{\bf S}_{21} - a_{1*}{\bf
 S}_{12*})_{*}(a_{1}{\bf S}_{21} - a_{1*}{\bf S}_{12*}) \equiv 0$, implying 
$| {\bf N}( j \mbox{\boldmath $\omega$}^{'}|^2=0$ for all  real $(k-1)$ tuples  
$\mbox{\boldmath $\omega$}^{'}$, and thus ${\bf N} = {\bf N}({\bf p}^{'}) \equiv 0$.

 Consequently, invoking the definition of ${\bf N}$ we may assert that
 $a_{1}{\bf S}_{21} = a_{1*}{\bf S}_{12*}$.  If we define a $(m \times 1)$
 column vector
 \be\label{scatt:131}
 {\bf P} = a_{1}{\bf S}_{21} = a_{1*}{\bf S}_{12*}
 \ee
 then we can further assert that ${\bf P}$ is, in fact, a polynomial column
 vector.  To see this, note first that ${\bf S}_{21}$ and ${\bf S}_{12}$ being
 sub-matrices of the lossless scattering matrix (cf. (\ref{scatt:123})), are
 holomorphic in $\Re{\bf p}'>0$.  Thus, due to a result in \cite{stb:Fett2} ${\bf P}=a_{1}{\bf S}_{21}$ is either a
 polynomial or has singularities in $\Re{\bf p}'<0$.  On the other hand, $a_{1}{\bf S}_{12}$ must be holomorphic in $\Re{\bf p}'>0$, and thus,  $a_{1*}{\bf S}_{12*} = (a_{1}{\bf S}_{12})_{*}$ must be holomorphic
 in $\Re{\bf p}'<0$.  It then follows from (\ref{scatt:131}) that ${\bf P}$ is a
 polynomial column vector of size $(m \times 1)$.

 We can now state the following result which is crucial for the developments of the present section.

 \begin{theorem}\label{scatt:T7}
 A necessary condition for the lossless scattering function $s = s({\bf p})$
 in $k$-variables as described  in (\ref{scatt:120}) and (\ref{scatt:121}) to admit a 
 lossless synthesis is that there exists an $(m \times 1)$ polynomial vector ${\bf P}$ satisfying
 \be\label{scatt:132}
 {\bf P}_{*}{\bf P} = a_{1}a_{1*} - a_{0}a_{0*}
 \ee
 where,  $a_0$ and $a_1$ are as described in (\ref{scatt:125}). 
 \end{theorem}
 {\em Proof:}
 We only need to observe that (\ref{scatt:128}) and  (\ref{scatt:131})  together yield
 (\ref{scatt:132}).

 \subsubsection{3-D all-pass functions of degree one}\label{3-D_all-pass}
 We will now show that a  synthesis for $s$ is feasible if the number of variables  $k=3$ 
 and  $g$ is linear in each of the variables, i.e., ${\deg}_{i}g = 1$ for $i =  1,2,3$.   
The obvious first step is to show that given $s$, or equivalently $g$, or $a_0$ and $a_1$, a polynomial vector ${\bf P}$ 
 satisfying (\ref{scatt:132}) can always be found. 
For this, let the polynomials $P_{i}$; $i=1$ to $m$ be the elements of ${\bf P}$.  Then
considering ${\bf p}^{'}=j \mbox{\boldmath $\omega$}^{'}$, where $\mbox{\boldmath
 $\omega$}^{'}=(\omega_{2},\omega_{3}, \cdots, \omega_{k})$ is a $(k-1)$-tuple
 of real numbers,  (\ref{scatt:132}) can be equivalently written as
 \be\label{scatt:133}
 |a_{1}(j \mbox{\boldmath $\omega$}^{'})|^{2}-|a_{0}(j \mbox{\boldmath
 $\omega$}^{'})|^{2} = \sum_{i=1}^{m}|P_{i}(j \mbox{\boldmath
 $\omega$}^{'})|^{2}.
 \ee
 Next, since $g = g_{0} +  g_{1}p_{1}$ is scattering Hurwitz polynomial, the rational function $g_{0}/g_{1}$ is a
 positive function (cf. \cite{stb:Fett2,rob:Basu1} for more detail) and, thus, $a_{0}/a_{1} =
 (g_{0}-g_{1})/(g_{0}+g_{1})$ is a bounded function, i.e., $|a_{0}/a_{1}| \leq
 1$ in $\Re{\bf p}'> 0$, which in turn implies that the left hand side of
 (\ref{scatt:133}) must be a nonnegative definite polynomial in $k-1$ real variables
 $\boldsymbol\omega '$.  

Recall the classical result of Hilbert on ternary quartic forms, which says 
that any  positive definite real polynomial in two-variables, say $\omega_{2}$, $\omega_{3}$, with {\em total degree not exceeding four}
 can be expressed as a sum of squares of  three real polynomials
 \cite{S_fact:hilb} (cf. Table \ref{table_SOS} entry $d=4$, $n=2$). This last mentioned result does indeed apply in the present
 situation since $|a_{0}(j \mbox{\boldmath $\omega$}^{'})|^{2} - |a_{1}(j
 \mbox{\boldmath $\omega$}^{'})|^{2}$ is a polynomial in $\omega_{2}$,
 $\omega_{3}$, with total degree not exceeding four because ${\rm
 deg}_{i}a_{0} \leq 1$, ${\deg}_{i}a_{1} \leq 1$, thus confirming the existence of the polynomials $P_i$, $i=1\text{ to } m$ with $m=3$ in (\ref{scatt:133}). Note that this argument is critical and cannot  be extended to a  larger number of variables (i.e., $k>3$), or if $g$ is super-linear in any of the variables. 

 Having obtained the polynomial\footnote{Note that ${\bf P}$ may not have real coefficients.}
 ${\bf P}$, we will demonstrate in the rest of this section that the choice of ${\bf S}({\bf p}')$ in (\ref{scatt:138}) provides a bounded scattering matrix that indeed describes the connection network, which when terminated by $p_1$-type capacitances, yields the desired synthesis for $s$.
 \be\label{scatt:138}
 S_{11} = \sigma \frac{a_{0*}}{a_{1}},\;\;
 {\bf S}_{21} = \sigma \frac{{\bf P}}{a_{1}},\;\;
 {\bf S}_{12} = \frac{{\bf P}_{*}}{a_{1}},\;\;
 {\bf S}_{22} = \frac{{\bf P}{\bf P}_{*}  - a_{1}(a_{1*} + a_{0})\one_m }{a_{1}(a_{1} + a_{0*})}  .
 \ee
We need to show that each of the above expressions are holomorphic in $\Re {\bf p}' >0$.  
For this, first note that $a_1 = g(1,{\bf p}')$ is a scattering Hurwitz polynomials, 
and thus we have $a_1 \neq 0$ for $\Re {\bf p}' >0$.

 Next, it  follows from  (\ref{scatt:120}),  (\ref{scatt:121}) and (\ref{scatt:125}) 
that  $s(1,{\bf p}^{'}) = a_{0*}/a_{1}$ is a bounded function in irreducible
 rational form (we have used the relative primeness of the scattering Hurwitz
 polynomial $g(1,{\bf p}^{'})$ with its own para-conjugate here). 
 Next, $(a_{1} + a_{0*})$ is the denominator of the
 irreducible positive function obtained by considering the bilinear transform
 of $s(1,{\bf p}')$.  Thus, $a_{1} + a_{0*} \neq 0$ in $\Re{\bf p}' >0$.

 Finally, it follows by straightforward algebraic manipulations
 from (\ref{scatt:123}) and (\ref{scatt:138}) that ${\bf S}{\bf S}_{*} = \one_{m+1}$ i.e., ${\bf S}$ is paraunitary.

 It only remains to show that ${\bf S}$
 indeed satisfies (\ref{scatt:122}). While this step involves purely algebraic manipulation, due to its intricate nature we provide some detail. First note that $(a_{1}a_{1*} - a_{0}a_{0*})$ is the 
$p_{1}$-{\em resultant}  between the pair of polynomials $g$ and $g_*$, which cannot be identically zero due to the relative primeness of  $g$ and $g_*$ inherited from the scattering Hurwitz property of $g$.
 Thus, from (\ref{scatt:132}) it follows that ${\bf P} \not\equiv 0$, and consequently, from
 (\ref{scatt:131}) ${\bf S}_{12} \not\equiv 0$.  Let ${\bf Q}$ be a square matrix\footnote{%
Note that ${\bf Q}$ has no role in the final result and is an tool for convenience of calculation.
} 
of full normal rank whose first row is ${\bf S}_{12}$, i.e., $ {\bf S}_{12} = {\bf e}{\bf Q}$, where ${\bf e}=[1,0\cdots0]$.
It then follows that 
\be\label{scatt:143}
\begin{split} 
{\bf S}_{12}(z_{1}^{-1}\one_{m} - {\bf S}_{22})^{-1}{\bf S}_{21}  &  =  {\bf S}_{12}{\bf Q}^{-1}(z_{1}^{-1}\one_{m} - {\bf Q}{\bf S}_{22}{\bf Q}^{-1})^{-1}{\bf Q} {\bf S}_{21} \\
& = {\bf e} \left[ \left( z_{1}^{-1} +
 \frac{a_{1*} + a_{0}}{a_{1} + a_{0*}} \right) \one_{m} + \frac{{\bf Q}{\bf  P}{\bf e}}{a_{1} + a_{0*}} \right]^{-1}
{\bf Q}{\bf S}_{21} \\
& = \left[ z_{1}^{-1} + \frac{a_{1*} + a_{0}}{a_{1} + a_{0*}} - \frac{{\bf e}{\bf  Q}{\bf P}}{a_{1} + a_{0*}} \right]^{-1} 
{\bf e} {\bf Q} {\bf S}_{21} \\
 & = \sigma \left( z_{1}^{-1} + \frac{a_{0}}{a_{1}} \right)^{-1}
 \frac{a_{1}a_{1*} -a_{0}a_{0*}}{a_{1}^{2}},
 \end{split}
 \ee
 where the 3rd and 4th equations in (\ref{scatt:138}),  have been used
 in deriving the second equality, whereas the fact that inverse of the lower
 triangular Toeplitz matrix so obtained inside the square brackets is also lower triangular with inverted
 diagonal elements has been used in the third equality. We then note in view of ${\bf S}_{12}={\bf e}{\bf Q}$ and (\ref{scatt:131}) that we may write 
${\bf P}{\bf P}_{*}=a_{1} {\bf e}{\bf Q}{\bf P}=a_{1}^{2} {\bf e}{\bf Q}{\bf S}_{21}$, 
and substitute for ${\bf P}{\bf P}_*$ from equation (\ref{scatt:132}), which  yields the  last equality.

By adding the first equation (\ref{scatt:138}) and  (\ref{scatt:143}), by invoking  (\ref{scatt:122}), it is easily verified that  (\ref{scatt:120}) holds, thus confirming that the connection network described by ${\bf S}$, when terminated in $p_1$-type capacitances indeed result in the desired all-pass function $s$. Since  ${\bf S}={\bf S}(p_1,p_2)$ has been shown to be a bivariate bounded scattering matrix, and thus  can be synthesized via the technique expounded in Section \ref{scatt:S21}, a synthesis of the all-pass function $s$ of degree one each of its three variables has been now conclusively demonstrated. 

{\flushleft\bf Remark:}
While the above discussion treats only a scalar all-pass transfer function $s$, it may be interesting to consider the a multiport analog of the problem. More specifically, one may consider a rational lossless bounded matrix ${\bf S}$ (i.e.,   ${\bf S}$ is holomorphic in $\Re {\bf p} >0$, $\one - {\bf S}_* (j\boldsymbol\omega) {\bf S} (j\boldsymbol\omega) \geq 0$ for all real $\boldsymbol\omega$, and ${\bf S}_*{\bf S}=\one$) of `low degree' (here the precise definition of degree may need to be clarified) and study its synthesizability.

{\flushleft\bf Remark:}
 We note finally that the present discussion also provides a proof of the fact that if  for $s=s({\bf p})$ as given in (\ref{scatt:120}) and (\ref{scatt:121}), a ${\bf P}$ satisfying the condition  (\ref{scatt:132}) of Theorem \ref{scatt:T7} can be found then one can find a $k-1$ variable lossless scattering matrix, say ${\bf S}_f ({\bf p}')$, whose $m$ terminals terminated in capacitances provide a synthesis for the all-pass function $s$. This statement can, in a sense be viewed as a sufficiency part of Theorem \ref{scatt:T7}, but  condition 
 (\ref{scatt:132}) cannot, in general, be satisfied for $k>3$, nor can it be satisfied when $k=3$ and $\deg_{i} g >1$ for any $i\neq1$. A concrete demonstration of this latter fact is the content of the section that immediately follows.

\subsubsection{Infeasibility of synthesis for higher order all-pass functions}
 In the rest of the present section we shall show via a counterexample that
 (\ref{scatt:132}) is not satisfied for a higher order all-pass
 function $s = s({\bf p})$ in three variables.    

In Section \ref{3-D_all-pass} we have argued that the left hand side of (\ref{scatt:133}) is a nonnegative definite 
polynomial function of the $k-1$ real variables  $\boldsymbol\omega '$.  
The fact that such polynomials cannot be
 expressed as a sum of squares, as required by Theorem \ref{scatt:T7},  is well known in view of Hilbert's celebrated
 result \cite{S_est:hilbert,S_fact:hilb}. However, since in our case $a_{0} = g_{0}-g_{1}$ and $a_{1}
 = g_{0}+g_{1}$, where $g_{0}$ and $g_{1}$ are the coefficients of a
 scattering Hurwitz polynomial, the left hand side of (\ref{scatt:133}) is not
 exactly an arbitrary nonnegative definite polynomial, but it arises in a
 somewhat special way, we need to examine the issue in greater detail.  As
 shown in \cite{kummert-conterexample}  there exist examples of $a_{0}$ and $a_{1}$ satisfying the above
 requirements which are such that the left hand side of (\ref{scatt:133})
 still cannot be expressed as a sum of squares of polynomials in $\boldsymbol\omega '$.

  For the stated purpose we consider the set of polynomials $b$,
 $c$, and $d$ as follows:
 \be\label{scatt:134}
 b = 1 + (p_{2} + p_{3})^{2} + \frac{1}{2}p_{2}^{2}p_{3}^{3}
 \ee
 \be\label{scatt:135}
 c = \frac{1}{\sqrt{7}}(p_{2} + p_{3})(4 + 3p_{2}p_{3})
 \ee
 \be\label{scatt:136}
 d = \sqrt{\frac{2}{7}}p_{2} + p_{2}^{2} + \sqrt{\frac{2}{7}}(1-
 p_{2}^{2})p_{3} + (1-\sqrt{\frac{2}{7}}p_{2} + \frac{1}{2}p_{2}^{2})p_{3}^{2}
 \ee
 \be\label{scatt:137}
 e = b + c
 \ee
Now, as before in (\ref{scatt:121}) we consider the polynomials 
\[
g=g_0+p_1 g_1, \text{ where } g_0 = e+d \text{ and } g_1=e-d,
\]
where the polynomials $d$ and $e$ are as specified in (\ref{scatt:136}) and (\ref{scatt:137}). One can indeed show that the polynomial $g$ so constructed  is scattering Hurwitz by using standard test procedures for testing stable multivariable  polynomials. Furthermore, in view of (\ref{scatt:125}) we then have $a_0=2d$ and $a_1 =2e$, and it can  be routinely verified that 
$|a_{0}(j\boldsymbol\omega ')|^{2} - |a_{1}(j\boldsymbol\omega ')|^{2} = 
4(|e(j\boldsymbol\omega ')|^{2} - |d(j\boldsymbol\omega ')|^{2})$
 cannot be expressed as a sum of squares of polynomials, if $e$ and $d$ are as specified.

 \section{Dissipative 2-D scattering Synthesis}\label{S22}
 We now consider the issue of synthesizability of $2$-D lossy or
 dissipative scattering bounded matrices.  More specifically, given a bounded
 rational matrix ${\bf S}_{11}={\bf S}_{11}({\bf p})$ in two variables ${\bf p}=(p_{1},p_{2})$
 we wish to view ${\bf S}_{11}$ as
 the scattering matrix associated with an electrical network consisting of
 resistive elements as well as of inductive and capacitive elements of $p_{1}$
 and $p_{2}$ types.  Alternatively, in the discrete domain given a discrete
 bounded rational matrix $\mbox{\boldmath $\Sigma$}_{11}$ in two variables 
 ${\bf Z}=(z_{1}, z_{2})$ i.e., if $\mbox{\boldmath $\Sigma$}_{11}$ is
 holomorphic and $\one - \boldsymbol\Sigma_{11}^{*}\boldsymbol\Sigma_{11} \geq 0$ 
 in $|{\bf z}| < 1$ then we wish to view $\mbox{\boldmath
 $\Sigma$}_{11}$ as the transfer function of a system consisting of fully
 absorbing ports as well as of $z_{1}$ and $z_{2}$ type shift elements (i.e.,
 delays in digital filter terminology) interconnected by a memoryless (i.e.,
 constant) network having an unitary transfer function matrix.  In the discrete
 case, the fully absorbing ports are those which are terminated in elements
 from which no signal is reflected, and correspond to resistive ports under
 proper matching conditions in the continuous case.

 In Section \ref{scatt:S21} the problem has already been solved
 in the special case when ${\bf S}_{11}$ or $\mbox{\boldmath $\Sigma$}_{11}$,
 is in addition, lossless.  Given that a synthesis scheme for lossless bounded
 matrices are available, the more general problem that we now address can,
 from a mathematical standpoint, be viewed as a problem of {\em unitary
 dilation} of ${\bf S}_{11}$ or $\mbox{\boldmath $\Sigma$}_{11}$.  What we
 wish to find is a bivariate {\em lossless} bounded rational matrix, of which
 the prescribed $\boldsymbol\Sigma_{11}$, or ${\bf S}_{11}$ is a sub-matrix.  This is
 also referred to as the {\em unitary bordering} or the {\em embedding}
 problem in network theoretic literature.  It is not hard to see, and is well
 known in the 1-D case, that the problem is solvable via the use of spectral
 factorability of rational para-hermitian matrices nonnegative on the imaginary
 axis (or unit circle in the discrete case).  The lack of such results in
 multidimensional context has been a bottleneck in demonstrating
 synthesizability of bounded matrices.  In 2-D by
 using results from real algebraic geometry  \cite{S_est:hilbert, S_fact:land} and their consequences on spectral
 factorability, as stated in Theorems \ref{S_fact:T5} or \ref{discrete_spectral_factor} to follow, the problem can at least partially be solved.  
However, more specific questions of a more detailed  nature such as the minimality of the
 synthesis in terms of number of dynamic elements, or the minimal number of
 resistors (or fully absorbing ports in the discrete case) needed in the
 realization are presently not fully known.

 Our treatment will be mainly in the continuous domain.  We will follow techniques reported in \cite{S_fact:basu} for proofs of more general results reported in \cite{S_fact:kumm2}. A discrete domain counterpart is then obatined by a convenient application of double bilinear transform (\ref{S_fact:a*}).

\subsection{Continuous Domain considerations}
 \begin{lemma}\label{scatt:L1}
 Let ${\bf S}_{11}={\bf S}_{11}(\bf p)$ be a $(m \times m)$ bounded  rational
 matrix of two variables
 ${\bf p}=(p_{1},p_{2})$. Then it is always possible to decompose the
 matrix ${\bf 1}_{m} - {\bf S}_{11*}{\bf S}_{11}$ as
 \be\label{scatt:94}
 {\bf 1}_{m} - {\bf S}_{11*}{\bf S}_{11} =
 \frac{1}{g_{*}ge_{*}e}{\bf R}_{*}{\bf R},
 \ee
where $g$ is a scattering Hurwitz polynomial, $e=e(p_{1})$ is a
 univariate polynomial nonzero in Re $p_{1}>0$, and
 ${\bf R}$ is a $(2r \times m)$ polynomial matrix, where
 $r = \rank (\one_{m} - {\bf S}_{11*}{\bf S}_{11})$.
 \end{lemma}
 {\em Proof:}
 Then (cf. \cite[Theorem xxx]{stb:Fett2}) ${\bf S}_{11}$ can be expressed as ${\bf S}_{11}={\bf
 P}_{11}/g$, where ${\bf P}_{11}$ is a $(m \times m)$  polynomial matrix and
 $g$ is a bivariate scattering Hurwitz polynomial in ${\bf p}=(p_{1},p_{2})$.  Next, we write
 \be\label{scatt:95}
 {\bf 1}_{m} - {\bf S}_{11*}{\bf S}_{11} = \frac{\boldsymbol\Phi}{g_{*}g}, 
 \text{ where }
 \boldsymbol\Phi = gg_{*}{\bf I}_{m} - {\bf P}_{11*}{\bf P}_{11}.
 \ee
Then we have $\boldsymbol\Phi(j\boldsymbol\omega) \geq 0$ for all
 $\boldsymbol\omega$ except possibly for at most a
 finite number of points $\boldsymbol\omega$ for which $g = 0$ (cf. \cite[Theorem xxx]{stb:Fett2}).
 Hence, by continuity argument $\boldsymbol\Phi \geq 0$ for all ${\bf p}=j{\boldsymbol\omega}$.
 Furthermore, $\boldsymbol\Phi$ is obviously a para-Hermitian polynomial matrix, i.e., $\boldsymbol\Phi_{*} = \boldsymbol\Phi$.  Thus, due to Theorem \ref{S_fact:T5} $\boldsymbol\Phi$ can be  spectrally factored as
 \be\label{scatt:96}
 \boldsymbol\Phi = \frac{{\bf R}_{*}{\bf R}}{e_{*}e}
 \ee
where $e$ is a univariate polynomial in $p_{1}$, nonzero in Re $p_{1} > 0$,
 and ${\bf R}$ is a $(2r \times m)$ polynomial matrix with 
$r =$ rank $\boldsymbol\Phi= \rank ({\bf 1}_{n}-{\bf S}_{11*}{\bf S}_{11})$.
The desired proof is completed by substituting $\boldsymbol\Phi$ from  (\ref{scatt:96}) in the first equation (\ref{scatt:95}).\hfill{Q.E.D}
 \begin{theorem}\label{scatt:T6}
 Let  ${\bf S}_{11}={\bf S}_{11}(\bf p)$ be an $(m \times n)$ bounded rational matrix, which is  function of  two variables
 ${\bf p}=(p_{1},p_{2})$. Then there exists a {\em lossless} bounded rational matrix 
${\bf S}={\bf  S}(\bf p)$ of size $(p+2r) \times (p+2r)$, where $p = {\rm max}(m,n)$ such that ${\bf S}$ can be partitioned as
 \be\label{scatt:matrix}
 {\bf S} = \left[ \begin{array}{cc}
 {\bf S}_{11} & {\bf S}_{12} \\
 {\bf S}_{21} & {\bf S}_{22} \end{array} \right]
 \ee
 where $r = {\rank}\;(\one_n - {\bf S}_{11*}{\bf S}_{11})$.
 \end{theorem}
 {\em Proof}:
 First we assume that $m=n$, i.e., ${\bf S}_{11}$ is square.  Since ${\bf
 S}_{11}$ is a bounded matrix, it can be expressed as ${\bf S}_{11} = {\bf
 P}_{11}/g$, where $g$ is a scattering Hurwitz
 polynomial and ${\bf P}_{11}$ is a polynomial matrix of size $(m \times
 m)$.  Next, we consider the decomposition indicated in Lemma {\ref{scatt:L1}}
 to obtain $g$, $e$, and the $(2r\times m)$ matrix ${\bf R}$ associated with ${\bf S}_{11}$, and set
 ${\bf S}_{21} = {\bf R}/ge$, so that we have
\be\label{firstentry-S}
{\bf S}_{11*}{\bf S}_{11}+{\bf S}_{21*}{\bf S}_{21}=\one_m ,
\ee
and define
\be\label{scatt:*1}
 {\bf S}_{12} = -\frac{\alpha_*}{\alpha}
  [(\one_m + {\bf S}_{11})(\one_m + {\bf S}_{11*})^{-1}{\bf S}_{21*} ],
 \ee
\be\label{scatt:*2}
 {\bf S}_{22} = \frac{\alpha_*}{\alpha}
 [\one_{2r} - {\bf S}_{21}(\one_m+ {\bf S}_{11*})^{-1}{\bf S}_{21*} ],
 \ee
 where $\alpha$ is a polynomial to be specified later in course of
 the proof.  It can then be routinely verified by algebraic manipulations using (\ref{firstentry-S}), (\ref{scatt:*1}), (\ref{scatt:*2}) that 
\[
{\bf S}_{11}{\bf S}_{11*}+{\bf S}_{12}{\bf S}_{12*}=\one_m, \quad
{\bf S}_{11}{\bf S}_{11*}+{\bf S}_{22}{\bf S}_{12*}={\bf 0}, \quad
{\bf S}_{21}{\bf S}_{21*}+{\bf S}_{22}{\bf S}_{22*}=\one_{2r},
\]
which along with  (\ref{firstentry-S}) shows that the square matrix ${\bf S}$ as in (\ref{scatt:matrix}),  in turn, satisfies ${\bf S}_{*}{\bf S}={\bf S}{\bf S}_{*} = \one_{2r+m}$, 
independent of the choice of $\alpha$.

 We need to show that by proper choice of $\alpha$ it is  possible to make ${\bf S}$ holomorphic in $\Re {\bf p} > 0$. Since ${\bf S}_{11}$ is a bounded matrix,  
$(\one_m-{\bf S}_{11})(\one_m+{\bf S}_{11})^{-1}$ is a positive matrix\footnote{%
A rational matrix ${\bf Z}={\bf Z}({\bf p})$ is called positive if  ${\bf Z}$ is holomorphic in $\Re {\bf p}>0$ and 
${\bf Z}^*(j\boldsymbol\omega)+{\bf Z}(j\boldsymbol\omega)$ is nonnegative definite for real $\boldsymbol\omega$, wherever ${\bf Z}(j\boldsymbol\omega)$ is well defined. If, in addition,  the coefficients of ${\bf Z}$ are real then we may be call it a positive real matrix.
}.  
Thus, it follows from the identity
\[
 (\one_m+{\bf S}_{11})^{-1} = \frac{1}{2}[\one_m + (\one_m -{\bf S}_{11})(\one_m +{\bf S}_{11})^{- 1}]
\]
 that $(\one_m +{\bf S}_{11})^{-1}$ is also a positive matrix.
Consequently, by invoking the continuous version of a result in \cite{rob:Basu1} it follows\footnote{%
The least common denominator of entries of a positive matrix in irreducible form is an immittance Hurwitz polynomial, which is product of a scattering Hurwitz factor and simple irreducible reactance Hurwitz factors, each of which are nonzero in $\Re {\bf p} >0$. We refer to \cite{stb:Fett2} for more details.}
 that the least common denominator of the entries of $ (\one_m+{\bf S}_{11})^{-1}$ in irreducible form denoted by the polynomial $d$ is nonzero  in $\Re {\bf p} >0$. On the other hand,  the denominator of ${\bf S}_{21*}$ is equal to $g_* e_*$, where $g$ is scattering Hurwitz thus nonzero in $\Re {\bf p} >0$, and $e=e(p_1)\neq 0$ for $\Re p_1 >0$.
We now choose $\alpha=dge$, which is nonzero in $\Re {\bf p} >0$ due to the reasons just mentioned. Furthermore, clearly  $\alpha_*$ cancels the denominator of $(\one_m + {\bf S}_{11*})^{-1}{\bf S}_{21*}$. Since ${\bf S}_{11}$ and ${\bf S}_{21}$ are {\em a fortiori}  holomorphic in $\Re {\bf p} > 0$, the same property for ${\bf S}_{12}$ and ${\bf S}_{22}$ follows from (\ref{scatt:*1}) and (\ref{scatt:*2}) and the fact that $\alpha\neq 0$ in $\Re {\bf p} >0$. The proof of the present theorem is thus complete for the case $m=n$.

 Next, if $m < n$ then we consider an $(n \times n)$ matrix $\check{\bf S}_{11}$ by adding
 zero rows to ${\bf S}_{11}$.  Then 
$\one_{n} - \check{\bf S}_{11*} \check{\bf  S}_{11} = \one_{n}- {\bf S}_{11*}{\bf S}_{11} \geq 0$ in Re ${\bf p}>0$,
 i.e., $\check{\bf S}_{11}$ is a bounded function.  On the other hand, if $m>n$
 we define $\check{\bf S}_{11}$ as the $(m \times m)$ matrix by adding columns
 of zeros to ${\bf S}_{11}$.  Then 
\[
\one_{m} - \check{\bf S}_{11*}\check{\bf S}_{11} = \left[ \begin{array}{cc} \one_{n}-{\bf S}_{11*}{\bf S}_{11} & {\bf 0} \\{\bf 0} & {\bf 0} \end{array} \right] \geq 0 \text{ in } \Re{\bf p} > 0,
\]
i.e., $\check{\bf S}_{11}$ is again a bounded function.  In either case, we have $r=\rank (\one - \check{\bf S}_{11*}\check{\bf S}_{11})=\rank (\one - {\bf S}_{11*}{\bf S}_{11})$, 
and it has already been shown that the square bounded matrix $\check{\bf S}_{11}$ can be bordered up to a
 lossless bounded matrix $\check{\bf S}$ with its submatrices $\check{\bf S}_{21}$, $\check{\bf S}_{12}$, and $\check{\bf S}_{22}$ of respective sizes $(2r\times p)$, $(p\times 2r)$, and $(2r\times 2r)$.  
The form ({\ref{scatt:*2}}) can then be obtained by identifying $\check{\bf S}$ with $\bf S$, and then by repartitioning it  accordingly. \hfill{Q.E.D}

 {\flushleft{\bf Remark:}}
 Note that the above result does not imply that ${\bf S}$ is a polynomial
 matrix, nor does it imply that the ${\bf S}$ obtained via the procedure
 outlined above has the smallest possible  size for a given ${\bf S}_{11}$.  Thus,
 if the lossless bounded ${\bf S}$ is synthesized as in Section \ref{scatt:S21} to
 yield a synthesis of ${\bf S}_{11}$, the minimality of neither the total number of dynamic
 elements, i.e., the number of $p_{1}$ and $p_{2}$ type elements\footnote{%
The greatest common denominator of ${\bf S}$ in (\ref{scatt:matrix}) is a factor of $ge$, and thus {\em at most} the $\deg_1 (ge)$ of $p_1$ type elements and $\deg_2 (ge)$  of $p_2$ type elements are required in the realization of the lossless bounded matrix $\bf S$ by following a synthesis strategy outlined in Section \ref{scatt:S21}.
}
nor the number of resistors (i.e., 'fully absorbing' ports in the discrete case) created in the synthesis procedure is ensured.

 {\flushleft{\bf Remark:}}
Careful examination of  (\ref{scatt:*1}) and (\ref{scatt:*2}) shows that it would suffice in the above
 proof to choose $\alpha =  g\cdot e\cdot \det (g\one_{m} + {\bf P}_{11}) $.  But the degree of
 least common denominator of ${\bf S}$ would not necessarily be minimal in  such a case.

We now examine the situation when  $\det (g\one_m+{\bf P}_{11})=0$ for some ${\bf p} = {\bf p}_{0}$ in $\Re{\bf p}_0 >0$ in some detail. This  implies the existence of a constant vector ${\bf x}$ with $(\one_m + {\bf S}_{11}){\bf
 x} = 0$, i.e., $||{\bf S}_{11}{\bf x}|| = ||{\bf x}||$. Normalizing ${\bf x}$ we may write $||{\bf S}_{11}{\bf x}||= ||{\bf x}||=1$ for ${\bf p} = {\bf p}_{0}$.

We will now  show that under the situation indicated above the vector ${\bf y}={\bf y}({\bf p})$ defined by ${\bf y}  = {\bf S}_{11}{\bf x}$ is, in fact, a vector of constant length. For this, note that ${\bf y}$ inherits the property of holomorphy in $\Re{\bf  p}>0$ from ${\bf S}_{11}$. On the other hand, since ${\bf S}_{11}$ is a bounded function of ${\bf p}$ we have 
$1-||{\bf y(j\boldsymbol\omega)}||^{2} = {\bf x}^{*}[{\one - {\bf S}_{11}^*(j\boldsymbol\omega)}{\bf S}_{11}(j\boldsymbol\omega)]{\bf x} \geq 0$, 
i.e., $||{\bf y}(j\boldsymbol\omega)|| \leq 1$ for all real $j\boldsymbol\omega$, wherever ${\bf S}_{11}(j\boldsymbol\omega)$ is well defined. Thus, ${\bf y}({\bf p})$ is a bounded function of $\bf p$, the modulus of which cannot (\cite[Theorem xxx]{stb:Fett2}) assume the value equal to unity at any interior point of the domain of holomorphy unless it is a constant. However, since  $||{\bf y}({\bf p}_0)||=1$, and ${\bf p}_0$ is an interior point of $\Re {\bf p} >0$, it follows that ${\bf y}$ is a constant vector of unit length. To this end, we have the following result.

 \begin{proposition}\label{scatt:P5}
 Let ${\bf S}_{11}={\bf S}_{11}({\bf p})$ be a bivariate $(m \times m)$ bounded
 rational matrix such that ${\bf S}_{11}{\bf x}$ is a
 constant vector of unit length for some constant vector ${\bf x}$ of
 unit length.  Then there exists constant  unitary matrices
 ${\bf U}_{x}$ and ${\bf U}_{y}$, satisfying
 \[ 
{\bf S}_{11}={\bf U}_{y}^{*} \left[\begin{array}{cc}
 \one & {\bf 0} \\
 {\bf 0} & {\bf H} \end{array} \right] {\bf U}_{x} ,
\]
 where ${\bf H}$ is a bounded rational matrix with the further property that
 there is no constant vector  ${\bf x}$ of unit length such that ${\bf y}={\bf
 Hx}$ is a constant vector of unit length.
 \end{proposition}
 {\em Proof:}
 Choose ${\bf U}_{x}$ and ${\bf U}_{y}$ to be  unitary matrices with
 respective first columns ${\bf x}$ and ${\bf y}$.  Consider then
 ${\bf G}=[g_{ij}] = {\bf U}_{y}^{*}{\bf S}_{11}{\bf U}_{x}$, which is clearly
 bounded with the further property that $g_{11}={\bf y}^{*}{\bf S}_{11}{\bf x}
 = 1$.  It also follows from boundedness of ${\bf G}$ that for almost all
 $\mbox{\boldmath $\omega$}$ we have: $\sum_{i=1}^{n}|g_{i1}(j\mbox{\boldmath
 $\omega$})|^{2} \leq 1$, which along with $g_{11} = 1$ yields that $g_{i1}
 \equiv 0$ for $i=2$ to $m$.

 Taking these into account, straightforward computation shows that
 the determinant of the $(2 \times 2)$ sub-matrix of $\one -
 {\bf G}^{*}(j\mbox{\boldmath $\omega$}){\bf G}(j\omega)$ consisting of the
 $1$-st and $i$-th  rows and columns is $-|g_{1i}(j\mbox{\boldmath $\omega$})|^{2}$, which being a principal minor, must
 be nonnegative in view of the fact that
 $\one -{\bf G}^{*}(j\mbox{\boldmath $\omega$}){\bf G}(j\mbox{\boldmath $\omega$})
 \geq 0$ for any real $\boldsymbol\omega$, whenever ${\bf G}(j\boldsymbol\omega)$ is well defined. 
Thus, $g_{1i} \equiv 0$ for $i = 2$ to $m$.  Consequently, ${\bf G}$ can be written as
\be\label{scatt:97}
 {\bf G} = \left[\begin{array}{cc}
 1 & {\bf 0} \\
 {\bf 0} & {\bf G}_{22} \end{array} \right],
 \ee
 where ${\bf G}_{22}$ is a  rational matrix of size $(m-1) \times
 (m-1)$.  Furthermore, ${\bf G}_{22}$ is a bounded matrix since ${\bf G}$ in
 (\ref{scatt:97}) is a bounded matrix.  The result of the present proposition
 is then obtained by repeating the above decomposition as many times as necessary.\hfill{Q.E.D}

Under the situation indicated above, a synthesis of ${\bf S}_{11}$  can be achieved essentially by synthesizing ${\bf H}$, which is of smaller size.

 \subsection{Discrete Domain considerations}
  In this section we first show by using Theorem \ref{scatt:T6} that a
 discrete version of the unitary dilation or embedding problem can be solved.
 This coupled with the synthesis of discrete lossless bounded matrices
 developed in Section \ref{scatt:S21} then provides a complete synthesis of an
 arbitrary discrete bounded matrix.  A few other remarks on the simplification
 in realization when conditions of Proposition \ref{scatt:P5} hold are also
 made.
 \begin{theorem}\label{scatt:TDE}
 Let $\boldsymbol\Sigma =\boldsymbol\Sigma ({\bf z})$ be a $(m\times m)$ bivariate discrete bounded rational matrix
  in the variables ${\bf z}=(z_{1},z_{2})$.  Then there exists a lossless discrete
  bounded rational matrix $\mbox{\boldmath $\Sigma$}$ of size $(p+2r)\times (p+2r)$, $p = {\rm max}(m,n)$,  such that
 $\mbox{\boldmath $\Sigma$}$ can be partitioned as
 \be\label{scatt:111a}
 \mbox{\boldmath $\Sigma$} = \left[ \begin{array}{cc}
 \mbox{\boldmath $\Sigma$}_{11} & \mbox{\boldmath $\Sigma$}_{12} \\
 \mbox{\boldmath $\Sigma$}_{21} & \mbox{\boldmath $\Sigma$}_{22}
 \end{array} \right],
 \ee
where $r = {\rank}\;(\one_n - \boldsymbol\Sigma_{11*}\boldsymbol\Sigma_{11})$.
 \end{theorem}
 {\em Proof:}
 We only provide a sketch of the proof.  Consider the action of the double bilinear transformation (\ref{S_fact:a*})
  on $\mbox{\boldmath $\Sigma$}_{11}$, which produces the bounded 
 rational matrix ${\bf S}_{11}$ as functions of the variables $p_{1}$, $p_{2}$, i.e., we have
 \be\label{scatt:113a}
 {\bf S}_{11} = \left[ \mbox{\boldmath $\Sigma$}_{11} \right]_{z_{i} =
 \frac{1-p_{i}}{1+p_{i}}}; \mbox{\boldmath $\Sigma$} = \left[ {\bf S}_{11}
 \right]_{p_{i} = \frac{1-z_{i}}{1+z_{i}}}; \;\; i=1,2
 \ee
 Invoking Theorem \ref{scatt:T6} we obtain the lossless bounded (real)
 rational matrix ${\bf S}$ as in \ref{scatt:matrix}.  Now consider the action of the
 inverse transform to yield
 \be\label{scatt:114a}
 \mbox{\boldmath $\Sigma$} = [{\bf S}]_{p_{i}=\frac{1-z_{i}}{1+z_{i}}}; \;\; i
 = 1,2
 \ee
 We then correspondingly also have
 \[ 
\mbox{\boldmath $\Sigma$}_{11}=[{\bf S}_{11}]_{p_{i}=\frac{1-
 z_{i}}{1+z_{i}}}; \;\; i= 1,2
\]
 Since it can be trivially shown from the corresponding property of
 ${\bf S}$ that $\mbox{\boldmath $\Sigma$}$ is a discrete lossless bounded
  rational matrix.\hfill{Q.E.D}

 To indicate exactly how the unitary dilation or embedding solves the
 problem of dissipative synthesis from a solution to lossless synthesis,
 consider the bounded $\mbox{\boldmath $\Sigma$}_{11}$ of size $(m \times m)$
 being bordered up into lossless bounded $\mbox{\boldmath $\Sigma$}$ of size
 $(m + r) \times (m + r)$ as in (\ref{scatt:111a}).  Then a synthesis of
 $\mbox{\boldmath $\Sigma$}$ is available according to the algorithm of
 Section \ref{scatt:S21}.  Let ${\bf a}_{1},{\bf a}_{2}$ be the incident
 signals (or waves in the continuous case) on the first $m$ and last $r$ ports
 of this realization and ${\bf b}_{1},{\bf b}_{2}$ be the outgoing signals (or
 reflected waves in the continuous case) from the first $m$ and last $r$ ports
 of the realization respectively.  Then it follows from (\ref{scatt:111a})
 that
 \[ 
\left[\begin{array}{c} {\bf b}_{1} \\ {\bf b}_{2} \end{array} \right] =
 \left[\begin{array}{cc}
 \mbox{\boldmath $\Sigma$}_{11} & \mbox{\boldmath $\Sigma$}_{12} \\
 \mbox{\boldmath $\Sigma$}_{21} & \mbox{\boldmath $\Sigma$}_{21} 
\end{array} \right]
 \left[\begin{array}{c} {\bf a}_{1} \\ {\bf a}_{2} \end{array} \right]. 
\]
 Now, if the last $r$ ports of the realization of $\boldsymbol\Sigma$ are terminated in fully
 absorbing elements, i.e.,  if ${\bf a}_{2} = {\bf 0}$ then we would have ${\bf b}_{1} = \boldsymbol\Sigma_{11}{\bf a}_{1}$.  
This shows that when the last $r$ ports of $\boldsymbol\Sigma$ is 
  terminated by fully absorbing elements, the transfer function at the first $m$ pairs of terminals, i.e., $m$ ports of  $\boldsymbol\Sigma$, 
is exactly $\boldsymbol\Sigma_{11}$, thus providing a realization for $\boldsymbol\Sigma_{11}$, assuming that a realization of $\boldsymbol\Sigma$ is available. 
 
 Let $\mbox{\boldmath $\Sigma$}_{11}$ be such that there exists a constant
 vector ${\bf x}$ of unit length, for which $\mbox{\boldmath
 $\Sigma$}_{11}{\bf x}$ is a constant vector of unit length.  Then by a
 discrete version of Proposition \ref{scatt:P5} the synthesis problem can be
 simplified as follows.  First note that a synthesis of $\diag [\one\;{\bf
 H}]$ as in Proposition \ref{scatt:P5} is easily achieved if a synthesis of
 ${\bf H}$ is available.  To see this, assume that $\one$ is of size $m$,
 whereas ${\bf H}$ is of size $p$.  Then if ${\bf a}_{1},{\bf a}_{2}$ and
 ${\bf b}_{1},{\bf b}_{2}$ are as defined in the last paragraph, but now for
 $\diag [\one \;{\bf H}]$ instead of $\boldsymbol\Sigma$,  we have ${\bf a}_{1} = {\bf
 b}_{1}$, ${\bf a}_{2}={\bf H}{\bf b}_{2}$.  This shows that a synthesis of
 $\diag [\one \;{\bf H}]$ is, in fact, one in which the last $r$ ports realize
 the transfer function ${\bf H}$  of smaller size, and the first $m$ ports
 consist of direct connections.  Finally, pre- and post-multiplication by {\em
 constant} unitary matrices ${\bf U}_{x}$ and ${\bf U}_{y}$ are accomplished
 by cascading the realization of $\diag [\one \; {\bf H}]$ with the realizations of constant
 unitary ${\bf U}_{x}$ and ${\bf U}_{y}$ in an appropriate order.  Since
 realization of constant unitary ${\bf U}_{x}$ and ${\bf U}_{y}$ are readily
 available (cf. literature on orthogonal filters \cite{dar:Dewi1}), this shows that the essential problem in realizing
 $\mbox{\boldmath $\Sigma$}_{11}$ is that of ${\bf H}$,  which is of smaller size, and is 
 lossless bounded with the further property that there is no constant vector
 ${\bf x}$ such that ${\bf Hx}$ is constant with $||{\bf Hx}|| = ||{\bf x}|| =
 1$.

 \section{2-D bounded real lemma}\label{scatt:S23}
 We next consider the state space characterization of passivity
 and losslessness of 2-D systems. While our developments are largely motivated
 by passive synthesis, in view of its importance  in various areas of
 system theory, 2-D version of such a result is believed to be of independent
 interest.  For this, we consider the Roesser 
 model\footnote{%
Alternate models, e.g., the behavioral models advanced by Jan Willems \cite{pillai} could be considered in this context.
} 
 for our 2-D state space considerations \cite{fb:bosebook}. 
 In relating several apparently disparate notions, the 1-D Kalman-Popov-Yakubovitch (KPY) 
 lemma, otherwise known as the positive real lemma, which provides a characterization of the property of
 dissipativeness of arbitrary (minimal) realization of impedance-like transfer functions has proven to be
 pivotal.  Equivalently, there also exits the so called bounded real lemma which provides a direct characterization of state space realization of corresponding passive scattering function. While in 1-D, the bounded real  lemma can be derived via techniques akin to those
 known in linear quadratic optimal control theory, it can be alternatively
 viewed as a combined consequence of synthesizability of  positive (real)
 transfer functions and spectral factorability of para-hermitian positive
 definite transfer matrices.  Note that
 in the case of scalar transfer functions this last result is a reformulation
 of the fact that positive polynomials can be expressed as the sum of squares
 of two real polynomials.  Since this latter fact does not extend to 2-D in
 general, our starting point has been the sum of squares representation of positive 
 polynomials originating in the work of Hilbert \cite{S_est:hilbert} and Landau \cite{S_fact:land}.  
 As seen in the Section \ref{S22}, when interpreted as a spectral
 factorability type result, it allows us to embed an arbitrary 
 bounded  transfer function matrix into a lossless  bounded  matrix.  This key observation along with synthesizability of
 lossless bounded 2-D transfer function matrices  then
 demonstrate the synthesizability of arbitrary positive of bounded 
 functions.  Thus, since a passive synthesis automatically provides a dissipative
 realization, it in turn leads the way to a {\em weak} form of 2-D KPY lemma.

 It may be noted that the 2-D state space theory differs from 1-D in a number of ways, e.g., it is well known that the local states such as those in 2-D Roesser's state  space model  do not contain the full
 information regarding the complete history of the system.   More importantly,  in 1-D, a minimal passive synthesis along with the {\em state space isomorphism}  result (i.e., any two minimal realizations of an 1-D system are related by a similarity transformation)
 directly yields the  the bounded real lemma.  Although, in our 2-D context we do not necessarily
 have a minimal passive synthesis via the procedure outlined here, and an appropriate state space isomorphism result is not available \cite{rob:Basu3,kung}, a weak 2-D version of the bounded real lemma can indeed be proven.

Once again we will treat the discrete domain version of the problem, while unless otherwise stated the continuous domain  versions will follow from an obvious analog.
 \begin{lemma}\label{scatt:L2}
 {\bf [(weak) 2-D Bounded Real (BR) Lemma]}\\
 Let ${\bf H}={\bf H}(z_{1},z_{2})$ be a rational matrix of size $(m \times
 n)$,  which, in addition, is discrete bounded. Then there exists integers
 $p$, $q$ and matrices ${\bf A}$, ${\bf B}$, ${\bf C}$ and ${\bf D}$ of
 appropriate sizes such that
 \be\label{scatt:100}
 {\bf H}(z_{1},z_{2}) = {\bf A} + {\bf B}(z_{1}^{-1}\one_{p} \oplus  z_{2}^{-1}\one_{q} - {\bf D})^{-1}{\bf C},
 \ee
 along with
 \be\label{scatt:101}
 \one  - {\bf T}^{*}{\bf T} \geq 0,
 \ee
 where
 \be\label{scatt:102}
 {\bf T}=
 \left[ \begin{array}{cc}
 {\bf A} & {\bf B} \\
 {\bf C} & {\bf D}
 \end{array} \right].
 \ee
 \end{lemma}

 {\flushleft{\bf Remark:}}
 Note that Lemma {\ref{scatt:L2}} essentially states that an {\em internally passive} 2-D
 Roesser's state space realization for the transfer function ${\bf
 H}(z_{1},z_{2})$ can be obtained. If ${\bf B}$, ${\bf C}$, ${\bf D}$ are
 partitioned as
 \be\label{scatt:103}
 {\bf D}=
 \left[ \begin{array}{cc}
 {\bf D}_{11} & {\bf D}_{12} \\
 {\bf D}_{21} & {\bf D}_{22}
 \end{array} \right],
\quad
 {\bf B}=
 \left[ \begin{array}{cc}
 {\bf B}_{1} & {\bf B}_{2}
 \end{array} \right],
 \quad
{\bf C}=
 \left[ \begin{array}{c}
 {\bf C}_{1} \\
 {\bf C}_{2}
 \end{array} \right],
 \ee
where ${\bf D}_{11}$, ${\bf D}_{22}$ are of respective sizes $(p\times p)$
 and $(q \times q)$; ${\bf B}_{1}$, ${\bf B}_{2}$ has respectively $p$ and $q$
 columns; ${\bf C}_{1}$, ${\bf C}_{2}$ has respectively $p$ and $q$ rows, then
 ${\bf H}(z_{1},z_{2})$ can be viewed as the transfer function between the
 input vector ${\bf u}(m,n)$ and the output vector ${\bf y}(m,n)$ in the
 Roesser's state space model

 \be\label{scatt:104}
 \left[ \begin{array}{c}
 {\bf x}_{h}(i+1,j) \\ {\bf x}_{v}(i,j+1)
 \end{array} \right]
 =
 \left[ \begin{array}{cc}
 {\bf D}_{11} & {\bf D}_{12} \\ {\bf D}_{21} & {\bf D}_{22}
 \end{array} \right]
 \left[ \begin{array}{c}
 {\bf x}_{h}(i,j) \\ {\bf x}_{v}(i,j)
 \end{array} \right] +
 \left[ \begin{array}{c}
 {\bf C}_{1} \\ {\bf C}_{2}
 \end{array} \right] {\bf u}(i,j),
 \ee
 \be\label{scatt:105}
 {\bf y}(i,j)=
 \left[ \begin{array}{cc}
 {\bf B}_{1} & {\bf B}_{2}
 \end{array} \right]
 \left[ \begin{array}{c}
 {\bf x}_{h}(i,j) \\ {\bf x}_{v}(i,j)
 \end{array} \right] + {\bf Au}(i,j).
 \ee

 {\em Proof of Lemma {\ref{scatt:L2}}}:
 It follows from 2-D discrete embedding Theorem (cf. Theorem \ref{scatt:TDE})
 that the $(m \times n)$ discrete bounded rational matrix  ${\bf H}$ can be embedded in a 2-D discrete {\em lossless} bounded rational matrix ${\bf G}$ of size, say $(\ell \times \ell)$ as in (\ref{scatt:115}), where $\ell=\max(m,n) +2\rank (\one-{\bf H}_*{\bf H})$. We thus have the situation depicted in Figure \ref{fig:BR-lemma}, in  which
\begin{figure}[h]
    \centering
    \includegraphics[width=15cm]{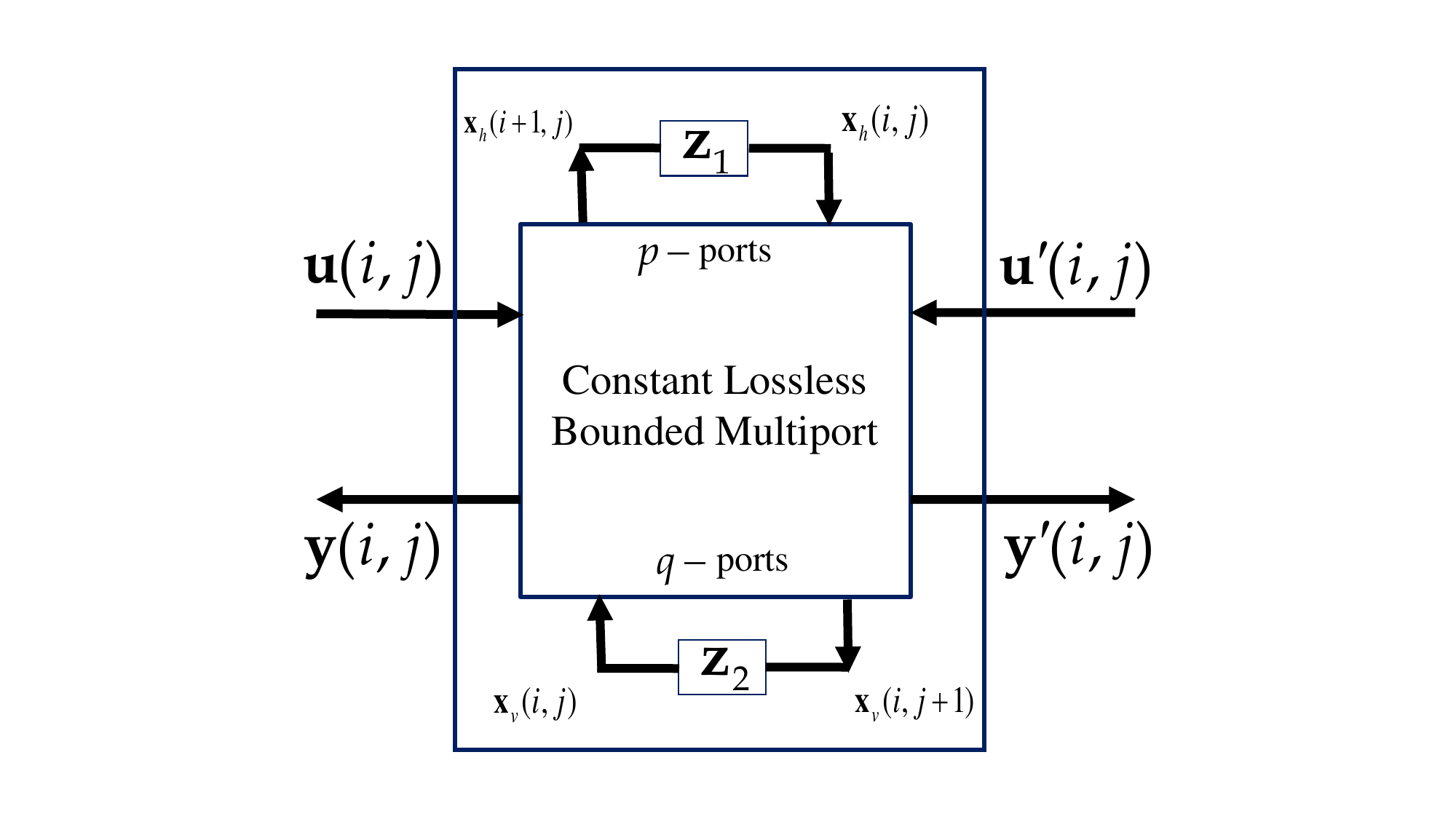}
    \caption{Proof of bounded real lemma. Input-output variables ${\bf u}$, ${\bf y}$, ${\bf u}'$, ${\bf y}'$ constitute $\ell$ ports.}
    \label{fig:BR-lemma}
\end{figure}
 \be\label{scatt:115}
 {\bf G}= 
\left[ \begin{array}{cc}
 {\bf H} & {\bf G}_{12} \\ {\bf G}_{21} & {\bf G}_{22}
 \end{array} \right],
\quad
\left[ \begin{array}{c} {\bf Y} \\ {\bf Y}' \end{array} \right] =
 {\bf G} \left[ \begin{array}{c} {\bf U} \\ {\bf U}' \end{array} \right],
 \ee
where ${\bf U}$, ${\bf U}'$, ${\bf Y}$, ${\bf Y}'$ are 2-D $z$-transforms of input and output variables ${\bf u}$, ${\bf u}'$, ${\bf y}$, ${\bf y}'$ respective sizes $n \times 1$, $(\ell-n) \times 1$, $m \times 1$, and $(\ell -m) \times 1$. Since we have ${\bf y}|_{{\bf u}'=0}= {\bf H}{\bf u}$, an input-output realization of $\bf H$ can thus be obtained from an input-output realization of $\bf G$ by setting ${\bf u}' = 0$.

 Now, it is known from our discussions in Section \ref{scatt:S21} that it is possible to obtain
 an internally  passive (in fact, minimal) synthesis for the 2-D discrete lossless bounded
 matrix ${\bf G}$ . More specifically, one way of
 viewing the synthesis procedure is to extract $p$ of $z_{1}$ type delays and
 $q$ of $z_{2}$ delays in such a way that we are left with a constant lossless
 bounded multiport with as many as $(p+q+\ell)$ ports. Let the transfer
 function of the constant lossless bounded $(p+q+\ell)$-port in  Figure \ref{fig:BR-lemma} be given by  $\bar{\bf T}$, i.e., 
 \be\label{scatt:117}
 \left[ \begin{array}{c} {\bf y}(i,j) \\{\bf x}_{h}(i+1,j) \\{\bf x}_{v}(i,j+1) \\ {\bf y}'(i,j) \end{array} \right] = \bar{\bf T}
 \left[ \begin{array}{c} {\bf u}(i,j) \\ {\bf x}_{h}(i,j) \\ {\bf x}_{v}(i,j) \\ {\bf u}'(i,j) \end{array} \right]
\ee
 and partition $\bar{\bf T}$ as
 \be\label{scatt:118}
 \bar{\bf T}=\left[ \begin{array}{ccc} {\bf A} & {\bf B} & {\bf T}_{13} \\
 {\bf C} & {\bf D} & {\bf T}_{23} \\
 {\bf T}_{31} & {\bf T}_{32} & {\bf T}_{33}
  \end{array} \right]
=
\left[
\begin{array}{c|c}
{\bf T} & \begin{array}{c} {\bf T}_{13}  \\  {\bf T}_{23} \end{array} \\
\hline
\begin{array}{cc} {\bf T}_{31} & {\bf T}_{32} \end{array} & {\bf T}_{33}
\end{array}
\right],
\ee
 where  ${\bf T}$ is as in ({\ref{scatt:102}}). Since $\bar{\bf T}$ is lossless bounded we
 have $\bar{\bf T}^{*}\bar{\bf T}={\bf I}$, which in view of
 ({\ref{scatt:118}}) yields 

 \be\label{scatt:119}
 {\bf I} - {\bf T}^{*}{\bf T}  \geq 0.
 \ee
 Clearly, since in Figure \ref{fig:BR-lemma}, ${\bf x}_{h}(i,j)$ and ${\bf x}_{v}(i,j)$ are
 outputs of $z_{1}$ and $z_{2}$  type delays respectively, they are valid
 state variables. Furthermore, by setting ${\bf u}'(i,j)=0$ it can be verified
 from (\ref{scatt:117}) that ${\bf T}$ obtained as above  indeed
 corresponds to Roesser's state  space model ({\ref{scatt:104}}),
 ({\ref{scatt:105}}). It then is well known and not difficult to see that the transfer function ${\bf
 H}$ between ${\bf u}(i,j)$ and ${\bf y}(i,j)$ is indeed given by the formula
 ({\ref{scatt:100}}), thus completing the proof of the present
 lemma.\hfill{Q.E.D}

{\flushleft{\bf Remark:}}
Next, for arbitrary nonsingular matrices  ${\bf Q}_1$ and ${\bf Q}_2$  of respective sizes $(p\times p)$ and $(q\times q)$ we consider the block similarity transform defined via ${\bf Q}={\bf Q}_{1}\oplus {\bf Q}_{2}$,  and 
\[
\check{\bf A}={\bf A},\quad
\check{\bf B}={\bf B}{\bf Q},\quad
\check{\bf C}={\bf Q}^{-1}{\bf C},\quad
\check{\bf D}={\bf Q}^{-1}{\bf D}{\bf Q}.
\]
Then $\check{\bf A}$, $\check{\bf B}$, $\check{\bf C}$ and $\check{\bf D}$, properly partitioned as in (\ref{scatt:103}), provide an alternative Roesser state space realization of the 2-D transfer function ${\bf H}$, because the expression (\ref{scatt:100}) remains valid with ${\bf A}$, ${\bf B}$, ${\bf C}$, ${\bf D}$ respectively replaced by $\check{\bf A}$, $\check{\bf B}$, $\check{\bf C}$ and $\check{\bf D}$. However, this latter realization may not be {\em internally passive}. Corresponding to  (\ref{scatt:101}) and  (\ref{scatt:102}) we then have
\be\label{general-BR}
{\bf G}  - \check{\bf T}^* {\bf G}\check{\bf T} \geq 0,
\ee
where
\[
\check{\bf T}={\check{\bf Q}}^{-1}{\bf T}\check{\bf Q}, 
\quad 
{\bf G} =\check{\bf Q}^*\check{\bf Q}\geq 0,
\text{ and }
\check{\bf Q}=\diag[\one,\; {\bf Q}].
\]
Expression (\ref{general-BR}) may be viewed as a generalized form of  expression (\ref{scatt:101}) valid for broader class of  realizations that are not necessarily internally passive. In 1-D, due to the state space isomorphism result any two minimal  realizations are related by a similarity transform, and thus in particular to  a minimal  internally passive realization. In 2-D, however, two arbitrary minimal realizations are not necessarily related by a block diagonal similarity transform \cite{kung} even when only minimal passive realizations are being considered \cite[Corollary 6.1]{rob:Basu3}. Equations ({\ref{scatt:101}}) and
 ({\ref{scatt:102}}) together can be viewed as a {\em weak} form of a 2-D bounded
 real lemma for the transfer function ${\bf H}$. It is weak in that it applies
 only to the specific realization obtained in Lemma {\ref{scatt:L2}}. The
 lack of  a 2-D state space isomorphism result \cite{rob:Basu3,kung}
prevents further generalization to arbitrary minimal realizations in 2-D.

{\flushleft{\bf Remark:}}
 The matrix inequality in ({\ref{scatt:101}}) is equivalent to the existence
 of two further matrices ${\bf L}$ and ${\bf W}$ such that
 \be\label{scatt:106}
 {\bf 1} - {\bf T}^{*}{\bf T} =
 \left[ \begin{array}{c} {\bf L}^{*} \\ {\bf W}^{*} \end{array} \right]
 \left[ \begin{array}{cc} {\bf L} & {\bf W} \end{array} \right],
 \ee
where ${\bf L}$ is of size $(r \times \ell)$ and ${\bf W}$ is of size $(r
 \times (p+q))$ for some $r$.  Note further that
 ({\ref{scatt:101}}), ({\ref{scatt:102}})  and (\ref{scatt:106}) combined together can be written as
 \be\label{scatt:107}
 \one_{n} - {\bf A}^{*}{\bf A} - {\bf C}^{*}{\bf C} = {\bf L}^{*}{\bf L},
 \ee
 \be\label{scatt:108}
 -{\bf A}^{*}{\bf B} - {\bf C}^{*}{\bf D} = {\bf L}^{*}{\bf W},
 \ee
 \be\label{scatt:109}
 \one_{p+q} - {\bf B}^{*}{\bf B} - {\bf D}^{*}{\bf D} = {\bf W}^{*}{\bf W}.
 \ee
We then have the following result.
\begin{corollary}\label{scatt:C2}
 Let ${\bf H}={\bf H}(z_{1},z_{2})$ be a (real) rational matrix of size $(m
 \times n)$, which, in addition, is discrete bounded. Let ${\bf A}$, ${\bf
 B}$, ${\bf C}$, ${\bf D}$ as in ({\ref{scatt:102}}) be a passive realization
 of ${\bf H}$. Let ${\bf L}$ and ${\bf W}$ be defined as in
 ({\ref{scatt:106}}) or equivalently, in ({\ref{scatt:107}}),
 ({\ref{scatt:108}}), ({\ref{scatt:109}}). Then we have
\be\label{scatt:110}
 \one -\tilde{{\bf H}}{\bf H} = \tilde{{\bf P}}{\bf P},
 \ee
 where ${\bf P}={\bf P}(z_{1},z_{2})$ is the transfer function of the
 Roesser's state space model given by
 \be\label{scatt:111}
 \left[ \begin{array}{c}
 {\bf x}_{h}(i+1,j) \\ {\bf x}_{v}(i,j+1)
 \end{array} \right]
 =
 \left[ \begin{array}{cc}
 {\bf D}_{11} & {\bf D}_{12} \\ {\bf D}_{21} & {\bf D}_{22}
 \end{array} \right]
 \left[ \begin{array}{c}
 {\bf x}_{h}(i,j) \\ {\bf x}_{v}(i,j)
 \end{array} \right] +
 \left[ \begin{array}{c}
 {\bf C}_{1} \\ {\bf C}_{2}
 \end{array} \right] {\bf u}(i,j),
 \ee
 \be\label{scatt:112}
 \mbox{\boldmath $\eta$}(i,j)=
 \left[ \begin{array}{cc}
 {\bf W}_{1} & {\bf W}_{2} \\
 \end{array} \right]
 \left[ \begin{array}{c}
 {\bf x}_{h}(i,j) \\ {\bf x}_{v}(i,j)
 \end{array} \right] + {\bf L u}(i,j),
 \ee
 with ${\bf W}=\left[ \begin{array}{cc} {\bf W}_{1} & {\bf W}_{2} \\
 \end{array} \right]$. In other words, ${\bf P}$ is given by
 \be\label{scatt:113}
 {\bf P}(z_{1},z_{2}) = {\bf L} + {\bf W}[z_{1}^{-1}{\bf I}_{p} \oplus
 z_{2}^{-1}{\bf I}_{q} - {\bf D}]^{-1}{\bf C}.
 \ee
 \end{corollary}
 {\em Proof:} 
 The proof is routine algebraic manipulation for which we adopt
 the compact notation $\boldsymbol\zeta=z_{1}{\bf I}_{p}\oplus z_{2} {\bf I}_{q}$.
 By substituting for ${\bf H}$ from ({\ref{scatt:100}}) and expanding the
 product, then from  ({\ref{scatt:107}}), ({\ref{scatt:108}}), ({\ref{scatt:109}}) by making the
 replacements 
$\one-{\bf A}^{*}{\bf A}= {\bf C}^{*}{\bf C}+{\bf L}^{*}{\bf L}$, 
${\bf A}^{*}{\bf B}=-({\bf C}^{*}{\bf D}+{\bf L}^{*}{\bf W})$, 
${\bf B}^{*}{\bf B}=\one-{\bf D}^{*}{\bf D}-{\bf W}^{*}{\bf W}$,
 and subsequently rearranging terms we can write
 \be\label{scatt:114}
 {\bf 1}-\tilde{\bf H}{\bf H} = \tilde{{\bf P}}{\bf P}
 +{\bf C}^{*}({\bf t}_{1}+{\bf t}_{2}+{\bf t}_{3}){\bf C},
 \ee
 where
 \[ 
{\bf t}_{1}=\one -(\boldsymbol\zeta-{\bf D}^{*})^{-1}(\boldsymbol\zeta^{-1} - {\bf D})^{-1}, 
\]
\[ 
{\bf t}_{2}={\bf D}(\boldsymbol\zeta^{-1}-{\bf D})^{-1} + (\boldsymbol\zeta-{\bf D}^{*})^{-1}{\bf D}^{*},
\]
 and
 \[ 
{\bf t}_{3}=(\boldsymbol\zeta-{\bf D}^{*})^{-1}{\bf D}^{*}{\bf D}(\boldsymbol\zeta^{-1}-{\bf D})^{-1} .
\]
 Since we obviously have
 \begin{eqnarray*}
 {\bf t}_{1} & = & (\boldsymbol\zeta-{\bf D}^{*})^{-1}\{(\boldsymbol\zeta-{\bf  D}^{*})(\boldsymbol\zeta^{-1}-{\bf D})-\one\}(\boldsymbol\zeta^{-1}-{\bf D})^{-1} \\
     & & = (\boldsymbol\zeta-{\bf D}^{*})^{-1}\{{\bf D}^{*}{\bf D}-{\bf  D}^{*}\boldsymbol\zeta^{-1}-{\bf D} \boldsymbol\zeta \}(\boldsymbol\zeta^{-1}-{\bf D})^{-1},
 \end{eqnarray*}

 and

 \begin{eqnarray*}
 {\bf t}_{2} & = & (\boldsymbol\zeta-{\bf D}^{*})^{-1}\{{\bf D}^{*}(\boldsymbol\zeta^{-
 1}-{\bf D})+(\boldsymbol\zeta-{\bf D}^{*}){\bf D}\}(\boldsymbol\zeta^{-1}-
 {\bf D})^{-1} \\
       & = & (\boldsymbol\zeta-{\bf D}^{*})^{-1}\{{\bf D}^{*}\boldsymbol\zeta ^{-1}+{\bf  D}\boldsymbol\zeta -2{\bf D}^{*}{\bf D}\}(\boldsymbol\zeta^{-1}-{\bf D})^{-1},
 \end{eqnarray*}
 it follows after further expansion that ${\bf t}_{1}+{\bf t}_{2}+{\bf t}_{3}=0$, thus proving the corollary.\hfill{Q.E.D}

{\flushleft{\bf Remark:}}
The transfer function $\bf P$ is Corollary \ref{scatt:C2} is essentially the transfer function at the dissipative (resistive) ports of the passive realization. This is only where all dissipation in the system takes place, the rest of the realization being fully lossless. This is further appreciated by observing that for any input vector $\bf u$, it follows from (\ref{scatt:113}) that for real frequencies (i.e., for any $z_i = \exp (j\omega_i)$, $\omega_i$ real, $i=1,2$) we have  $||{\bf u}||^2 - ||{\bf S}{\bf u}||^2 = ||{\bf P}{\bf u}||^2$, which is a manifestation of the fact that all energy absorbed by the system described by the transfer function $\bf S$ goes into the transfer function $\bf P$.  

We can now state the converse of the bounded real Lemma \ref{converse-BR} in the following form. The essential ingredients of the proof are drawn from \cite{rob:Basu3}.
\begin{lemma}\label{converse-BR}
{\bf [Converse form of 2-D BR Lemma \ref{scatt:L2}]}\\
Given a 2-D Roesser state space realization as in equations (\ref{scatt:104}) and (\ref{scatt:105}) such that (\ref{scatt:101}) is satisfied, with $\bf T$ being given as in (\ref{scatt:102}) and (\ref{scatt:103}),  the 2-D transfer function $\bf H$ given by equation (\ref{scatt:100}) is  a discrete bounded rational matrix.
\end{lemma}

{\em Proof:} Given the Roesser state space model  (\ref{scatt:104}),  (\ref{scatt:105}) with nonnegative $\bf T$ as in (\ref{scatt:101}) one can, by using the matrix 'square root'  operation, obtain $\bf L$ and $\bf W$ satisfying  (\ref{scatt:106}). Furthermore, by running the same algebraic manipulations as in the proof of Corollary \ref{scatt:C2} we may obtain equation (\ref{scatt:110}). Therefore, by considering $\bf z$ on the so called distinguished boundary of the unit bi-disc $|{\bf z}| =1$, i.e., $z_i = \exp (j\omega_i)$, $\omega_i$ real, $i=1,2$ we may conclude that $\one -{\bf H}^*{\bf H} = {\bf P}^*{\bf P}\geq {\bf 0}$.

In order to prove that $\bf H$ is a bounded rational matrix, it only remains to show that $\bf H$ is holomorphic in the open unit bi-disc $|{\bf z}|<1$. For contradiction, we assume the existence of a fixed zero ${\bf z}=(\zeta, \eta)$ of the denominator of (\ref{scatt:100}) in  $|{\bf z}|<1$, i.e., 
\[
\det(\zeta^{-1}\one_{p} \oplus  \eta^{-1}\one_{q} - {\bf D})=0, 
\quad 
|\zeta|<1 \text{ and } |\eta|<1.
\]
This, in turn, implies the existence of a nonzero constant vector, say, $\bf v$ such that $(\zeta^{-1}\one_{p} \oplus  \eta^{-1}\one_{q} - {\bf D}){\bf v}=0$, i.e.,  ${\bf D}{\bf v}=(\zeta^{-1}\one_{p} \oplus  \eta^{-1}\one_{q}){\bf  v}$. If  ${\bf v}_1$ and  ${\bf v}_2$ are respectively the first $p$ and last $q$ components of ${\bf v}$, then it is easy to see that
\be\label{holomorp-1}
{\bf v}^{*}{\bf v}  - {\bf v}^{*}{\bf D}^{*}{\bf D}{\bf v} =
(1 - |\zeta|^{-2})|{\bf v}_1|^2 + (1-|\eta|^{-2})|{\bf v}_2|^2.
\ee
We now consider (\ref{scatt:109}), which is an immediate algebraic consequence of  (\ref{scatt:101}). Pre-multiplying by ${\bf v}^{*}$ and post-multiplying  by ${\bf v}$  equation (\ref{scatt:109}) yields
\be\label{holomorp-2}
 {\bf v}^{*}{\bf v}  - {\bf v}^{*}{\bf D}^{*}{\bf D}{\bf v} ={\bf v}^{*}{\bf B}^{*}{\bf B}{\bf v} + {\bf v}^{*}{\bf W}^{*}{\bf W}{\bf v}\geq 0.
 \ee
Clearly, (\ref{holomorp-1}) and (\ref{holomorp-2}) are in contradiction due to that fact that $|\zeta|<1$  and  $|\eta|<1$,  proving that ${\bf H}$ is  holomorphic in the open unit bi-disc $|{\bf z}|<1$, and thus concluding the proof of the present lemma.\hfill{Q.E.D} 
{\flushleft{\bf Remark:}}
Lemma \ref{scatt:L2} and \ref{converse-BR}, as presented, provides characterization of passivity of a scattering matrix in terms of (Roesser) state space realization of a discrete domain 2-D system. A similar result can be obtained from it for continuous time systems essentially via the use of double bilinear transformation (\ref{S_fact:a*}). Furthermore, analogous characterization of the positivity property (e.g., as for the 1-D Positive Real Lemma), both for the continuous and the discrete case, can also be conveniently derived from it via the use of Cayley transform, namely, that rational matrix $\bf S$ is  a  bounded function if and only if the rational matrix  ${\bf Z}=(\one-{\bf S})(\one+{\bf S})^{-1}$ is a positive function.

 \section{Alternative approaches to synthesis}\label{Sec:7}
In order to round out the status of synthesis of multidimensional lossless transfer functions we next briefly describe two alternative approaches to the synthesis problem that have been pursued in the literature.  Admittedly, the discussion of this section does not rest upon the SOS problem - at least directly. While the approach described in previous sections involves extracting inductive or capacitive elements (or delays in the discrete domain formulation), and can thus be thought as a `global' approach to synthesis in an unconstrained structure (i.e., in an unconstrained topology in a graph theoretic sense), the two  approaches described in the present section either consists of attempts at decomposing the prescribed system into smaller subsystems of the same type, or alternatively at generating in a bottom up fashion a hopefully rich and interesting class of transfer functions that are synthesizable. 

\subsection{Multidimensional lossless two-ports }\label{other-synthesis}
This approach has only been considered for 2-ports described by the associated  transfer scattering matrix ${\bf T}$ or the scattering matrix ${\bf S}$ of the system. Factorization of ${\bf T}$  leads to decomposition in cascade structure in the synthesis, whereas other topological structures result from factoring ${\bf S}$. We will use the multidimensional extension of the  Belevitch canonical form parameterization  \cite{stb:Bele1} of a two-port given by a triplet of polynomials $\{f,g,h\}$ and a constant $\sigma$ such that
\[
{\bf S}=\frac{1}{g}\left[ \begin{matrix}
   h & \sigma {{f}_{*}}  \\
   f & -\sigma {{h}_{*}}  \\
\end{matrix} \right];
\quad 
{\bf T}=\frac{1}{f}\left[ \begin{matrix}
   \sigma {{g}_{*}} & h  \\
   \sigma {{h}_{*}} & g  \\
\end{matrix} \right],
\]
$g$ is a scattering Hurwitz polynomial, $g{{g}_{*}}=hh*+f{{f}_{*}}$, and $|\sigma |=1$. The zeros of the polynomial $f$ are called the transmission zeros of the system. We shall only consider the factorization ${\bf T}={\bf T}'{\bf T}''$, where $ {\bf T}'$ and ${\bf T}''$ are transfer scattering matrices of smaller `size'. The factorization of $\bf S$ being similar, but more straightforward. We refer to \cite{fb:basu+tan} and references therein for details including discrete versions of the problem. 

Let the desired ${\bf T}'$ and ${\bf T}''$ be parameterized  \`{a}  la Belevitch  respectively by $\{{f}',{g}',{h}'\}$, ${\sigma }'$ and  $\{{f}'',{g}'',{h}''\}$, ${\sigma }''$. Then, given the polynomial factorization $f={f}'{f}''$, unimodular constants ${\sigma }'$, ${\sigma }''$ and nonnegative integers ${{{n}'}_{i}}$ , ${{{n}''}_{i}}$ such that $\sigma ={\sigma }'{\sigma }''$, $\text{de}{{\text{g}}_{i}}{f}'\le {{{n}'}_{i}}$ ,  $\text{de}{{\text{g}}_{i}}{f}''\le {{{n}''}_{i}}$, $\text{de}{{\text{g}}_{i}}g={{n}_{i}}={{{n}'}_{i}}+{{{n}''}_{i}}$, we seek  polynomials ${g}'$, ${h}'$, ${g}''$, and ${h}''$. Note that the problem as described so far is highly nonlinear from the algebraic point of view. We will call this solution an {\em algebraic solution} to the problem. However, in order for  $\{{f}',{g}',{h}'\}$, ${\sigma }'$ and  $\{{f}'',{g}'',{h}''\}$, ${\sigma }''$ to be valid Belevitch canonical parameterizations of ${\bf T}'$ and ${\bf T}''$, we {\em additionally} need both ${g}'$ and ${g}''$ to be {\em scattering Hurwitz polynomials}. 

Two most unexpected results allows us to approach this apparently intractable problem.  
\begin{fact}\label{surprise-1}
An algebraic solution to the problem described above is such that the resulting ${g}'$ and ${g}''$ are necessarily scattering Hurwitz polynomials. 
\end{fact}
The unexpected nature of the result in Fact \ref{surprise-1} arises from the realization that a stringent analytic property of the solution, namely stability, is enforced by purely algebraic constraints of the solution. However, while it suffices to solve the algebraic problem only, its solution is still formidable due to apparent high nonlinearity. Towards this end, the equation 
\be\label{fundamental-eqn}
 h{g}'-g{h}'={\sigma }'{h}''{f}'{f}'_{*}, \; \deg_{i}{g}'\leq {n}',\; \deg_{i}{h}'\leq {n}'_{i},\;\deg_{i}{h}''\leq {n}''_{i}  
\ee
plays a fundamental role. Since the only unknowns  are the polynomials ${g}'$, ${h}'$, and ${h}''$ (\ref{fundamental-eqn}) is a linear equation.  Then we have the following most unexpected result.
\begin{fact}\label{surprise-2}
The   factorization problem ${\bf T}={\bf T}' {\bf T}''$, as described above, has a solution if and only if the fundamental equation (\ref{fundamental-eqn}), viewed as a linear simultaneous equation in terms of the coefficients of the polynomials  ${g}'$, ${h}'$, and ${h}''$, has \underline{two} linearly independent solutions.
\end{fact}
By counting the degrees of the polynomials in (\ref{fundamental-eqn}), it is straightforward to observe that there are ${{N}_{e}}$ equations involving ${{N}_{u}}$ unknown coefficients, where 
\[
{{N}_{e}}=\prod\limits_{i}{(2{{{{n}'}}_{i}}+{{{{n}''}}_{i}}+1)},\ {{N}_{u}}=2\prod\limits_{i}{({{{{n}'}}_{i}}+1)}+\prod\limits_{i}{({{{{n}''}}_{i}}+1)}.
\] 
In the univariate case it is trivial to see that ${{N}_{u}}-{{N}_{e}}=2$, and thus the factorization problem is solvable. However, in two or more dimensions we have ${{N}_{e}}>{{N}_{u}}$ in general, i.e., the system of fundamental equations  (\ref{fundamental-eqn}) is overdetermined, thus indicating the failure of the factorization  ${\bf T}={\bf T}' {\bf T}''$. A concrete numerical example of this phenomenon is not hard to construct and has been made available both for the continuous and for the discrete \cite{basu-wavelet} version  of the problem. Indeed, the system of linear equations becomes increasingly more overdetermined as the number of dimensions goes up.

{\flushleft{\bf Remark:}}
The problem of factoring the 2-port scattering matrix $S$ proceeds in a similar way. In this case, however, we begin at the very outset by considering the polynomial factorization $g={g}'{g}''$. The scattering Hurwitz properties of  ${g}'$  and ${g}''$ then trivially follow from the scattering Hurwitz property of $g$, thus making an analog of Fact \ref{surprise-1} unnecessary. However, an analog of Fact \ref{surprise-2} remains fully in force, and it again turns out that the factorization  ${\bf S}={\bf S}' {\bf S}''$ is feasible if and only if the linear simultaneous equations corresponding to the fundamental equation (\cite{basu-wavelet}) has two linearly independent solutions - a condition that is generically satisfied in 1-D but fails in higher dimensions. 

{\flushleft{\bf Remark:}}
The approach of this section could  conceivably be explored in systems with a larger number of ports than two. Unavailability of parameterizations of lossless $n$-ports akin to the Belevitch canonical form, makes this more clumsy but should not be a serious hindrance. 

 \subsection{Ratios of elementary symmetric functions}\label{scatt:S25}
 The synthesis problem seems to suffer from the following hiatus.
 Lossless functions are provably not synthesizable in three or higher
 dimensions.  Synthesis of bounded or positive (disspative) transfer functions suffers from the additional
 bottleneck that due to lack of spectral factorability type results, lossless
 embedding fails to hold.  This raises the question of characterization of the
 classes of synthesizable (lossless) bounded rational functions or matrices in
 three and higher dimensions.  While at present this problem is largely open,
 in this section we describe a class of $n$-D $(n \geq 3)$ functions that
 admits passive synthesis.  Although, it proves to be convenient to describe
 these functions as {\em immittances of continuous domain networks}, corresponding
 classes of discrete domain and/or scattering functions can be 
 identified by using the bilinear transformation either on the transform
 variables, or on the functions themselves, or on both.

 Such a subclass can be obtained from the well-known elementary
 symmetric functions.  Indeed, it can be shown \cite{fettweis-symm} that
\begin{fact}
If $s_{m,n}$ denotes an elementary symmetric function of (total) degree $m$ in $n$ variables then  $z_{m,n}=s_{m,n}/s_{m-1,n}$  is a reactance function in irreducible rational form that can be synthesized as a series-parallel network. 
\end{fact}
Simple  examples, for which series-parallel realizability is obvious, are 
\[
z_{2,2} = \frac{p_1 p_2}{p_1 +p_2},\quad z_{3,3}=\frac{p_{3}z_{22}}{p_{3}+z_{22}}=\frac{p_{1} p_{2} p_{3}}{p_{1} p_{2} + p_{2} p_{3} + p_{3}  p_{1}} \text{  etc.}
\]
While the functions thus  obtained are of first degree in each one of the variables, a richer class of synthesizable reactance functions of  higher degree can easily be generated by starting from elementary symmetric functions of a sufficiently larger number of variables and then setting some  of the variables identically equal. Furthermore, lossy (dissipative) impedance functions could be generated by setting some of the variables equal to constants having positive real part.

{\flushleft\bf Remark:}
The richness of the above  class of functions generated in this way, which are all synthesizable,  within the class of all multivariable positive real functions is unknown.

\section{Conclusions}\label{conclude}
We reviewed the status of multidimensional passive synthesis to date. We have seen that the synthesis both lossy and lossless  transfer functions critically hinges upon the of the classical SOS problem, and the (partial) solutions  available for it. In summary, lossless (minimal) synthesis is feasible in two-dimensions, but not in three or higher dimensions.  As for dissipative transfer functions, they can be embedded in a lossless system and subsequently synthesized, but again this works in two-dimensions at the most, and issues of detailed nature, e.g., minimality, number of dissipative elements  etc. are still unclear. In higher dimensions the embedding fails. A little known, or under appreciated fact attributed to Cassel reduction \cite{S_fact:land} that  a positive polynomial in two variables can be expressed as a sum of squares of rational functions, the denominators of which are univariate polynomials play, a critical role in our 2-D theory. It may also be noted that at the current state of our understanding,  2-D bounded real functions (i.e., bounded rational functions or matrices with real coefficients) may not necessarily  be realizable by using only real valued elements (capacitors/inductors), but complex valued elements may be required. The reason for the latter is that a procedure for obtaining real rational spectral factors of a 2-D parahermitian real rational matrix positive definite on the frequency axis, although desirable, is not presently known. The latter goal  may, however, be attainable via further sharpening of solutions to the SOS problem, which remains to be explored.

As reflected in the references cited, the problem considered here has a long history spanning approximately half a century, and various elements of the discussion presented have roots at various points of time. Relatively recently, however, an exposition with somewhat different flavor \cite{Joseph_Ball} became available from the community of operator theorists in mathematics. While we have not undertaken a detailed comparison of the techniques exploited in  that exposition, it is fair to remark that present paper is more aligned with perspectives of circuits and systems theory in the engineering community. Thus, classical network synthesis, as for example, originally developed in \cite{stb:Bele1} is more emphasized, and is our starting point. The present discussion focuses on rational matrices, which are crucial for network synthesis. We also present and discuss algebraic aspects of the Sum of Squares (SOS) problem as a critical element of the theory in much greater detail, while the treatment of \cite{Joseph_Ball} mentions it only briefly. In this vein, a little known result in \cite{S_fact:land}, that can be understood in most elementary terms, plays an important role in our discussions.  Also, relevant aspects of minimal synthesis, synthesis of first-order $3$-D allpass functions, and treatment of dissipative $2$-D synthesis are not undertaken in \cite{Joseph_Ball}, as are the discussions on alternate possibilities of synthesis in constrained structures discussed in Section \ref{Sec:7}.

\section*{Acknowledgement}
The author would like to thank Professor Anton Kummert of University of Wuppertal, Germany for many years of  extended discussions on relevant problems, and in particular, for his comments on the present manuscript.

\end{document}